\documentclass[12pt]{article}

\usepackage{amsfonts,amssymb,amsmath,amsthm,epsfig,euscript,bbm,amsthm}
\usepackage{algorithm}
\usepackage{algpseudocode}
\usepackage{amsthm}

\usepackage{tikz}
\usetikzlibrary{arrows}
\usepackage {graphicx}
\usepackage{enumerate}
\usepackage[margin=1in]{geometry}

\usepackage{subcaption}
\usepackage{caption}
\usepackage[authoryear,sort]{natbib}
\usepackage[titletoc,title]{appendix}
\usepackage[pdftex,colorlinks]{hyperref}
\usepackage{mathtools}
 \usepackage{xcolor}

\usepackage{graphicx}
\usepackage{subcaption}
 
\usepackage[font=small,labelfont=bf]{caption}
\usepackage{footnote}
\makesavenoteenv{tabular}
\makesavenoteenv{table}
 \definecolor{orange}{RGB}{230,170,120}
  \definecolor{green}{RGB}{120,200,120}
\usepackage{rotating,multirow,makecell,array}
\usepackage {tikz}
\usetikzlibrary {positioning}
\definecolor {processblue}{cmyk}{0.96,0,0,0}
\usepackage{colortbl}
\usepackage{amsthm}
\usepackage{threeparttable,subcaption}
 \usepackage{booktabs}

 \hypersetup{
 	colorlinks=true,
 	linkcolor=blue,
 	filecolor=blue,      
 	urlcolor=cyan,
 }

 \hypersetup{colorlinks,linkcolor={blue},citecolor={blue},urlcolor={red}} 
 \usepackage{geometry}                
 \usepackage{multirow}
\usepackage[section]{placeins}
 \usepackage{bbm}
 \usepackage{graphicx}
 \usepackage{amssymb}
 \usepackage{epstopdf}
 \usepackage{tikz} 
 \usepackage{amsmath} 
\usepackage[utf8]{inputenc}
\usepackage{algorithm}

 \theoremstyle{plain}
 \newtheorem{thm}{Theorem}[section]
 \newtheorem{lem}[thm]{Lemma}

 \theoremstyle{definition}

 \newtheorem{exmp}{Example}[section]
  \newtheorem{ass}{Assumption}
\usetikzlibrary{positioning, fit, shapes.geometric, arrows}
 \theoremstyle{definition}
 \newtheorem{rem}{Remark}
 
 \makeatletter
 \def\BState{\State\hskip-\ALG@thistlm}
 \makeatother
 
 \usepackage{authblk}

 \newcommand{\tr}{\text{tr}}
 \newcommand{\R}{\mathbb{R}}
 \renewcommand{\P}{\mathbb{P}}
 \newcommand{\diag}{\text{diag}}
 \newcommand{\E}{\mathbb{E}}
 \newcommand{\Var}{\mathrm{Var}}
 \newcommand{\lb}{\left(}
 \newcommand{\rb}{\right)}
\usepackage{graphicx}
\usepackage{capt-of}    
\usepackage{tabularx} 
 \renewcommand{\P}{\mathbb{P}}

 \newcommand{\one}{\mathbf{1}}

\def\spacingset#1{\renewcommand{\baselinestretch}%
{#1}\small\normalsize} \spacingset{1}

 \newcommand{\indep}{\perp \!\!\! \perp}
 \newcommand{\ea}{\mathrm{ea}}
 \newcommand{\ep}{\mathrm{ep}}

\definecolor{coquelicot}{rgb}{1.0, 0.22, 0.0}

\usepackage{algorithm,algpseudocode}

\newcounter{algsubstate}
\renewcommand{\thealgsubstate}{\alph{algsubstate}}
\newenvironment{algsubstates}
  {\setcounter{algsubstate}{0}%
   \renewcommand{\State}{%
     \stepcounter{algsubstate}%
     \Statex {\footnotesize\thealgsubstate:}\space}}
 
 \newcommand{\blind}{0}
 
 \usepackage{setspace}

\begin{document}
  
\if0\blind
{  
  \title{\bf Causal clustering: design of cluster experiments under network interference\footnote{We thank Toru Kitagawa, Michael Leung, Jesse Shapiro, Jann Spiess, Max Tabord-Meehan for helpful comments. All mistakes are our own.} }
\author{Davide Viviano\footnote{Department of Economics, Harvard University. Email: dviviano@fas.harvard.edu} $\quad$ Lihua Lei\footnote{Graduate School of Business, Stanford University. Email: lihualei@stanford.edu} $\quad$ Guido Imbens\footnote{Graduate School of Business and Department of Economics, Stanford University. Email: imbens@stanford.edu} \\ Brian Karrer\footnote{FAIR, Meta. Email: briankarrer@meta.com } $\quad$ Okke Schrijvers\footnote{Central Applied Science, Meta. Email: okke@meta.com.} $\quad$ Liang Shi\footnote{Central Applied Science, Meta. Email: liangshi@meta.com.}}
    \date{This version: \today \\ First version: October, 2023}
  \maketitle
} \fi

\if1\blind
{
  \bigskip
  \bigskip
  \bigskip
  \ 
 \\
 
 \ 
 \\

  \medskip
} \fi

  \begin{abstract}
\noindent

This paper studies the design of cluster experiments to estimate the global treatment effect in the presence of network spillovers. We provide a framework to choose the clustering that minimizes the worst-case mean-squared error of the estimated global average treatment effect. We show that optimal clustering solves a novel penalized min-cut optimization problem computed via off-the-shelf semi-definite programming algorithms. Our analysis also characterizes simple conditions to choose between \textit{any} two cluster designs, including choosing between a cluster or individual-level randomization. We illustrate the method's properties using unique network data from the universe of Facebook's users and existing data from a field experiment.

\end{abstract}

\noindent%
\vfill

\newpage 

\bigskip

\spacingset{1.7}


\section{Introduction} 

Many experiments in economics are conducted in environments where units interact. A cash-transfer program may affect untreated households through local prices or informal insurance; an information campaign may spread through social networks; an online experiment may affect users connected through a platform. In such settings, treatment assignment is not only a device for generating exogenous variation. It also determines which estimand is identified by the experiment and how much independent variation remains for inference.

A common design choice in such settings is a cluster experiment: units are partitioned into clusters, entire clusters are assigned to treatment or control, and effects are estimated by differences in mean outcomes (possibly with covariate adjustment). Cluster randomization avoids specifying a model of how outcomes depend on neighbors' assignments when the goal is to estimate the total effect of the treatment, but it forces researchers to choose a partition of the network and to take a stand on how spillovers may propagate. Examples include algorithmic partitions of an online social network or administrative units such as villages.
In applications, however, pre-specified clusters such as villages, schools or platform partitions may either fail to capture all relevant interactions  that generate spillovers when too small, or introduce large variance and underpowered studies when too coarse.  

Thus, the choice of clusters is a statistical design problem: the researcher must choose how much bias to tolerate through smaller clusters in order to gain precision. This paper asks how to choose both the structure and the number of clusters when the researcher observes a single network and wants to estimate the global average treatment effect, the target estimand in many applications \citep[e.g.][]{egger2019general, muralidharan2017experimentation}. We develop a decision-theoretic framework when the goal is to minimize the worst-case mean-squared error (MSE) of the difference-in-means estimator. 
Unlike clustering methods aimed at community detection, we choose clusters with the downstream causal objective in mind.

Our analysis proceeds as follows. We first present our main setup and focus on settings where spillover effects are local (within neighbors) and potentially of the same order as standard errors. We formalize this through a local asymptotic framework in which first-order spillover effects are local to zero, yet non-negligible for inference,\footnote{Higher-order interference can be accommodated by assuming it is of smaller order than first-order spillovers.} which is empirically relevant in many applications as we discuss further in Section \ref{sec:local} \citep[examples include][]{karrer2021network, lewis2015unfavorable, athey2023digital}. Within this framework, optimal clustering depends on the expected magnitude of the largest spillover effects. Calibrating this magnitude is an important practical input --- akin to specifying minimum detectable effects in power analysis \citep[e.g.][]{baird2018optimal}. This parameter can be calibrated using previous studies or obtained through a particular model (see  Section \ref{rem:psi} for practical recommendations).

The paper delivers three sets of results. First, we provide sharp worst-case characterizations of bias and variance for the difference-in-means estimator under a given clustering. The worst-case bias is governed by a simple measure of between-cluster connectedness: the average fraction of an individual's neighbors that lie outside her cluster, closely related to notions of cluster ``purity'' \citep[e.g.][]{brennan2022cluster, toulis2013estimation}. The variance component, by contrast, can be complicated because dependence arises both from interference and from correlated assignments induced by clustering. Our local asymptotic framework sheds light on the first-order effects of the variance: up to a small approximation error, the variance is proportional to the sum of squared cluster size. Together, these results yield a tractable expression for worst-case mean squared error that makes the bias--variance trade-off explicit.

Second, we use this characterization to derive a simple rule for choosing between a cluster design based on a given partition and individual-level randomization (a Bernoulli design) that depends on finite sample properties of the chosen partition and network.

Third, we turn to the core design problem of constructing MSE-optimal clustering. We show that choosing clusters reduces to a novel penalized min-cut problem in which the penalty internalizes both cross-cluster exposure (bias) and cluster-size dispersion (variance). We provide a convex-optimization algorithm that can be easily computed in practice with a provable certificate of its performance.

The empirical analysis illustrates two uses of the framework. In large online networks, the method can rank existing scalable clustering systems and quantify when cluster randomization is preferable to individual randomization. In field experiments, the method can evaluate whether administrative units are appropriate experimental clusters and can construct alternative partitions when social ties cross administrative boundaries. We illustrate both uses with unique de-identified network data from Facebook users and with village network data from a field experiment in rural China from \cite{cai2015social}.

Several extensions in Section \ref{sec:extensions} and Appendix \ref{app:more_extensions} illustrate the broader applicability of the method. These include cluster design for bias-aware inference, minimax regret choice with unknown spillover magnitude, higher order interference, among others.

\subsection{Related literature}

This paper connects to work on experimental design under interference, causal
inference with spillovers, and graph clustering. Most existing results on
experimental designs with spillovers---such as cluster and saturation designs
studied in \cite{taylor2018randomized,baird2018optimal,basse2016analyzing,
pouget2018dealing,viviano2020policy,karrer2021network,cai2022optimizing}---
take the clustering, or the class of clustering algorithms, as given and analyze
properties of the resulting estimators under alternative assignment schemes. In
contrast, our main object is the choice of the clustering itself. We derive a
decision-theoretic criterion from the bias and variance of the causal estimator,
and use this criterion to rank and construct clusterings. We provide details below. 

\paragraph{Clustering for community detection} There is a large literature on graph clustering and community detection,
including spectral clustering, normalized cuts, modularity-based methods, and
balanced graph partitioning
\citep[e.g.][]{shi2000normalized,von2007tutorial,newman2013spectral,
blondel2008fast,kabiljo2017social,lei2019unified,lei2020consistency,
li2022hierarchical}. These methods provide principled ways to partition a graph to recover communities. Here instead we choose clusters to minimize the
worst-case causal risk for estimating the global average treatment effect, motivating a different objective and analysis for our causal inference goal. 

\paragraph{Clustering for causal inference} The closest papers in spirit are those that use graph-based clusterings to
improve experimental designs under interference \citep[e.g.][]{ugander2013graph,eckles2017design,ugander2023randomized}, and in particular designs under interference that decays in network distance
\citep{leung2022rate,leung2023design,faridani2022rate}. The difficulty highlighted by \cite{leung2022rate}—that variance depends jointly on the network and on the clustering—helps explain why much of this literature either targets bias control alone \citep[e.g.][]{brennan2022cluster, taylor2018randomized} or works in settings where bias can be driven asymptotically to zero \citep{leung2022rate, faridani2022rate}. 
As a result, much of this literature focuses on rate-optimal cluster scales under assumptions that allow the bias to vanish asymptotically. Common (and useful) approaches are variations of $\varepsilon$-net clustering \citep{leung2022rate, faridani2022rate, ugander2013graph} consisting of randomly selects focal units and then selects the size of the clusters around these units using decay conditions on the size of the spillovers. 

The main distinction here is that instead the choice of the clusters is driven by the finite sample properties of the network through an explicit optimization problem. Instead of prescribing how cluster sizes should scale in terms of rates of convergence, we derive a finite-sample objective that trades off worst-case bias and variance
for a given estimator and estimand. This explicit characterization is made
possible by our local asymptotic framework; it allows us to characterize optimal clustering as a function of the underlying graph, where the trade-off between the bias and variance is explicit in the optimization problem (see e.g. Theorem \ref{thm:opt1}). 

\paragraph{Experimentation without clustering} Additional references on experimentation with and without networks that are not
primarily concerned about choosing optimal clusterings include \cite{basse2018model} and
\cite{jagadeesan2020designs} for direct rather than global effects,
\cite{kang2016peer} for encouragement designs, \cite{viviano2020experimental}
for using pilot information, and \cite{basse2018limitations} for limitations of
design-based inference under interference. Recent work on experimental design
for i.i.d. settings includes
\cite{cai2024performance,tabord2023stratification}. These contributions provide
important guidance for design and inference, but they do not address the problem
studied here of selecting a partition of a single observed network to optimize a
causal bias--variance criterion.

\paragraph{Causal inference under interference} Finally, our paper is related to the broader literature on identification and
inference under network interference---e.g.,
\cite{aronow2017estimating,hudgens2008toward,toulis2013estimation,
manski2013identification,athey2018exact,goldsmith2013social,
savje2017average,ogburn2017causal,manresa2013estimating,li2022random}---which
typically studies estimands, estimators, and assumptions rather than the
construction of an optimal clustering. We also connect to work on clustering in
economics \citep{wooldridge2003cluster,abadie2017should}. The key distinction
is that our clustering criterion is dictated by a downstream causal estimand and
a worst-case mean-squared-error objective, rather than by community detection.

\section{Setup} \label{sec:setup}

We consider a setting with $i \in \{1, \cdots, n\}$ units. Let $Y_i \in \mathbb{R}$ denote the observed outcome of interest for individual $i$, $D_i \in \{0,1\}$ denote a binary treatment assignment, and $\mathbf{D} \in \{0,1\}^n$ the vector of treatment assignments of each unit. Define $\mathbf{A}$ a symmetric adjacency matrix with $\mathbf{A}_{i,j} \in \{0,1\}$, and 
$
Y_i(\mathbf{d}), \mathbf{d} \in \{0,1\}^n
$ 
the potential outcome as a function of the entire vector of treatment assignments, with $Y_i = Y_i(\mathbf{D})$. We implicitly condition on $\mathbf{A}$, observed to the researchers, and (unobserved) potential outcomes $Y(\mathbf{d})$ (taking therefore a design-based perspective), unless otherwise 
specified (see Remark \ref{rem:unobserved_a} for an extension with unobserved $\mathbf{A}$). We define 
$\mathcal{N}_i = \{j: \mathbf{A}_{i,j} = 1\}$ the set of individuals connected to $i$. To avoid dividing by zero, we define $|\mathcal{N}_i| = \max\{1, \sum_j \mathbf{A}_{i,j}\}$. Let $\mathcal{N}_{n, \max} = \max_{i \in \{1, \cdots, n\}} |\mathcal{N}_i|$, the maximum degree.  
We focus on estimating the global average treatment effect (GATE), 
\begin{equation} \label{eqn:overall}
\small 
\begin{aligned} 
\tau_n = \frac{1}{n} \sum_{i=1}^n \Big[Y_i(\mathbf{1}) - Y_i(\mathbf{0})\Big]. 
\end{aligned}
\end{equation} 
The GATE defines the effect if all units had received the treatment compared to none of the individuals receiving the treatment. This is a target estimand in many applications when researchers are interested in implementing the policy at scale, e.g., \cite{egger2022general}, \cite{muralidharan2017experimentation}, \cite{karrer2021network} (see Appendix \ref{sec:saturations} for a discussion on other estimands). 

\vspace{-4mm}

\subsection{Potential outcomes and spillover effects} 
\label{sec:local}

Next, we impose restrictions on the spillover effects. For a set $S$ of indices and vector $\mathbf{x}$, we denote $\mathbf{x}_S$ the entries of $\mathbf{X}$ with indices in $S$. 

\begin{ass}[First-order local interference] \label{ass:first_order} For $i \in \{1, \cdots,n\}$, 
$
Y_i(\mathbf{d}) = \mu_i(\mathbf{d}_i, \mathbf{d}_{\mathcal{N}_i}), \forall \mathbf{d} \in \{0,1\}^n,
$
for some functions $\mu_i(1, \cdot) \in \mathcal{M}_{1,i}, \mu_i(0, \cdot) \in 
\mathcal{M}_{0, i}$ for some set of functions $\mathcal{M}_{1, i}, \mathcal{M}_{0, i}$.
\end{ass}

Assumption \ref{ass:first_order} states that spillovers occur between neighbors and allows for arbitrary dependence of potential outcomes with neighbors' assignments. One-degree neighborhood is consistent with models often used in applications, e.g., \cite{cai2015social}, \cite{sinclair2012detecting}, \cite{muralidharan2017experimentation}. \cite{athey2018exact} provide a framework for testing Assumption \ref{ass:first_order}.  Higher-order interference can be accommodated, although higher-order effects are often small and difficult to detect (see Remark \ref{rem:higher_order}).  With a slight abuse of notation, we occasionally write $\mu_i(\mathbf{d})$ for $\mu_i(\mathbf{d}_i, \mathbf{d}_{\mathcal{N}_i})$ for notational convenience.

We will refer to $\tau_{n, \mu}$ as the GATE in Equation \eqref{eqn:overall} to make the dependence of $\tau_n$ on $(\mu_1, \cdots, \mu_n)$ explicit. We do not assume that we know or can estimate consistently $\mu_i$. (The functions $\mu_i$ and their classes $\mathcal{M}_{0, i}, \mathcal{M}_{1, i}$ are indexed by $i$, and therefore they cannot be consistently estimated.) Instead, we allow for \textit{arbitrary} classes $\mathcal{M}_{1, i}, \mathcal{M}_{0, i}$ of potential outcome functions, as long as such classes satisfy the conditions below. These classes have to be sufficiently large to accommodate rich structures in the data, as well as satisfy some restrictions to be able to estimate features of causal effects.

\begin{ass}[Class of potential outcome functions] \label{ass:exposure_restriction} The potential outcomes $\mu = (\mu_1, \cdots, \mu_n) \in \mathcal{M}  = \otimes_{i=1}^n \mathcal{M}_i$ where $\mathcal{M}_i$ is the set such that $\mu_i(1, \cdot) \in \mathcal{M}_{1, i}, \mu_i(0, \cdot) \in \mathcal{M}_{0, i}$. For all $i \in \{1, \cdots, n\}$, $d \in \{0,1\}$, the function class $\mathcal{M}_{d,i}$ is the class of functions that include \textit{all} functions $\mu_i(d, \cdot)$ satisfying the following two conditions: 
\begin{itemize} 
\item[(i)] $\mu_i(d, \cdot)\in [\psi_L, \psi_R]$ for some constants $-\infty < \psi_L < \psi_R < \infty$ that do not vary with $n$; 
\item[(ii)] for all $\mathbf{d} \in \{0,1\}^{|\mathcal{N}_i|}$,  for some $\bar{\phi}_n \le \psi_R - \psi_L$ 
\begin{equation} \label{eqn:interference} 
\small 
\begin{aligned} 
\Big|\mu_i(0, \mathbf{d}) - \mu_i(0, \mathbf{0})\Big| & \le  \frac{\bar{\phi}_n}{|\mathcal{N}_i|} \sum_{k \in \mathcal{N}_i} \mathbf{d}_k, \quad 
  \Big|\mu_i(1, \mathbf{d}) - \mu_i(1, \mathbf{1})\Big| & \le  \frac{\bar{\phi}_n}{|\mathcal{N}_i|} \sum_{k \in \mathcal{N}_i} \Big(1 - \mathbf{d}_k\Big). 
 \end{aligned} 
\end{equation} 
\end{itemize} 
\end{ass}

Assumption \ref{ass:exposure_restriction} imposes two restrictions. Condition (i)  states that potential outcomes are uniformly bounded.\footnote{This restriction is common in the literature, e.g., \cite{kitagawa2020should}, and can be relaxed by assuming random sub-gaussian potential outcomes.
}  Condition (ii) is our main restriction on the exposure mapping. Condition (ii) is attained if $\mathcal{M}_{1, i}, \mathcal{M}_{0, i}$ are Lipschitz function classes in the share of treated neighbors. Condition (ii) states that potential outcomes vary with the share of neighbors' treatments by \textit{at most} $\bar{\phi}_n$. Specifically, the first condition (resp. second condition) states that as the number of treated friends increases, the difference from the pure control (resp. pure treatment) potential outcome increases (resp. decreases) in a worst-case sense linearly, with marginal effect bounded by $\bar{\phi}_n$. Here $\bar{\phi}_n$ captures the magnitude of (largest) spillovers. The condition that $\bar{\phi}_n \le \psi_R - \psi_L$ is a regularity assumption that guarantees that $(ii)$ is consistent with the boundedness restriction in $(i)$.  In Appendix \ref{sec:random_outcomes}, we relax Assumption \ref{ass:exposure_restriction}, by imposing $(ii)$ only on the expected value of the potential outcomes (instead of the realized potential outcomes). In Appendix \ref{sec:heterogeneity}, we relax Assumption \ref{ass:exposure_restriction} by allowing for heterogeneity in the range of potential outcomes and the spillover effects.

The component $\bar{\phi}_n$ quantifies the largest magnitude of spillover effects and will play an important role in our asymptotic analysis as $n \rightarrow \infty$.

\begin{ass}[Spillover local asymptotics] \label{defn:local} We consider a sequence of $\Big(\mathbf{A}, (Y_{i}(\cdot))_{i=1}^n\Big)$, indexed by $n$, where $\lim \sup_n \bar{\phi}_n \min\{n, \mathcal{N}_{n, \max}^2\} \le \varepsilon$, for some finite constant $\varepsilon \ge 0$. 
\end{ass}

Assumption \ref{defn:local} formalizes the local asymptotic framework considered here, where $\bar{\phi}_n  \mathcal{N}_{n, \max}^2 $ denotes the product of the largest spillover effects with the squared maximum degree. We assume that such a product is bounded in the limit by $\varepsilon \ge 0$. 
For $\varepsilon = 0$, spillover effects converge to zero asymptotically. This may occur if $\bar{\phi}_n \rightarrow 0$ at an arbitrary rate for networks with bounded degree \citep[e.g.][where $\mathcal{N}_{n, \max}$ is bounded]{de2018identifying}. In this case, the rate of convergence of $\bar{\phi}_n$ characterizes the magnitude of the spillover effects. As in a local asymptotic framework \citep[e.g.][]{hirano2009asymptotics}, rates of convergence of order equal or slower than the effective number of observations (allowed here) indicate non-negligible spillover effects even with $\varepsilon = o(1)$. Although we impose no restrictions on $\varepsilon$, we will think of $\varepsilon$ as small. Specifically, we provide a framework that characterizes approximately optimal experimental design, where the quality of the approximation depends on $\varepsilon$.   Also, note that GATE may also exhibit the same (or similar) asymptotic behavior, see Remark \ref{rem:direct_asymptotics}. 

Scenarios where $\bar{\phi}_n$ is local to zero formalize the idea that spillover effects (and possibly but not necessarily also overall treatment effects $\tau_n$) are small but possibly non-negligible for inference. 
For example, in contexts with online advertising \cite{lewis2015unfavorable} observe that a treatment effect of 0.005 of a standard deviation of the outcome would be highly profitable for the company given the cost of advertising \citep[see also][for a discussion]{athey2023semi}. Similar examples in digital advertisements with small effects are  
\cite{athey2023digital} with $i.i.d.$ data and \cite{karrer2021network} with network (Facebook) data, where the GATE are also of the order of 0.005 standard deviations. Motivated by these applications we therefore focus on settings where first-order effects are small but possibly non-negligible. 

As we show below, this framework allows us to provide explicit characterization of the bias and variance of the clustering, different from models of decaying interference as e.g., in \cite{leung2022causal}, where such characterization is not possible.

\begin{exmp}[Linear exogenous peer effects] \label{exmp:1} 
Consider a class of functions of the form
$\mu_i(\mathbf{d}) = \mu(T_i(\mathbf{d})) + \nu_i$, with $T_i(\mathbf{d}) = \Big[\mathbf{d}_i, (1 - \mathbf{d}_i) \times \frac{\sum_{j \neq i} \mathbf{A}_{i,j} \mathbf{d}_j}{\sum_{j \neq i} \mathbf{A}_{i,j}}, \mathbf{d}_i \times \frac{\sum_{j \neq i} \mathbf{A}_{i,j} \mathbf{d}_j}{\sum_{j \neq i} \mathbf{A}_{i,j}}\Big]$, and $\mu(t) = t^\top \beta$ for $\beta \in [-\bar{\Delta}_n,\bar{\Delta}_n] \times [-\bar{\phi}_n,\bar{\phi}_n]^2,$ for some arbitrary $\bar{\Delta}_n, \bar{\phi}_n$, and $\nu_i$ that is not a function of $\mathbf{d}$. Then Assumption \ref{ass:exposure_restriction} holds. \qed 
\end{exmp}

\begin{rem}[Restriction on the average squared degree instead of maximum degree] Assumption \ref{defn:local} can be sharpened by imposing restrictions on 
the \textit{average} second-order degree instead of maximum degree; see Equation \eqref{eq:sum_square_degree} in the proof of Lemma \ref{lem:rel}.  \qed  
\end{rem}

\begin{rem}[Direct and global effect] \label{rem:direct_asymptotics} Assumption \ref{defn:local} allows direct treatment effects (and GATEs) to be of the \textit{same} magnitude of spillover effects. Specifically, it is possible that $\tau_n$ is of the same order of $\bar{\phi}_n$. Therefore, our local asymptotic assumption does not require that spillover effects are local to the GATE (since $\tau_n$ can also converge to zero). Instead, our local asymptotics formalizes settings where the signal-to-noise ratio can be small (common in  experiments on online platforms,  \cite{karrer2021network}). \qed 
\end{rem} 

\begin{rem}[Squared bias and variance of the same order]
It is important to note that Assumption \ref{defn:local} allows the squared bias
and variance to be of the same order. For example, when the degree
is uniformly bounded \citep[common in economics, e.g.][]{jackson2012social,cai2015social,de2018identifying},  and \(K_n\asymp n\), 
the bias is of order \(\bar\phi_n=o(1)\), the squared bias is of order
\(\bar\phi_n^2\), and the variance is of order \(1/n\). Thus, squared bias and
variance are of the same order when \(\bar\phi_n\asymp n^{-1/2}\), while either
component may dominate depending on whether \(\bar\phi_n\) converges more slowly
or more quickly than \(n^{-1/2}\), so that Assumption \ref{defn:local} is not restrictive on the relative order of magnitude between the two. 

Similar reasoning applies with growing degree, provided that the maximum degree does not grow too fast in the sample size.\footnote{For example, suppose for small $\eta > 0$
\(\mathcal N_{n,\max}\asymp n^{\kappa - \eta}\), \(K_n\asymp n^\rho\), and
\(\bar\phi_n\asymp n^{-a}\), with \(\kappa\le 1/2\). Then
Assumption
\ref{defn:local} holds with \(a = 2\kappa\). In this case, the squared bias is of the same or larger order
than variance if \(4\kappa \le \rho\), i.e., provided that the network is sufficiently sparse.} 
\qed
\end{rem}

\begin{rem}[Inference on treatment effects] 
Here, we study optimal designs under Assumption \ref{defn:local}. However,  once the experiment is conducted based on a given clustering, inference on the GATEs does not require Assumption \ref{defn:local} to hold, provided that suitable conditions on the covariance structure and sparsity on the maximum degree hold. See Theorem \ref{thm:inference1} and Appendix \ref{sec:estimated_variance2}. \qed   
\end{rem} 

\begin{rem}[Higher order interference and endogenous peer effects] \label{rem:higher_order} 
Our setting generalizes to higher-order interference in two scenarios. 

First, suppose that friends up to degree $d < \infty$ generate spillovers in magnitude similar to first-degree friends. Our results extend after we define the set of friends as the set of friends up to degree $d$, and $\bar{\phi}_n$ the largest effect that such friends generate. Sparsity restrictions on the largest degree in Assumption \ref{defn:local} are with respect to the number of friends up to degree $d$. 

In this case the same algorithm follows through, after appropriately redefining the set of friends $\mathcal{N}_i$ to include higher degree friends.

In the second scenario, suppose that the assumption of first-order effects approximates higher-order effects up to a term of smaller order than first-order effects. In this case our results and algorithm follow similarly as we discuss in Appendix \ref{example:endogenous}, since this additional remainder is negligible relative to first order effects.  \qed 
\end{rem} 


\vspace{-4mm} 

\subsection{Experimental design and estimation}

Next, we turn to the class of designs and estimators considered here. Define a clustering of size $K_n$ as a collection of subsets of unit indices satisfying
$$
\small 
\begin{aligned} 
\mathcal{C}_{ n} = \left\{c_k \subseteq \{1, \cdots, n\}, k \in \{1, \cdots, K_n\}, \bigcup_k c_k = \{1, \cdots, n\}, c_k \bigcap c_{k'} = \emptyset \text{ for } k \neq k' \right\}.
\end{aligned} 
$$
Here, $\mathcal{C}_{ n}$ denotes
a particular partition of the units in the population with $K_n$ exclusive clusters. With an abuse of notation, let $c(i) \subseteq \{1, \cdots, n\}$ denote the cluster assigned to unit $i$, and $|c_k| = n_k$ the number of individuals in cluster $k$. 
For ease of exposition, we focus our discussion and assumptions below in the presence of a given clustering $\mathcal{C}_{ n}$ and return to choosing the optimal clustering in Section \ref{sec:opt}.

\begin{ass}[Cluster designs] \label{ass:clusters} Let $\mathbf{A}$ and $\{Y_i(\mathbf{d})\}_{i \in \{1, \cdots, n\} \mathbf{d} \in \{0,1\}^n}$ be fixed (non-random). Assume that $\mathcal{C}_{ n}$ is a deterministic function of $\mathbf{A}$ and is such that the number of units in cluster $k$ is $n_k = \gamma_k \frac{n}{K_n}$ with $\max_k \gamma_k \le \bar{\gamma} < \infty$, for some $\bar{\gamma} < \infty$. Assume that $D_i = \tilde{D}_{c(i)}$ almost surely, for independent random variables $(\tilde{D}_1, \cdots, \tilde{D}_{K_n})$ where for $k \in \{1, \cdots, K_n\}$, $\tilde{D}_{k} \sim \mathrm{Bern}(0.5)$.  
\end{ass}

Assumption \ref{ass:clusters} states that the clustering is constructed using information from the adjacency matrix only (i.e., it is independent of potential outcomes), clusters are proportional in size\footnote{This restriction is sufficient but not necessary for our analysis, and it is imposed for expositional convenience. It is possible to relax such an assumption by assuming that $\sum_{k = 1}^{K_n} \gamma_k^2/K_n = \mathcal{O}(1)$.}, and individuals in a given cluster are all assigned either treatment or control with equal probability. Assumption \ref{ass:clusters} restricts the class of designs to cluster designs, motivated by our focus on overall treatment effects and empirical practice.

Motivated by standard practice both in industry and in field experiments, we consider estimators obtained by simple differences in means between treated and control clusters \citep[which are those used in most of applications on online platforms,][]{karrer2021network}; see Remark \ref{rem:estimator} for comparison with other estimators studied in the literature.  Because $P(D_i = 1) = 1/2$, we construct a (biased) estimator of treatment effects as\footnote{Here, for expositional convenience, we make explicit the dependence of $\hat{\tau}_n$ with $\mathcal{C}_n$, where the clustering $\mathcal{C}_n$ only characterizes the distribution of treatment assignments $D_i$ but not the structure of the estimator.} 
\begin{equation} \label{eqn:estimatormain}
\small 
\begin{aligned} 
\hat{\tau}_n(\mathcal{C}_{ n}) = \frac{2}{n} \sum_{i=1}^n \Big[D_i Y_i - (1 - D_i) Y_i  \Big]  = \frac{2}{n}\sum_{i=1}^{n}(2D_i - 1)Y_i. 
\end{aligned} 
\end{equation} 
The estimator $\hat{\tau}_n(\mathcal{C}_{ n})$ takes the difference between the treated units and control units' outcomes and appropriately rescales this difference by the probability of treatment. 
Therefore, $\hat{\tau}_n(\mathcal{C}_{ n})$ depends on the clustering $\mathcal{C}_{ n}$ \textit{only} because the distribution of the treatments depends on the clusters under Assumption \ref{ass:clusters}. Studying the estimator in Equation \eqref{eqn:estimatormain} is a natural starting point for the analysis of cluster experiments. Variants of difference in means estimators (possibly also with regression adjustments discussed in Remark \ref{rem:reg_adj}) are often used or studied in practice \citep{karrer2021network, holtz2020reducing, savje2017average}. One could also normalize each sum in Equation \eqref{eqn:estimatormain} by the number of treated and control units (instead of using knowledge about the treatment probability). This would improve the stability of the estimator but complicate the analysis when the number of treated units is stochastic (e.g., when clusters have different sizes).

\begin{rem}[Alternative estimators]  \label{rem:estimator}
Alternative estimators studied in the literature are inverse probability weights estimators \citep[e.g.,][]{aronow2017estimating, ugander2013graph}. These estimators reweight by the probability that the assignment of \textit{all} friends being either zero or one. Unless researchers impose additional restrictions on the exposure mapping, in the network context, these estimators are highly sensitive to instability of the propensity score because the propensity score scales to zero \textit{exponentially} in the maximum degree.\footnote{Even if the maximum degree grows at a slow rate, such as $\mathcal{N}_{n, \max} \propto \log(n), \bar{\phi}_n \propto 1/\log^2(n)$, so that effects are local to zero but non negligible for inference the propensity score converges to zero at rate $1/n$.} As a result, the IPW estimator is particularly sensitive to the choice of the design, as in contexts where neighbors are assigned to different clusters, the propensity score converge to zero fast in the number of friends.  
An alternative is model-based estimators, which are subject to model misspecification. A third class of estimators are trimming estimators, which trim observations based on their position in the network and also optimize over the choice of the estimators and designs. We study this class of estimators and designs in Appendix \ref{sec:other_estimators}. 
\qed   
\end{rem}

\begin{rem}[Covariate adjustment] \label{rem:reg_adj} Our framework generalizes to settings that use covariate adjustment. Denote $\bar{\mu}_i$ an arbitrary predictor for $\mu_i(\mathbf{0})$ that only uses information from some arbitrary baseline observable characteristics (i.e., does not depend on the treatments or end-line outcomes in the experiment).  The estimator with such an adjustment takes the form 
\begin{equation} \label{eqn:reg_adj}
\small 
\begin{aligned} 
\frac{2}{n} \sum_{i=1}^n (Y_i - \bar{\mu}_i)(2 D_i - 1).
\end{aligned} 
\end{equation}
Note that $Y_1 
- \bar{\mu}_1, \ldots, Y_n - \bar{\mu}_n$ typically have different ranges ex ante even under Assumption \ref{ass:exposure_restriction}. To account for this, we can apply the more general result in Appendix \ref{sec:heterogeneity}. 
\qed 
\end{rem}

\begin{rem}[Saturation designs] It is possible to consider designs where treatments are assigned with different cluster-level probabilities. Our main insights for the bias and variance characterization follow similarly to what we present for cluster experiments,  with appropriate modifications presented in Appendix \ref{sec:saturations}. \qed 
\end{rem}

\vspace{-5mm}

\section{Bias, Variance and MSE of a given clustering } \label{sec:when}

As a necessary first step, this section focuses on studying the bias, variance, and mean-squared error of a \textit{given} clustering $\mathcal{C}_{ n}$, deferring the question of choosing an optimal clustering to the next section.  Denote $\mathbb{E}_\mu[\cdot]$ the expectation for given potential outcome functions $\mu$ (conditional on $\mathbf{A}$). The components $\mathcal{O}(\cdot), o(\cdot)$ in the following lemmas and theorems hold uniformly over all cluster designs with $K_n$ many clusters that satisfy Assumption \ref{ass:clusters}. The notation for $\mathcal{O}(\cdot), o(\cdot)$ may depend on finite constants defined in Section \ref{sec:setup}, including $\bar{\psi}, \bar{\gamma}$.

Our objective reads as follows 
\begin{equation} \label{eqn:worst-case} 
\small 
\begin{aligned} 
\mathcal{B}_n(\mathcal{C}_{ n}, \lambda) = \sup_{\mu \in \mathcal{M}} \left\{\mathbb{E}_\mu\Big[\Big(\hat{\tau}_n(\mathcal{C}_{ n}) - \mathbb{E}_\mu[\hat{\tau}_n(\mathcal{C}_{ n})]\Big)^2\Big] + \lambda \Big(\tau_{n,\mu} - \mathbb{E}_\mu[\hat{\tau}_n(\mathcal{C}_{ n})]\Big)^2 \right\},  
\end{aligned} 
\end{equation}  
where the supremum is intended over $(\mu_i)_{i=1}^n$, and $\mathcal{M}$ is the product space of potential outcome functions as in Assumption \ref{ass:exposure_restriction}. The user may specify $\lambda$.

For $\lambda = 1$, $\mathcal{B}_n(\mathcal{C}_n, \lambda)$ is the worst-case mean-squared error. 
It is also possible to minimize the \textit{length of confidence intervals}, see Section \ref{sec:bias_aware_inferece}. 


\vspace{-3mm} 

\subsection{Worst-case bias, variance, and MSE}

Let
\begin{equation}\label{eq:bn}
\small 
\begin{aligned} 
b_n(\mathcal{C}_n) = \frac{1}{n}\sum_{i=1}^n\frac{1}{|\mathcal{N}_i|} \Big|\mathcal{N}_i \bigcap \Big\{j: c(i) \neq c(j)\Big\} \Big|
\end{aligned} 
\end{equation}
and $\Big\{j: c(i) \neq c(j)\Big\}$ denotes the set of units $j$ in a different cluster from unit $i$.

\begin{lem}[Worst-case bias] \label{lem:worst_case_bias}  Let Assumptions \ref{ass:first_order}, \ref{ass:exposure_restriction},  \ref{ass:clusters} hold. Then
$
\sup_{\mu \in \mathcal{M} } \Big|\tau_{n, \mu} - \mathbb{E}_\mu[\hat{\tau}_n(\mathcal{C}_{ n})]\Big| = \bar{\phi}_n b_n(\mathcal{C}_n).
$ 
\end{lem}

See Appendix \ref{sec:proof_bias} for the proof. 
The worst-case bias can be expressed as the number of friends of $i$ in a different cluster from $i$, appropriately reweighted by the overall number of friends of $i$. The size of the worst-case bias also depends on the magnitude of spillover effects $\bar{\phi}_n$. Lemma \ref{lem:worst_case_bias} provides a formal justification to notions of between clusters connectedness based on the worst-case bias of the treatment effect estimator. In the proof, the worst-case bias is attained as the condition in Equation \eqref{eqn:interference} holds with equality (Appendix \ref{sec:heterogeneity} studies settings where effects are heterogeneous across individuals). The expression for the bias connects to measures of cluster purity such as those in \cite{goel2009respondent}, \cite{toulis2013estimation}, motivating clustering purity based on the class of models in Assumption \ref{ass:exposure_restriction}.

To derive the worst-case variance, we define 
\begin{equation}\label{eq:psibar}
\small 
\begin{aligned} 
\bar{\psi}^{1/2} =  2\max\{|\psi_R|, |\psi_L|\} .
\end{aligned} 
\end{equation}
By Assumption \ref{ass:exposure_restriction} (i), $\bar{\psi}$ can also be expressed as $\max_{\mu_i\in \mathcal{M}_i}(\mu_i(\mathbf{1}) + \mu_i(\mathbf{0}))^2$ for every $i$.

\begin{lem}[Worst-case variance] \label{prop:main} 
Let Assumptions \ref{ass:first_order}, \ref{ass:exposure_restriction}, \ref{defn:local}, \ref{ass:clusters} hold. Then as $K_n \rightarrow \infty$
$$
\small 
\begin{aligned} 
&  \sup_{\mu \in \mathcal{M}} \mathbb{E}_\mu\Big[\Big(\hat{\tau}_n(\mathcal{C}_{ n}) - \mathbb{E}[\hat{\tau}_n(\mathcal{C}_{ n})]\Big)^2\Big]  = \frac{1}{n^2} \sum_{k=1}^{K_n} n_k^2 \times  \Big[\bar{\psi} + \mathcal{O}(\varepsilon) \Big]. 
\end{aligned} 
$$
\end{lem}

See Appendix \ref{app:propmain} for the proof. 
Lemma \ref{prop:main} characterizes the worst-case variance as the number of clusters $K_n$ grows (at an arbitrary rate). Such bounds depend on the variation in cluster sizes, captured through
$
\frac{1}{n^2} \sum_{k=1}^{K_n} n_k^2, 
$
times a positive constant $\bar{\psi}$. 
By leveraging the local asymptotic framework in Assumption \ref{defn:local}, a key insight is that the variance is (asymptotically) driven by the within-cluster correlations instead of cross-cluster connections, up-to a $\varepsilon$-approximation.  
To our knowledge, this characterization of the worst-case variance is novel to the literature.

\begin{thm} \label{thm:Bn_Bnstar} Suppose that Assumptions \ref{ass:first_order}, \ref{ass:exposure_restriction}, \ref{defn:local}, \ref{ass:clusters} hold.  Then as $K_n \rightarrow \infty$
$
\mathcal{B}_n(\mathcal{C}_{ n}, \lambda) = \mathcal{B}_n^*(\mathcal{C}_{ n}, \lambda)(1 + \mathcal{O}(\varepsilon)),
$
where
\begin{equation}\label{eqn:bound2}
\small 
\begin{aligned} 
\mathcal{B}_n^*(\mathcal{C}_{ n}, \lambda) = \bar{\psi}\cdot \frac{1}{n^2} \sum_{k=1}^{K_n} n_k^2  + (\lambda \bar{\phi}_n^2) \cdot b_n^2(\mathcal{C}_n).
\end{aligned} 
\end{equation} 
\end{thm}

See Appendix \ref{app:worst_case_mse} for the proof.
Theorem \ref{thm:Bn_Bnstar} leverages the lemmas above for the upper bound, and, in addition, it constructs an admissible data-generating process so that the upper bound is attained for the worst-case MSE to show sharpness.  Theorem \ref{thm:Bn_Bnstar} formalizes the intuition that as the number of clusters increase, and clusters have approximately the same size, the variance decreases as the correlation between treatment decreases; on the other hand, the bias increases since more units may have friends in different clusters.

\begin{rem}[Choice and sharper bound for $\bar{\psi}$] Because here we consider a finite population with fixed potential outcomes, the worst-case variance depends on $\bar{\psi}$, a constant proportional to the square of the largest value of the realized outcomes. In Appendix \ref{sec:random_outcomes}, we provide similar calculations in settings where potential outcomes are random functions, in which case we can interpret $\bar{\psi}$ as a sharper bound on the \textit{moments} of potential outcomes.     \qed 
\end{rem} 
\begin{rem}[Unobserved $\mathbf{A}$] \label{rem:unobserved_a} Suppose that $\mathbf{A}$ is unobserved or \textit{partially} observed, and researchers have a prior over $\mathbf{A}$. In this case, the worst-case variance does not change, the worst-case bias is replaced by $\bar{\phi}_n\E_{\mathbf{A}}[b_n^2(\mathcal{C}_n)]$, 
where the expectation may depend on partial network information \citep[][]{breza2020using}. 
\qed 
\end{rem}

\vspace{-3mm} 

\subsection{Comparison between a given clustering and a Bernoulli design}

Before discussing the optimal choice of the cluster, we explore a simpler question, that is central in many randomized trials with spillovers. Suppose we are given a clustering $\mathcal{C}_n$, but we have the option to run a Bernoulli design where treatments are assigned independently between individuals. When should we prefer the (more complex) cluster design over a Bernoulli design? Denote a Bernoulli design as $\mathcal{C}_{\mathrm{B},n} = \{c_k = \{k\}, k \in \{1, \cdots, n\}\}$. 

\begin{thm}[Rule of thumb] \label{thm:thumb} Suppose that Assumptions \ref{ass:first_order},  \ref{ass:exposure_restriction}, \ref{ass:clusters} hold.  Let $\underline{\gamma} = \frac{1}{K_n} \sum_{k=1}^{K_n} \gamma_k^2$, where $\gamma_k = n_k K_n/n$. Suppose that each individual has at least one connection. Then  
$
\frac{\mathcal{B}_n^{*}(\mathcal{C}_{\mathrm{B},n}, \lambda)}{\mathcal{B}_n^{*}(\mathcal{C}_{ n}, \lambda)} \ge 1 \quad  \text{ if }   \quad  (\lambda \bar{\phi}_n^2/\bar{\psi})^{-1} \le \frac{K_n(1 - 
b_n(\mathcal{C}_n)^2) }{\underline{\gamma} }. 
$
\end{thm}

See Appendix \ref{app:thumb} for the proof. 
Theorem \ref{thm:thumb} provides an explicit rule for choosing a cluster design with given clusters $\mathcal{C}_{ n}$ over a Bernoulli design $\mathcal{C}_{\mathrm{B},n}$, using the approximate objective $\mathcal{B}_n^*(\mathcal{C}_n, \lambda)$. The right-hand side depends on three \textit{observables}: (i) $\underline{\gamma}$, (ii) the expected bias of the clustering method (as a function of $\mathbf{A}$), and (iii) the number of clusters $K_n$. The left-hand side depends on $\bar{\phi}_n^2/\bar{\psi}$. 
For $\lambda = 1$, known $\bar{\psi}$, the rule of thumb provides the smallest spillover effects that would guarantee that the cluster design dominates the Bernoulli design. The figure on the right-hand side of Table \ref{fig:min_size} illustrates the exact threshold for the spillovers that justifies a cluster experiment. 
In our empirical application in Section \ref{sec:empirics} using data  from the universe of users at Facebook we find $\bar{\phi}_n \ge 0.003$ motivates a cluster experiment for binary outcomes.

\begin{table}[!ht]
\caption{Minimum spillover effects' size that justifies running a cluster experiment instead of a Bernoulli design in the presence of equally sized clusters $(\gamma_k = 1$ for all $k \in \{1, \cdots, K_n\})$, and where the clustering $\mathcal{C}_n$ is given. The bias denotes  $\frac{1}{n} \sum_{i=1}^n \frac{1}{|\mathcal{N}_i|} |j \in \mathcal{N}_i: c(i) \neq c(j)| \in [0,1]$. The table on the left-hand side reports the value of $\bar{\phi}_n$ for different values of the bias and number of clusters, fixing $\bar{\psi} = 4$ (e.g., valid when outcomes are bounded between $[-1,1]$). In the figure on the right-hand side, the y-axis reports in log-scale the minimum size of the spillover effects $\bar{\phi}_n/\sqrt{\bar{\psi}}$ to run a cluster experiment as in Theorem \ref{thm:thumb}. The x-axis reports the number of clusters. } \label{fig:min_size}
    \begin{tabularx}{\linewidth}{*{2}{>{\centering\arraybackslash}X}}
    \begin{tabular}[b]{cccc}
     Bias & $K_n$ & $\bar{\psi}$ & $\bar{\phi}_n$ \\
    \hline
      0.25 & 100 & 4 & 0.21 \\ 
      0.5 & 100 & 4 & 0.23 \\ 
      0.75 & 100 & 4 & 0.3 \\ 
      0.25 & 200 & 4 & 0.15 \\ 
      0.5 & 200 & 4 & 0.16 \\ 
      0.75 & 200 & 4 & 0.21 \\ 
      0.25 & 500 & 4 & 0.09 \\ 
      0.5 & 500 & 4 & 0.10 \\ 
      0.75 & 500 & 4 & 0.14 \\
        \hline
    \end{tabular}
    &
  \includegraphics[scale = 0.3]{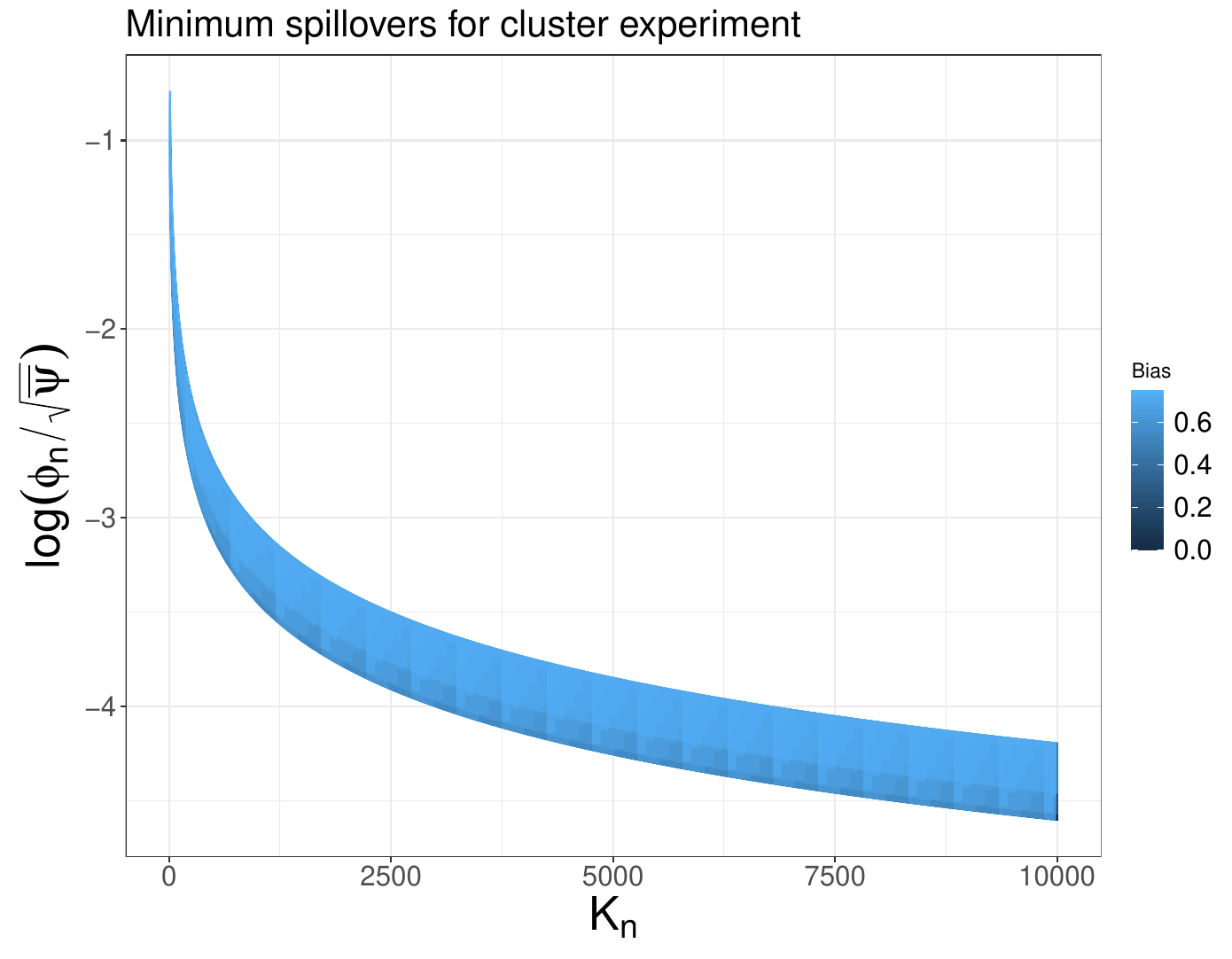}    \\
    &
    \end{tabularx}
\end{table}

\vspace{-15mm} 

\section{Optimization over the choice of the clustering}
\label{sec:opt}

In this section, we turn to designing the optimal clustering. 
Theorem \ref{thm:Bn_Bnstar} shows that, up to the approximation error induced by the local asymptotic framework, the relevant design objective is
$ 
\mathcal{B}_n^*(\mathcal C_n,\lambda)
=
\bar\psi\frac{1}{n^2}\sum_{k=1}^{K_n}n_k^2
+
\lambda\bar\phi_n^2 b_n(\mathcal C_n)^2 .
$ 
Since multiplying the objective by the positive constant \(\bar\psi\) does not affect the minimizer, the optimization problem is equivalent to minimizing
\begin{equation} \label{eqn:optimization_v10}
R_n(\mathcal C_n;\xi_n)
=
\frac{\xi_n }{n^2}\sum_{k=1}^{K_n}n_k^2
+
b_n(\mathcal C_n)^2 .
\end{equation} 
where $\xi_n = (\lambda \bar{\phi}_n^2/\bar{\psi})^{-1}$. 
The first component captures the worst-case variance contribution, and the second component captures the squared worst-case bias contribution.

Equation \eqref{eqn:optimization_v10} solves a weighted min-cut problem with an additional \textit{penalization} term that depends on the variance of the clusters' size. The optimization problem in Equation \eqref{eqn:optimization_v10} differs from the minimum normalized-cut problem \citep[e.g.,][]{shi2000normalized, ling2020certifying}, or related graph clustering problems \citep[see][and references therein]{pokhilko2019d, parker2017optimal, zhang2023asymptotics} whose objective only depends on measures of cluster purity such as  
$
\sum_{k=1}^K \sum_{i: c(i) = k}|j \in \mathcal{N}_i: c(i)\neq c(j) |/ \sum_{i: c(i) = k}|\mathcal{N}_i|.    
$
Different from our proposal, 
standard min-cut problems use a different denominator (since the objective is not motivated by the bias of treatment effect estimators), and, most importantly, do not account for the additional variance component.
Here instead, researchers should balance the bias \textit{and variance} for the clusters. 
To our knowledge,  the problem in Equation \eqref{eqn:optimization_v10} has not been studied before.

Let $\mathbf{V} =  \diag(|\mathcal{N}_{1}(1)|, \ldots, |\mathcal{N}_{n}(1)|)$. Define  
\begin{equation}\label{eq:laplacian}
\small 
\begin{aligned} 
\mathbf{L} = \mathbf{V}^{-1}\mathbf{A}.
\end{aligned} 
\end{equation}
Given a set of $K$ clusters (i.e., fixing the number of clusters), for any cluster mapping $c: \{1, \cdots, n\} \mapsto \{1, \cdots, K\}$, let $\mathbf{M}_c(K)\in \R^{n\times K}$ with
$$
\small 
\begin{aligned} 
\mathbf{M}_{c, ik}(K) = 1\{c(i)=k\}, \quad i\in \{1, \cdots, n\}, \quad  k\in \{1, \cdots, K\},
\end{aligned} 
$$ 
where $\mathbf{M}_c(K)$'s dimension depends on $K$. Here, $\mathbf{M}_{c, ik}(K) = 1$ if unit $i$ is in cluster $k$. 
The following result gives an exact reformulation of the design problem through a trace-optimization program with a semi-definite (convex) constraint for a fixed number of clusters.

\begin{thm}[Trace-optimization program]
\label{thm:opt1}
Fix \(K\ge 2\), and let \(\mathfrak C_K\) denote the set of clusterings with \(K\) nonempty clusters with each cluster $k$ having size $n_k \le \bar{\gamma} n/K$ for given $\bar{\gamma}$ in Assumption \ref{ass:clusters}.  
A clustering \(c^\star\in\mathfrak C_K\) minimizes
\(\mathcal B_n^*(\mathcal C_n,\lambda)\) if and only if
\[
\begin{aligned}
\min_{c \in \mathfrak C_K,z,t}\quad
&  \frac{\xi_n }{n^2} \mathrm{Trace}\left( \mathbf{1} \mathbf{1}_n^\top \mathbf{M}_c(K) \mathbf{M}_c(K)^\top \right)+ t, \\ 
\text{s.t.}\quad
& z= \frac{1}{n} \mathrm{Trace}\left(\mathbf{L}^\top (\mathbf{1}_n \mathbf{1}_n^\top - \mathbf{M}_c(K) \mathbf{M}_c(K)^\top \right) , \qquad 
\begin{pmatrix}
t & z\\
z & 1
\end{pmatrix}
\succeq 0 .
\end{aligned}
\]
\end{thm}

The proof is in Appendix \ref{proof:thm:opt1}.

Theorem \ref{thm:opt1} is an exact reformulation, but it is still combinatorial because $\mathbf{M}_c(K)$ is a binary matrix. We therefore replace \(\mathbf{M}_c(K)\mathbf{M}_c(K)^\top\) with a convex relaxation, a common approach in the literature for clustering \citep[e.g.][]{hong2021optimal}. 
Let $\mathbf{X}_c(K) = \mathbf{M}_c(K) \mathbf{M}_c(K)^\top$. 
Every exact
\(K\)-cluster matrix \(\mathbf X_c(K)=\mathbf M_c(K)\mathbf M_c(K)^\top\)
satisfies
$ 
\mathbf X_c(K)\succeq0,
X_{c,ii}(K)=1,
0\le X_{c,ij}(K)\le1.
$ 
In addition,
$ 
\frac{n^2}{K}
\le
\operatorname{tr}\left(
\one_n\one_n^\top\mathbf X_c(K)
\right)
=
\sum_{k=1}^K n_k^2
\le
\bar{\gamma} \frac{n^2}{K}.
$ 
The lower bound follows from Cauchy--Schwarz. The upper bound avoids unbalanced clusters (Assumption \ref{ass:clusters}).  This motivates the following SDP relaxation:
\begin{equation}
\label{eq:SDP}
\small
\begin{aligned}
\underline R_{n,K}^{\mathrm{sdp}}(\xi_n)
=
\min_{\mathbf X,z,t}\quad
&
\frac{\xi_n }{n^2}
\operatorname{tr}\left(
\one_n\one_n^\top \mathbf X
\right)
+
 t
\\
\text{s.t.}\quad
&
\mathbf X\succeq0,
\qquad
\mathbf{X}_{ii}=1,\quad i=1,\ldots,n,
\\
&
z=
\frac1n
\operatorname{tr}\left(
\mathbf L^\top
\left[
\one_n\one_n^\top-\mathbf X
\right]
\right), \qquad \begin{pmatrix}
t & z\\
z & 1
\end{pmatrix}
\succeq0
\\
& \mathrm{Optional: } \quad   \frac{n^2}{K}
\le
\operatorname{tr}\left(
\one_n\one_n^\top \mathbf X
\right)
\le
\bar{\gamma} \frac{n^2}{K}, \quad 0\le \mathbf{X}_{ij}\le1,
\quad i,j=1,\ldots,n,.
\end{aligned}
\end{equation}

The objective and all constraints in \eqref{eq:SDP} are convex, providing us with a simple semi-definite program. Let \(\widehat{\mathbf X}_K\) denote a solution
to \eqref{eq:SDP}. We then retrieve a feasible clustering by applying \(K\)-means
to the first \(K\) eigenvectors of \(\widehat{\mathbf X}_K\), whose properties are well studied in the literature \citep{von2007tutorial}. Finally, we compare solutions for different values of $K$ and choose the clustering with the smallest objective, see Algorithm \ref{alg:1}. Some constraints can be optional to improve optimization speed, see Remark \ref{rem:optional}.

\begin{algorithm}[!ht]
\caption{Causal Clustering}
\label{alg:1}
\begin{algorithmic}[1]
\Require Adjacency matrix \(\mathbf A\), smallest and largest number of clusters
\(\underline K,\bar K\), parameter $\bar{\gamma}$ (default is $\bar{\gamma} = 10$), and calibration \(\xi_n=(\lambda\bar\phi_n^2/\bar\psi)^{-1}\). Boolean ``$K$-constraint''. 

\item \textit{If $K$-constraint is false:} Solve the SDP relaxation in Equation \eqref{eq:SDP} without the optional constraints and let
    \(\widehat{\mathbf X}\) denote the solution.
\For{\(K\in\{\underline K,\ldots,\bar K\}\)}
    \begin{algsubstates}
    \State \textit{If $K$-constraint is true:} Solve the SDP relaxation in Equation \eqref{eq:SDP} with the optional constraints, and let
    \(\widehat{\mathbf X}\) denote the solution.

    \State Retrieve a feasible clustering with \(K\) clusters by applying
    \(K\)-means to the first \(K\) eigenvectors of \(\widehat{\mathbf X}\).
    Call the resulting clustering \(\widehat{\mathcal C}_{n,K}\).

    \State Compute the objective in Equation \eqref{eqn:optimization_v10},
    $ 
    R_n(\widehat{\mathcal C}_{n,K};\xi_n)
    =
    \frac{\xi_n}{n^2}\sum_{k=1}^{K}\widehat n_k^2
    +
     b_n(\widehat{\mathcal C}_{n,K})^2 .
    $ 
    \end{algsubstates}
\EndFor

\State Denote the chosen clustering and the approximation error from discretization as
\[
\mathcal C_n^\star
\in
\arg\min_{\widehat{\mathcal C}_{n,K}:K\in\{\underline K,\ldots,\bar K\}}
R_n(\widehat{\mathcal C}_{n,K};\xi_n), \qquad \Gamma_n := R_n(\mathcal{C}_n^\star;\xi_n)\Big/\min_{K\in\{\underline K,\ldots,\bar K\}}
\underline R_{n,K}^{\mathrm{sdp}}(\xi_n).
\]

\Return \(\mathcal C_n^\star\), and  $\Gamma_n$
\end{algorithmic}
\end{algorithm}

The following theorem provides an observable bound on the optimization error.

 \begin{thm}[Observable approximation bound]
\label{thm:observable_approximation_error}
Let \(\Gamma_n\) be the observed approximation error from discretization of the SDP solution reported by Algorithm \ref{alg:1} and \(\mathcal C_n^\star\) be the corresponding clustering returned by Algorithm
\ref{alg:1}. Let $\xi_n = (\lambda \bar{\phi}_n^2/\bar{\psi})^{-1}$. Let \(\bar{\mathfrak C}\) denote the set of clusterings with \(K \in \{\underline{K}, \cdots, \bar{K}\}\) nonempty clusters with each cluster $k$ having size $n_k \le \bar{\gamma} n/K$ for given $\bar{\gamma}$ in Assumption \ref{ass:clusters}.   Then 
\[
\frac{
\mathcal B_n^*(\mathcal C_n^\star,\lambda)
}{
\min_{\mathcal C_n \in \bar{\mathfrak C}}
\mathcal B_n^*(\mathcal C_n,\lambda)
}
\le
\Gamma_n.
\]
\end{thm}

The proof is in Appendix \ref{proof:thm:observable_approximation_error}.
 Theorem \ref{thm:observable_approximation_error} provides us with a certificate to researchers about the quality of the clustering algorithm, where the upper bound $\Gamma_n$ is observed.   

\begin{rem}[Difference with spectral clustering]
Unlike minimum normalized cut typical for spectral clustering \citep{von2007tutorial}, the 
objective combines a cluster-size penalty,
$ 
\frac{1}{n^2}
\operatorname{tr}\!\left(
\one_n\one_n^\top \mathbf M_c(K)\mathbf M_c(K)^\top
\right),
$ 
with the square of a left-normalized cut, \\ 
$ 
\left[
\frac1n
\operatorname{tr}\!\left(
\mathbf L^\top
\left\{
\one_n\one_n^\top
-
\mathbf M_c(K)\mathbf M_c(K)^\top
\right\}
\right)
\right]^2, 
$ 
and does not necessarily admit a spectral relaxation. 
This differs from standard spectral clustering, which fixes \(K\)
and recovers clusters from eigenvectors of a graph Laplacian designed to
optimize a cut or normalized-cut criterion \citep{von2007tutorial}. Here,
instead, the SDP solution \(\widehat{\mathbf X}_K\) is obtained from an
optimization program that depends jointly on  
\(\mathbf  L\), which captures the bias induced by cross-cluster exposure, and
on \(\one_n\one_n^\top\), which captures the variance contribution from cluster
sizes. The subsequent eigendecomposition is only used as a rounding
device to the semidefinite solution. 
\qed
\end{rem}

\begin{rem}[Optional constraints] \label{rem:optional}
Note that the two sets of constraints are optional; these can improve stability and guarantee that Assumption \ref{ass:clusters} is satisfied by design but are not required for the validity of Theorem \ref{thm:observable_approximation_error} below. In practice, because they can make optimization slower (as without such constraint we only have to compute $\hat{\mathbf{X}}$ once, instead of each time for each $K$), we recommend not to impose these constraints, and only impose them if the resulting solution obtained through Algorithm \ref{alg:1} presents highly sized-unbalanced clusters violating Assumption \ref{ass:clusters}.  \qed 
\end{rem} 


\section{Inference, unknown spillovers and extensions} \label{sec:extensions}

Here we present settings where researchers are interested in inference on the GATE, settings where $\bar{\phi}_n^2/\bar{\psi}$ is unknown and a brief description of our additional extensions in Appendix \ref{app:more_extensions}. 

\vspace{-3mm}

\subsection{Experimental design for bias-aware inference} \label{sec:bias_aware_inferece}

In this section, we turn to settings where researchers are interested in conducting inference on the global treatment effect $\tau_n$. 
First, define 
  \begin{equation}\label{eq:exante_CBA_CI}
  \small 
  \begin{aligned} 
    \mathcal{I}_{\ea}(b, \sigma) = [\hat{\tau} \pm \chi(b, \sigma)],
    \end{aligned} 
  \end{equation}
  where
  \begin{equation}
  \small 
  \begin{aligned} 
    \label{eq:chi}
    \chi(b, \sigma) = (1-\alpha)\text{-th quantile of }|b + \sigma Z|\text{ where }Z \sim N(0, 1).
    \end{aligned} 
  \end{equation}
 In addition, define 
  $$
  \small 
  \begin{aligned} 
  \beta_n(\mu) = \tau_{n,\mu} - \mathbb{E}_\mu[\hat{\tau}_n(\mathcal{C}_n)], \quad \nu_n(\mu) = \sqrt{\mathbb{E}_\mu\left[\Big(\mathbb{E}_\mu[\hat{\tau}_n(\mathcal{C}_n)] - \hat{\tau}_n(\mathcal{C}_n)\Big)^2\right]}, \quad \sigma_{*,n}^2 =  \frac{\bar{\psi}}{n^2} \sum_{k=1}^{K_n} n_k^2. 
  \end{aligned} 
  $$
Here $b_{*,n}$ and $\sigma_{*,n}^2$ correspond to the worst-case bias and variance derived in Section \ref{sec:when}.

  \begin{thm} \label{thm:inference1} Suppose that Assumptions \ref{ass:first_order}, \ref{ass:exposure_restriction}, \ref{ass:clusters} hold. Let $\mathcal{N}_{n, \max} = o(K_n^{1/8})$. 
  For any $\mu$ and sequences $\bar{b}_n, \bar{\sigma}_n$ such that $|\beta_n(\mu)| \le \bar{b}_n$, $\nu_n(\mu) \le \bar{\sigma}_n$, with $\bar{\sigma}_n^2 \ge \frac{\eta + o(1)}{K_n}$ for some $\eta > 0$, it follows that 
  $
  \liminf_{n \rightarrow \infty} \mathbb{P}(\tau_{n, \mu} \in \mathcal{I}_{\ea}(\bar{b}_n, \bar{\sigma}_n)) \ge 1 - \alpha. 
  $ 
  \end{thm}

See Appendix \ref{proof:lem:inference1} for the proof. 
Theorem \ref{thm:inference1} establishes that the confidence interval $\mathcal{I}_{\ea}(\bar{b}_n, \bar{\sigma}_n)$ guarantees valid coverage pointwise for any $\mu$, provided that $\bar{b}_n$ and $\bar{\sigma}_n$ are valid upper bounds on the bias and standard deviation under a potential outcome model $\mu$.
In addition, Theorem \ref{thm:inference1} requires that such upper bound on the variance $\bar{\sigma}_n^2$ scales at order $1/K_n$ (or slower) to avoid degenerate solutions.\footnote{ This condition is attained when observations are e.g., positively dependent. It is also consistent with and attained by the worst-case variance derived in Lemma \ref{prop:main}, under a local asymptotic framework with $\varepsilon = o(1)$ in Assumption \ref{defn:local}.}  Note that the upper bound for the bias can be obtained directly from Lemma \ref{lem:worst_case_bias}. 

We optimize over an objective of the form 
$ 
\sup_{\mu \in \mathcal{M}} |\mathcal{I}_{\ea}(\beta_n(\mu), \nu_n(\mu)|.
$ 
As for our results about optimal clustering, in the characterization of the optimal confidence length for experimental design we leverage Assumption \ref{defn:local}.

\begin{thm} \label{thm:worst_case_inference} Suppose that Assumptions \ref{ass:first_order}, \ref{ass:exposure_restriction}, \ref{defn:local}, \ref{ass:clusters} hold. Then 
$$
\small 
\begin{aligned} 
\lim_{n\rightarrow \infty} \frac{\sup_{\mu \in \mathcal{M}} |\mathcal{I}_{\ea}(\beta_n(\mu), \nu_n(\mu))|}{|\mathcal{I}_{\ea}^*|(1 + \mathcal{O}(\varepsilon))} = 1, \quad 
|\mathcal{I}_{\ea}^*| = 2\left|\chi\Big(\bar{\phi}_n b(\mathcal{C}_n), \sigma_{*,n}\Big)\right|
\end{aligned} 
$$
\end{thm}

See Appendix \ref{proof:thm:worst_case_inference} for the proof. 
Theorem \ref{thm:worst_case_inference} shows that the length of the confidence intervals using the worst-case approximations for the bias and variance derived in Section \ref{sec:when} approximates the worst-case length of the confidence interval. 

Given this result, we can design the clusters with minimum worst-case confidence interval length. Since \(\chi(b,\sigma)\) is increasing in both arguments, any clustering that
minimizes \(|\mathcal I_{\ea}^*|\) must be Pareto efficient in worst-case bias
and variance. Therefore, the approximate algorithm proceeds as follows. First,
we compute candidate points on the Pareto frontier between worst-case variance
and bias by solving the SDP relaxation in Equation \eqref{eq:SDP} over a grid
of scalarization weights \(\xi\) and numbers of clusters \(K\). We then evaluate the target criterion
$ 
|\mathcal I_{\ea}^*|
=
2\chi\left(
\bar\phi_n b_n(\mathcal C_n),
\sqrt{
\bar\psi\frac{1}{n^2}\sum_{k=1}^{K_n}n_k^2
}
\right)
$ 
at the rounded clusterings and choose the clustering among those selected that yield the shortest interval.
Finally, note that once the experiment is conducted, it is possible to conduct inference on $\tau_\mu$ by using an estimator for the variance (instead of a worst-case bound), and invoke Theorem \ref{thm:inference1} for inference. Such estimator can be constructed using, for instance, the HAC estimator in \cite{gao2023causal}. We provide details in Appendix \ref{sec:estimated_variance2}. 

\vspace{-3mm}

\subsection{MSE with unknown spillover effects: unknown $\xi$}
\label{sec:unknown_spillover_size}

The objective in Equation \eqref{eqn:bound2} depends on the scalar
$\xi = \lambda\bar{\phi}_n^2/\bar{\psi}$, which may be difficult to calibrate ex ante. The goal of this subsection is to show how our framework and our main results can be used to also accommodate worst-case values of $\xi$. For simplicity, we consider a simple criterion corresponding to the ratio of MSE between the chosen clustering and the oracle that knows $\xi$. Other criteria are however possible (as for example worst-case/bias-aware interval lengths studied in Section \ref{sec:bias_aware_inferece}).

Specifically, for any feasible clustering $\mathcal C_n$, define its relative
loss at a given value of $\xi_n = (\lambda \bar{\phi}_n^2/\bar{\psi})^{-1}$ by the regret in terms of MSE-ratio for a given range of values $\xi \in [\underline{\xi}, \bar{\xi}]$
$$  
\small 
\begin{aligned} 
\mathcal{A}_{\xi_n}(\mathcal{C}_n) := \frac{
 \frac{\xi_n}{n^2}\sum_{k=1}^{K_n}n_k^2
+
 b_n(\mathcal C_n)^2
}{
\inf_{\widetilde{\mathcal C}_n}
\left\{
 \frac{\xi_n}{n^2}\sum_{k=1}^{| \widetilde{\mathcal C}_n| }\widetilde n_k^2
+
 b_n(\widetilde{\mathcal C}_n)^2
\right\}
},\qquad \bar{\mathcal{A}}(\mathcal{C}_n) = \sup_{\xi_n \in [\underline{\xi}, \bar{\xi}]} \mathcal{A}_{\xi_n}(\mathcal{C}_n) 
\end{aligned} 
$$ 
where the infimum is over the same feasible class, and $\widetilde n_k$ denotes
the size of cluster $k$ under $\widetilde{\mathcal C}_n$. Let
$$  
\small 
\begin{aligned} 
R_n^*(\xi) 
=
\inf_{\widetilde{\mathcal C}_n}
R_n(\mathcal C_n;\xi). 
\end{aligned} 
$$ 

The next result characterizes the ratio regret \citep[e.g.][]{epanomeritakis2025learning, armstrong2025adapting} in our context of spillovers, building on results in \cite{epanomeritakis2025learning}. In our context, the regret regularizes the worst-case choice of $\xi$ by introducing a penalty that depends on the oracle performance. Alternative notions of optimality are also possible, here omitted for brevity. 

\begin{thm} \label{thm:adaptation_regret} Let $\mathcal{C}_n^\star \in \mathrm{arg} \min_{\mathcal{C}_n} \bar{\mathcal{A}}(\mathcal{C}_n)$ and $\infty\ge \bar{\xi} >\underline{\xi} > 0$. Then 
$$
\small 
\begin{aligned} 
\mathcal C_n^\star
\in
\arg\min_{\mathcal C_n}
\max\left\{
\frac{R_n(\mathcal{C}_n, \bar{\xi})}{R_n^*(\bar{\xi})},
 \frac{R_n(\mathcal{C}_n, \underline{\xi})}{R_n^*(\underline{\xi})}
\right\}. 
\end{aligned} 
$$
\end{thm} 

 See Appendix \ref{proof:thm:adaptation_regret} for the proof.  Note that $\bar{\xi} = \infty$ is allowed within our framework.

Given Theorem \ref{thm:adaptation_regret}, we present a semi-definite relaxation as the one in Section \ref{sec:opt} by targeting regret as opposed to MSE in Algorithm \ref{alg:adaptive_causal_clustering}. To facilitate optimization, we consider lower bound on $R_n^*(\xi)$ that are readily available from Algorithm \ref{alg:1} in Appendix \ref{app:a} through the SDP relaxation. Using lower bounds guarantees that the chosen objective in Algorithm \ref{alg:adaptive_causal_clustering} provides a valid surrogate (upper bound) to the target regret.

\subsection{Additional extensions}

In Appendix \ref{app:more_extensions} we consider a series of other extensions. These are: (i) potential outcomes are random variables and we impose restrictions only on the moments of such random variables; (ii) spillover effects may be asymmetric between different individuals; (iii) researchers consider a larger class of estimators and optimize both over designs and estimators; (iv) saturation designs with positive treatment probabilities.

\vspace{-5mm}

\section{Empirical illustration and practical recommendations}\label{sec:empirics}

In this section, we illustrate the properties of the procedure in two empirical applications, one using unique data from the Facebook friendship and messaging graphs and one using data from an experiment conducted in rural China by \cite{cai2015social}. 

The first application shows how our results can be used to compare existing clustering algorithm through our proposed objective function; the second application illustrate how our results can be used to derive an explicit clustering. 

\vspace{-3mm} 

\subsection{Comparisons between existing clusterings at Facebook}
 
 We consider two clustering algorithms implemented at scale on social networks owned by Meta:\footnote{We focus on existing clusterings (instead of implementing the causal clustering proposed here), as implementing new clustering methods requires changes to the entire engineering pipeline as described in \cite{karrer2021network}, outside the scope of this paper.}  Louvain algorithm \citep{blondel2008fast}, and Balanced Partitioning \citep{kabiljo2017social}. Louvain and Balanced Partitioning clustering produces a hierarchical clustering structure with a growing number of clusters. We consider three ``types" of such algorithms, corresponding to three levels in the clustering structure hierarchy (i.e., different numbers of clusters), defined as Type 1, 2, and 3, respectively. For Louvain, higher-order types denote less clusters, whereas it is the opposite for Balanced Partitioning. For each clustering, we also report $\log(n/K_n)$ in Table \ref{tab:meta1}.

We consider two graphs owned by Meta. In the first graph, edges capture the strength of the friendship in the Facebook graph, and in the second graph, they capture connections based on messaging on Facebook. In both cases, the data was aggregated and de-identified. In each graph, edges are continuous variables. For each graph, we construct three adjacency matrices by setting the edges to be zero if below the $5^{th}$, $10^{th}$, and $50^{th}$ percentile value of the edges. We will refer to these graphs obtained as dense, moderate (mod), and sparse after thresholding. We also note that all graphs have a bounded maximum degree.

We refer to the worst-case bias and variance weighted by $\xi$ as 
$
b_n(\mathcal{C}_n),  \frac{\xi}{n^2} \sum_{k=1}^{K_n} n_k^2. 
$   

\medskip \noindent \textbf{Results}
Table \ref{tab:meta1} in the appendix collects each of these algorithms' worst-case bias and variance and the number of clusters per individual. A larger number of clusters increases the bias but decreases the variance throughout all graphs. (The bias increases for sparser networks because the denominator $|\mathcal{N}_i|$ decreases.) Louvain algorithm dominates Balanced Partitioning for most designs in variance and bias, suggesting that practitioners may prefer Louvain.  

We study different types of Louvain algorithms (each having a different number of clusters) in Figure \ref{fig:meta1}. Figure \ref{fig:meta1} reports the worst-case mean squared error for different graphs, Louvain algorithms, and degrees of sparsity. For $\xi \approx 1$, the Louvain algorithm with the smallest number of clusters (Type 3) dominates the other two algorithms in all but one case. 

As $\xi$ increases, 
Louvain clustering with the smallest number of clusters (Type 3) is optimal for a dense graph but sub-optimal as the graph becomes more sparse. This result is intuitive: for denser networks, a smaller number of clusters may best control the estimator's bias relative to the bias induced by other clustering algorithms. However, we observe that Type 1 is the most robust in terms of variance in the mean-squared error across different values of spillover effects. This is suggestive that Louvain clustering with the largest number of clusters (Type 1) may be preferred in various applications. 

\medskip 
\noindent \textbf{Comparison between Bernoulli and Cluster design} To conclude, we study when a cluster design with our preferred clustering (Louvain Type 1) should be preferred over a Bernoulli design taking $\lambda = 1$.
Using Theorem \ref{thm:thumb}, and the dense Messaging graph,   Theorem \ref{thm:thumb} for Louvain clustering (Type 1), for 
$
\sqrt{\frac{\bar{\psi}}{\bar{\phi}_n^2}} \le \bar{\xi}^{1/2}, \quad \bar{\xi}^{1/2} = 706, 
$
we should run a cluster experiment instead of a Bernoulli design. 
Therefore, researchers should run a Cluster design if 
$
\bar{\phi}_n^2 \ge \bar{\psi}/\bar{\xi}.
$
This implies that the researcher should run a cluster design if spillovers are larger than $\bar{\phi}_n \ge \bar{\psi}^{1/2}/706$. This comparison is suggestive that a cluster design with a Louvain clustering may be preferred over a Bernoulli design over a wide range of values of spillover effects (and outcomes variation $\bar{\psi}$). For binary outcomes, we should prefer cluster experiments for $\bar{\phi}_n \ge 0.003$.

In summary, our results shed light on using clustering algorithms for large-scale implementation on online platforms. These results suggest that Louvain clustering with possibly many clusters performs best in practice, and may often be preferred to a Bernoulli design. It confirms heuristic arguments in \cite{karrer2021network} based on AA-tests.  

\definecolor{glaucous}{rgb}{0.38,0.51,0.71}

\begin{figure}[!ht]
    \centering
    \includegraphics[scale = 0.6]{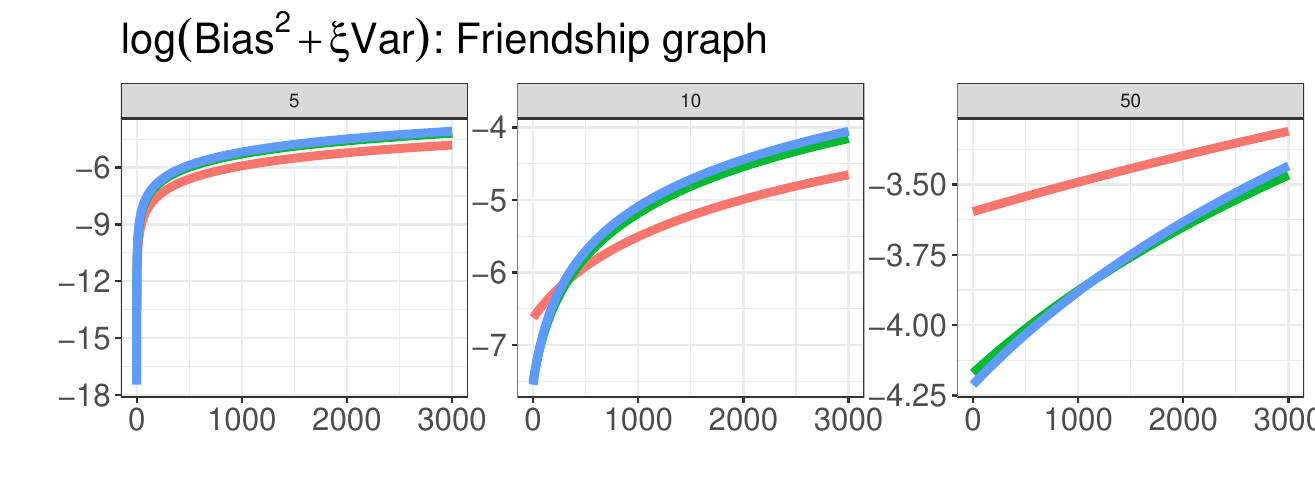}
     \includegraphics[scale = 0.6]{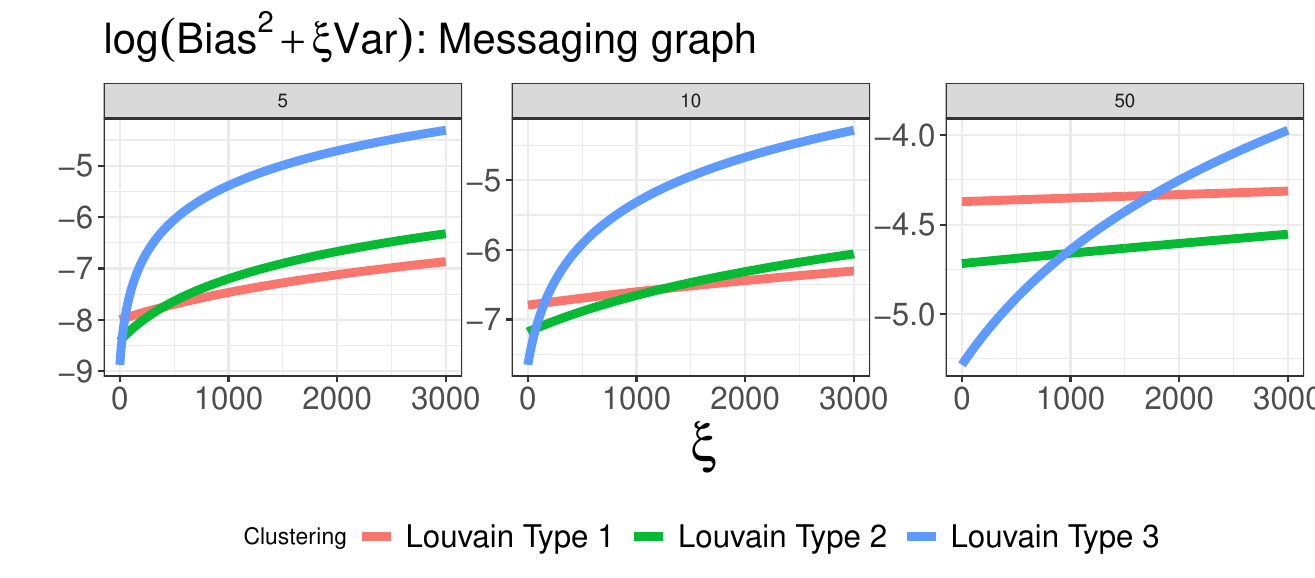}
   \caption{Clusters comparisons for Louvain clustering. Different Types correspond to different numbers of clusters (with Type 1 having the largest number of clusters). Different panels correspond to different graphs where two individuals are not connected if the connection (measured with a continuous variable) is below the $5^{th}, 10^{th}, 50^{th}$ percentile (dense, moderate, and sparse graph). The two graphs in the panels are Facebook friendship and Facebook messaging. }
    \label{fig:meta1}
\end{figure}

\vspace{-3mm} 

\subsection{Clustering in the field}
\label{sec:info}

\begin{figure}[!ht]
    \centering
      \includegraphics[scale = 0.5]{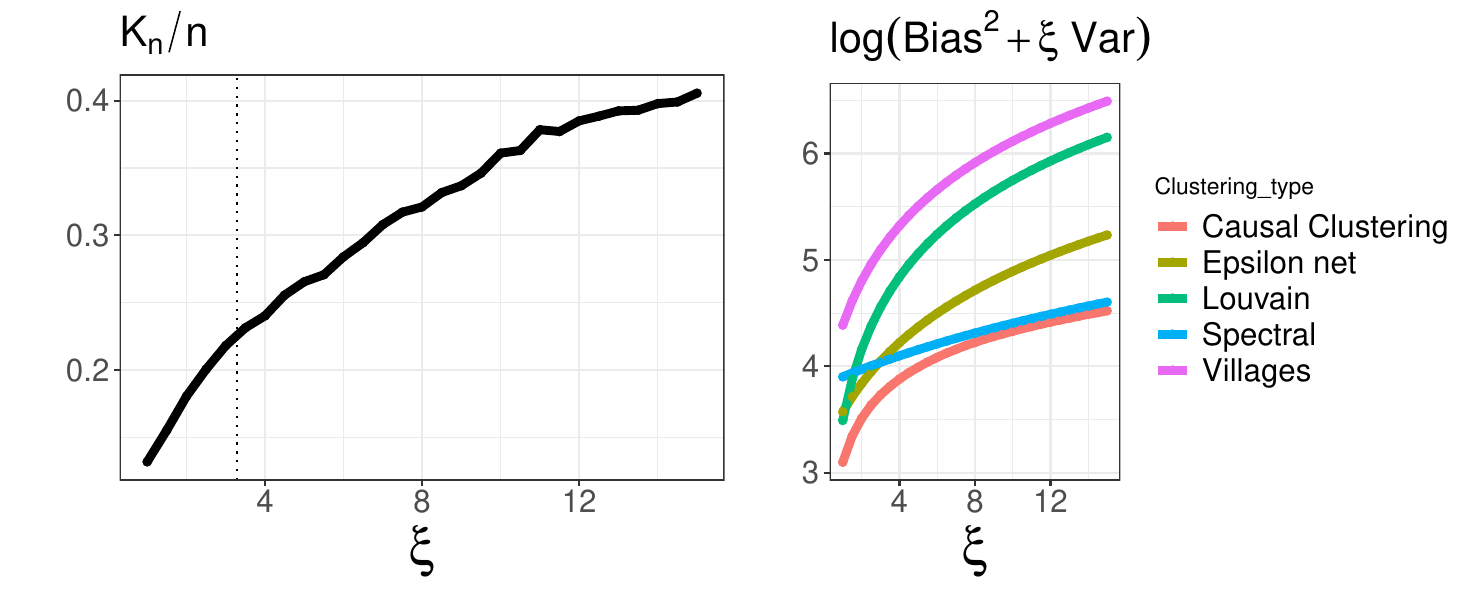}
    \caption{Causal clustering in the field application. The left panel reports the number of clusters selected by Algorithm \ref{alg:1}, normalized by the number of units in the region, as a function of \(\xi\). The dotted vertical line corresponds to the calibration \(\xi=3.29\) discussed in the text. The right panel reports the natural logarithm of the design objective \(b_n(\mathcal C_n)^2+\xi n^{-2}\sum_k n_k^2\), average over the regions in \cite{cai2015social} for different clustering procedures. Algorithm \ref{alg:1} corresponds to Causal Clustering. Results use the network data from \cite{cai2015social} and average across the regions in \cite{cai2015social}.}
    \label{fig:1}
\end{figure}

Next, we present an application to the field experiment of \cite{cai2015social}, where the treatment consists of informing individuals about an insurance product.

\cite{cai2015social} collected network information in 185 villages in rural China from 48 larger regions. We use these network data to study the properties of the proposed method. The network has a total of \(7,649\) nodes, once we also include individuals who are friends of surveyed individuals. Individuals have on average about \(50\%\) of their connections outside their own village. By contrast, individuals in different regions are almost disconnected: about \(99\%\) of links are within region. We therefore construct clusterings separately within each region. We use the ``weak ties'' adjacency matrix, where two individuals are connected if either person lists the other as a friend.

We compare our procedure, denoted ``Causal Clustering,'' with four alternative clustering rules: \(\varepsilon\)-net clustering as in \cite{eckles2017design}, using \(\varepsilon=3\) as suggested in \cite{eckles2017design}; spectral clustering with a fixed number of clusters equal to \(n/3\), where \(n\) is the population size in the region; Louvain clustering with default parameters selected by the R package \texttt{igraph}; and clustering based on village identity, creating one additional cluster for individuals with missing village information.\footnote{These procedures are all off-the-shelf. None is designed to minimize the worst-case mean-squared error criterion for the GATE studied in this paper. Whenever possible, we select tuning parameters using either values suggested by the literature, as for \(\varepsilon\)-net clustering, or default software implementations, as for Louvain clustering.}

\medskip
\noindent\textbf{Design objective and number of clusters.}
Figure \ref{fig:1} reports the behavior of Algorithm \ref{alg:1} and compares its realized objective with the benchmark clusterings. The left panel shows that the selected number of clusters increases with \(\xi\). This is consistent with the bias--variance trade-off in Equation \eqref{eqn:optimization_v10}: larger \(\xi\) places more weight on the variance component and therefore favors designs with more clusters. The right panel compares the realized squared-bias objective,
\[
b_n(\mathcal C_n)^2+\xi\frac{1}{n^2}\sum_{k=1}^{K_n}n_k^2,
\]
across clustering methods. Causal Clustering attains the lowest objective over the full range of \(\xi\) values considered. Spectral clustering is the closest competitor only for large values of $\xi$ while it significantly underperforms for values of $\xi \in [1,4]$. Louvain, \(\varepsilon\)-net clustering, and village-based clustering also have substantially larger values than Causal Clustering. The poor performance of village-based clustering is intuitive in this application: many social ties cross village boundaries, so administrative units do not align well with the network structure relevant for spillovers.

To interpret the magnitude of \(\xi\), take \(\bar\psi=0.24\), approximately equal to the outcome variance in \cite{cai2015social}, and \(\bar\phi_n=0.27\), corresponding to the spillover coefficient in Table 2 of \cite{cai2015social}. This gives
\[
\xi=\frac{\bar\psi}{\bar\phi_n^2}\approx 3.29.
\]
At this calibration, the selected number of clusters is much larger than the number of administrative villages. This illustrates the practical relevance of optimizing the number and structure of clusters using the downstream causal objective rather than relying on pre-existing administrative partitions.

\begin{figure}[!ht]
    \centering
      \includegraphics[scale = 0.45]{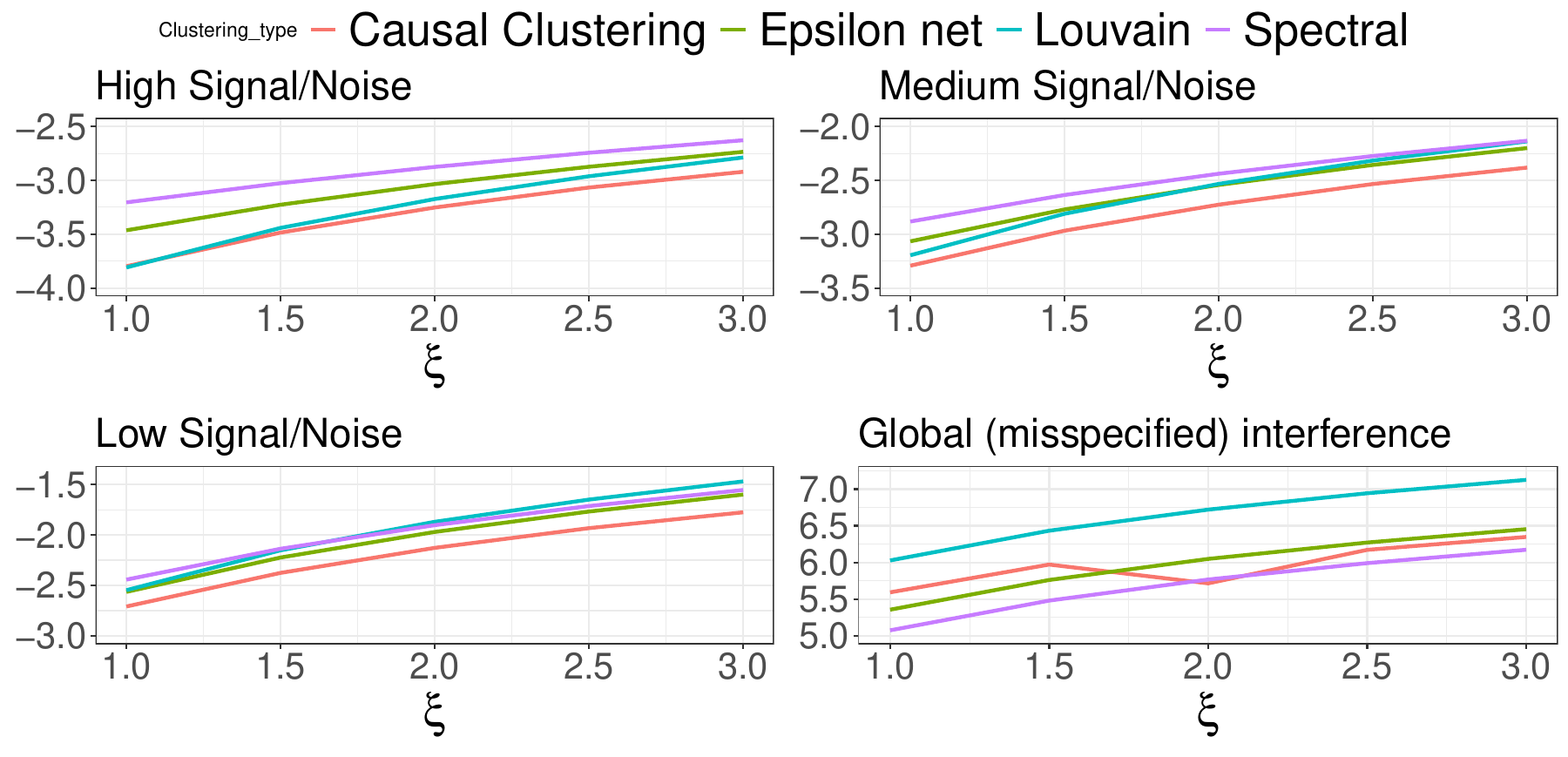}
    \caption{Weighted mean-squared error in calibrated simulations. The figure reports the natural logarithm of \(\widehat{\mathrm{Bias}}^2+\xi\,\widehat{\mathrm{Var}}\), averaged across all regions in \cite{cai2015social}. The first three panels use a first-order interference model calibrated to Table 2, Column 4 of \cite{cai2015social}, varying the residual variance \(\sigma^2\in\{1/4,1/2,1\}\) to represent high, medium, and low signal-to-noise ratios. The last panel uses a misspecified global-interference design calibrated to Table 5, Column 4 of \cite{cai2015social}, where outcomes depend on neighbors' outcomes rather than only on neighbors' treatments. Lower values indicate better performance.}
    \label{fig:cai_mse}
\end{figure}

\medskip
\noindent\textbf{Weighted mean-squared error comparisons.}
We next evaluate the clusterings in calibrated simulations. For each clustering estimated from the observed network, we repeatedly assign treatment at the cluster level and simulate outcomes under the estimated model from \cite{cai2015social}. For each clustering \(\mathcal C_n\) and each value of \(\xi\), we compute the weighted mean-squared-error criterion
\begin{equation}
\label{eqn:M_xi}
\small
\begin{aligned}
M_\xi(\mathcal C_n)
=
\Big(\mathbb E[\hat\tau(\mathcal C_n)]-\tau\Big)^2
+
\xi\,\mathbb V\!\left(\hat\tau(\mathcal C_n)\right).
\end{aligned}
\end{equation}
This is the simulation analogue of the design criterion in Equation \eqref{eqn:optimization_v10}: the first term captures squared bias, while the second term captures the variance penalty indexed by \(\xi\). We report the natural logarithm of \(M_\xi(\mathcal C_n)\). Note that $M_\xi$ captures an actual bias and variance (as opposed to the upper bound minimized by our objective). 

The first three panels of Figure \ref{fig:cai_mse} use a first-order interference model calibrated to \cite{cai2015social}, Table 2, Column 4, where outcomes depend on neighbors' treatments. We vary the residual variance \(\sigma^2\in\{1/4,1/2,1\}\), corresponding to high, medium, and low signal-to-noise ratios. Causal Clustering performs best throughout the range of \(\xi\) values shown. In contrast, any other clustering algorithm performs well only for some or none of the ranges of $\xi$ and values of $\sigma^2$.  

The last panel considers a misspecified global-interference environment calibrated to \cite{cai2015social}, Table 5, Column 4. In this design, outcomes depend on neighbors' outcomes rather than only on neighbors' treatment assignments, violating the first-order interference model used to motivate our objective. This exercise is useful as a robustness check. Even in this misspecified setting, Causal Clustering is competitive to the benchmark methods over the displayed range of \(\xi\), although not necessarily the optimal depending on the values of $\xi$. The results suggest that optimizing the clustering for the causal bias--variance trade-off can deliver practical gains under the maintained first-order model primarily, while avoiding significant underperformance even in settings with richer forms of network propagation.

To gain further insights, 
Table \ref{tab:weighted_mse_regret} summarizes these comparisons by reporting the average regret of each method relative to the best-performing clustering rule within each of the competitors considered in simulations at each value of \(\xi\).  Causal Clustering has essentially zero regret in the first-order interference designs: its average regret is \(0.2\%\) in the high-signal design, \(0.0\%\) in the medium-signal design, and \(0.0\%\) in the low-signal design. By contrast, the benchmark methods are more context-dependent. Spectral clustering performs poorly in the first-order designs, with average regret equal to \(51.8\%\), \(36.3\%\), and \(26.4\%\) in the high-, medium-, and low-signal designs, respectively, even though it performs best in the misspecified global-interference design. Louvain is relatively competitive in the high-signal design, with \(7.6\%\) average regret, but its performance deteriorates as the signal-to-noise ratio falls, reaching \(20.1\%\) in the medium-signal design and \(28.1\%\) in the low-signal design; it also performs especially poorly under global misspecification, with \(161.8\%\) regret. \(\varepsilon\)-net clustering is more stable but never close to the best method, with regrets between \(17.2\%\) and \(33.8\%\) across designs. Although Causal Clustering is not pointwise optimal in the misspecified global-interference design, its average regret across the four simulation designs is \(8.6\%\), compared with \(24.9\%\) for \(\varepsilon\)-net clustering, \(28.9\%\) for spectral clustering, and \(54.4\%\) for Louvain. These results suggest that Causal Clustering provides a robust procedure across the simulation designs considered, while the benchmark procedures can perform well in some settings but poorly in others.

\begin{table}[!ht]
\centering
\caption{Average percent excess weighted MSE relative to the best competitor. For each scenario and value of $\xi$, regret is computed as $100\times(M_\xi/M_{\min}-1)$ and then averaged over $\xi\in[1,3]$. The final column averages across the four simulation designs. Lower values indicate better performance.}
\label{tab:weighted_mse_regret}
\begin{tabular}{lcccc|c}
\toprule
Method/Regret & High S/N & Medium S/N & Low S/N & Global misspec. & Average \\
\midrule
Causal Clustering & 0.2 & 0.0 & 0.0 & 34.0 & 8.6 \\
Epsilon net & 27.3 & 21.3 & 17.2 & 33.8 & 24.9 \\
Spectral & 51.8 & 36.3 & 26.4 & 1.1 & 28.9 \\
Louvain & 7.6 & 20.1 & 28.1 & 161.8 & 54.4 \\
\bottomrule
\end{tabular}
\end{table}

 \subsection{Practical recommendation: Tuning parameter $\xi$} \label{rem:psi}
Similarly to standard power analysis \citep[e.g.][]{baird2018optimal}, our method depends on the choice of $(\bar{\phi}_n^2/\bar{\psi})^{-1}$ (the size of spillover effects $\bar{\phi}_n$ relative to the outcomes' largest squared deviation $\bar{\psi}$). In practice, researchers may choose a range of values; in these cases,  
our recommended choice for $\bar{\psi}$ is to consider values $\bar{\psi} = c \bar{\sigma}^2, c \in [1,4]$ and $\bar{\sigma}^2$ is the baseline variance of the residuals after removing the covariate adjustment, and $c$ is a constant between one and four.\footnote{To gain further intuition, let 
$\lambda = 1$. Let $\bar{\mu}_i$ be a prediction for $\mu_i(\mathbf{0})$ as in Remark \ref{rem:reg_adj} and consider an estimator as in Equation \eqref{eqn:reg_adj}. When using regression adjusted estimators, we have $\bar{\psi} = (\mu_i(\mathbf{1}) - \bar{\mu}_i + \mu_i(\mathbf{0}) - \bar{\mu}_i)^2$. By approximating $(\mu_i(\mathbf{d}) - \bar{\mu}_i)^2 \approx \bar{\sigma}^2$, we can write $\bar{\psi} \le 4 \bar{\sigma}^2$.}  Alternatively, if available, the choice of $\bar{\phi}_n$ can be based on spillover effects observed in previous experiments, in the spirit of minimum detectable effects \citep{baird2018optimal}. 

Finally, if researchers do not want to specify any information about $\xi$ but only a reasonable range, they can use the worst-case regret approach we present in Algorithm \ref{alg:adaptive_causal_clustering}. 

\section{Conclusions} 

This paper provides a formal characterization of the bias and variance of the clustering algorithm to estimate the overall treatment effect. Using such characterization, it introduces a novel clustering method that takes into account the downstream task of causal inference. Applicability of our method spans online experiments, cash-transfer programs, information campaigns to cite some. 
The key parameter of the algorithm is how to trade-off variance and bias through $\xi$. We provide practical recommendations below. 

Many open questions remain. Future research should study explicit optimal clustering with non-local network interference as for example in \cite{bramoulle2009identification, leung2022causal} among others. Topics of future research also include clustering in dynamic settings, or optimal clustering for a larger class of estimators.


\bibliography{reference}
\bibliographystyle{chicago}

\appendix 

\newpage 
\setcounter{page}{1}

\spacingset{1.5}

\section{Additional tables and algorithms} \label{app:a}

\begin{table*}[!ht]\centering
  \label{tab:mse2} 
\begin{tabular}{@{}lrrrcrrrcrrr@{}}\toprule
Friendship & \multicolumn{3}{c}{Type 1} & \ & \multicolumn{3}{c}{Type 2} & \ & \multicolumn{3}{c}{Type 3} \\
\cmidrule{2-4} \cmidrule{6-8}  \cmidrule{10-12}  
&\small dense&\small mod &\small sparse && \small dense&\small  mod  & \small sparse && \small dense& \small mod  & \small sparse \\\Xhline{.8pt} 
 \rowcolor{lightgray}\multicolumn{12}{c}{Balanced Partitioning Algorithm} \\ \Xhline{.8pt} 
 \small $100$ Bias & $38.41$ & $38.54$ & $39.95$&&$55.21$ & $55.17$ & $55.89$ &&$69.00$ & $68.92$ & $69.37$\\
  \small $10000$ Variance & $9.76$ & $ 9.76$ & $9.85$&&$0.30$ & $0.30$ & $0.30$ &&$0.01$ & $0.01$ & $0.01$\\

      \cline{2-4}  \cline{6-8}  \cline{10-12} 
  \small $\log(n/K_n)$ & $14.41$ & $14.37$ & $14.17$&& $10.95$ & $ 10.91$ & $10.70$&&7.48&7.44&7.24\\
\Xhline{.8pt}  \rowcolor{lightgray}\multicolumn{12}{c}{Louvain Algorithm} \\ \Xhline{.8pt} 
 \small $100$ Bias & $0.17$ & $3.64$ & $16.55$&&$0.01$ & $2.34$ & $12.43$ &&$0.01$ & $2.29$ & $12.17$\\
  \small $10000$ Variance & $0.03$ & $0.03$ & $ 0.03$&&$0.05$ & $0.05$ & $0.05$ &&$0.06$ & $0.06$ & $0.06$\\
      \cline{2-4}  \cline{6-8}  \cline{10-12} 
  \small $\log(n/K_n)$ & $2.30$ & $2.28$ & $2.16$&& $4.51$ & $4.56$ & $4.59$&& 5.32&5.50&6.05\\
\bottomrule
\end{tabular}
    
\begin{tabular}{@{}lrrrcrrrcrrr@{}}\toprule
Messaging & \multicolumn{3}{c}{Type 1} & \ & \multicolumn{3}{c}{Type 2} & \ & \multicolumn{3}{c}{Type 3} \\
\cmidrule{2-4} \cmidrule{6-8}  \cmidrule{10-12}  
&\small dense&\small mod &\small sparse && \small dense&\small  mod  & \small sparse && \small dense& \small mod  & \small sparse \\\Xhline{.8pt} 
 \rowcolor{lightgray}\multicolumn{12}{c}{Balanced Partitioning Algorithm} \\ \Xhline{.8pt} 
 \small 100 Bias & $17.26$ & $18.10$ & $22.57$&&$26.94$ & $27.62$ & $31.00$ &&$37.06$ & $37.46$ & $39.06$\\
  \small 10000 Variance & $97.65$ & $97.65$ & $97.71$&&$3.05$ & $3.05$ & $3.05$ &&$0.09$ & $0.09$ & $0.09$\\

      \cline{2-4}  \cline{6-8}  \cline{10-12} 
  \small $\log(n/K_n)$ & $16.04$ & $16.02$ & $15.76$&& $12.57$ & $12.55$ & $12.30$&&9.10&9.09&8.83\\
\Xhline{.8pt}  \rowcolor{lightgray}\multicolumn{12}{c}{Louvain Algorithm} \\ \Xhline{.8pt} 
 \small 100 Bias & $1.83$ & $3.34$ & $11.24$&&$1.48$ & $2.76$ & $9.45$ &&$1.18$ & $2.18$ & $7.12$\\
  \small 10000 Variance & $0.02$ & $0.02$ & $0.02$&&$0.05$ & $0.05$ & $0.05$ &&$0.44$ & $0.44$ & $0.45$\\
      \cline{2-4}  \cline{6-8}  \cline{10-12} 
  \small $\log(n/K_n)$ & $3.50$ & $3.49$ & $ 3.50$&& $6.05$ & $6.09$ & $6.19$&&6.45&6.52& 6.88\\
\bottomrule
\end{tabular}
 \caption{Worst-case bias and variance for Balanced Partition Clustering and Louvain clusterings, and for two different graphs owned by Meta. Different types correspond to algorithms with an increasing number of clusters for Balanced Partition and a decreasing number of clusters for Louvain. Note that for Type 1 to Type 3 of Balanced partition, the number of clusters increases, while for Louvain algorithm, the number of clusters decreases. Here $\log(\cdot)$ indicates natural log. }
\label{tab:meta1}
\end{table*}

 \begin{algorithm}[!ht]
\caption{Causal Clustering over a Range of Spillover Calibrations}
\footnotesize 
\label{alg:adaptive_causal_clustering}
\begin{algorithmic}[1]
\Require Adjacency matrix \(\mathbf A\), smallest and largest number of clusters
\(\underline K,\bar K\), calibration range \([\underline\xi,\bar\xi]\).

\State Solve Equation \eqref{eq:SDP} without optional constraint for $\xi \in \{\underline{\xi}, \bar{\xi}\}$ and compute the corresponding objective $\underline R_{n}^{\mathrm{sdp}}(\underline{\xi}), \underline R_{n}^{\mathrm{sdp}}(\bar{\xi})$. These provide valid lower bound for $R_n^*(\underline{\xi}), R_n^*(\bar{\xi})$, respectively. 
    \State Solve the SDP relaxation
    \[
    \begin{aligned}
    \min_{\mathbf X,z,t,\rho}\quad
    & \rho\\
    \text{s.t.}\quad
    & \mathbf X\succeq0,\qquad X_{ii}=1,\quad i=1,\ldots,n, \qquad \rho \ge 0\\
    & z=
    \frac1n
    \operatorname{tr}\!\left(
    \mathbf L^\top
    [\one_n\one_n^\top-\mathbf X]
    \right),
    \qquad
    \begin{pmatrix}
    t & z\\
    z & 1
    \end{pmatrix}
    \succeq0,\\
    & \underline\xi
    \frac1{n^2}
    \operatorname{tr}(\one_n\one_n^\top\mathbf X)
    +t
    \le
    \rho \underline R_{n}^{\mathrm{sdp}}(\underline{\xi}),\\
    & \bar\xi
    \frac1{n^2}
    \operatorname{tr}(\one_n\one_n^\top\mathbf X)
    +t
    \le
    \rho \underline R_{n}^{\mathrm{sdp}}(\bar{\xi}). 
    \end{aligned}
    \]
    \State Let \(\widehat{\mathbf X}\) denote the solution.

\For{\(K\in\{\underline K,\ldots,\bar K\}\)}
    \begin{algsubstates}
    \State Retrieve a feasible clustering with \(K\) clusters by applying
    \(K\)-means to the first \(K\) eigenvectors of \(\widehat{\mathbf X}\).
    Call the resulting clustering \(\widehat{\mathcal C}_{n,K}\).

    \State Compute
    \[
    \widehat\rho_K
    =
    \max\left\{
    \frac{
    R_n(\widehat{\mathcal C}_{n,K};\underline\xi)
    }{
    \underline R_{n}^{\mathrm{sdp}}(\underline{\xi})
    },
    \frac{
    R_n(\widehat{\mathcal C}_{n,K};\bar\xi)
    }{
   \underline R_{n}^{\mathrm{sdp}}(\bar{\xi})
    }
    \right\}, \qquad 
    R_n(\widehat{\mathcal C}_{n,K};\xi)
    =
    \frac{\xi}{n^2}
    \sum_{k=1}^{K}\widehat n_k^2
    +
    b_n(\widehat{\mathcal C}_{n,K})^2.
    \]
    \end{algsubstates}
\EndFor

\State Choose
\[
\widehat{\mathcal C}_n
\in
\arg\min_{\widehat{\mathcal C}_{n,K}:K\in\{\underline K,\ldots,\bar K\}}
\widehat\rho_K.
\]

\Return \(\widehat{\mathcal C}_n\).
\end{algorithmic}
\end{algorithm}

\newpage

\section{Omitted proofs in the main text} \label{app:proofs}

Throughout the proofs, expectations are conditional on the adjacency matrix $\mathbf{A}$. 

\subsection{Auxiliary lemmas}

In the following lines, we study the variance of the estimator $\hat{\tau}_n(\mathcal{C}_{ n})$.  Observe that 
$$
\small 
\begin{aligned} 
\mathbb{E}_\mu\Big[\Big(\hat{\tau}_n(\mathcal{C}_{ n}) - \mathbb{E}_\mu[\hat{\tau}_n(\mathcal{C}_{ n})]\Big)^2\Big] = \frac{4}{n^2} \sum_{i,j} \mathrm{Cov}\Big(\mu_i(D_i, \mathbf{D}_{-i})[2 D_i - 1], \mu_j(D_j, \mathbf{D}_{-j})[2 D_j - 1]\Big). 
\end{aligned} 
$$

\begin{lem}[Zero covariances] \label{lem:2} Suppose that Assumptions \ref{ass:first_order}, \ref{ass:clusters} hold. For all $i \in \{1, \cdots, n\}$
$$
\begin{aligned} 
 & \mu_i(D_i, \mathbf{D}_{-i})\Big[2D_i - 1\Big] \perp  \mu_j(D_j, \mathbf{D}_{-j})\Big[2D_j - 1\Big] \quad \forall j \not \in B_i \cup G_i,  
 \end{aligned}
$$
where 
$$
\small 
\begin{aligned}
B_i& =  \Big\{v \in \{1, \cdots, n\}: \text{ either } c(v) = c(i) \text{ or } c(v) = c(v'), \text{ for some } v' \in \mathcal{N}_i \Big\}. \\ 
G_i &= \Big\{ g \in \{1, \cdots, n\}: \mathcal{N}_g \cap B_i \neq \emptyset \Big\}
\end{aligned} 
$$
 \end{lem}

\begin{proof}[Proof of Lemma \ref{lem:2}] See Appendix \ref{proof:lem:2}. 
\end{proof}

Lemma \ref{lem:2} states that two realized outcomes are independent (and therefore have zero covariance) if two individuals (i) are in two different clusters, such that none of the two clusters contains a friend of the other individual, and (ii) are not friends  or share a common friend (set), and if there is no friend of $j$ in a cluster that contains a friend of $i$. Lemma \ref{lem:2} is equivalent to state that $\mu_i(D_i, \mathbf{D}_{-i})[2D_i - 1], \mu_j(D_j, \mathbf{D}_{-j})[2D_j - 1]$ have zero covariance if  $B_i \cap B_j = \emptyset$.\footnote{By definition, $B_i \cap B_j \neq \emptyset$ implies the existence of $i'\in \{i\}\cup \mathcal{N}_i, j' \in \{j\}\cup \mathcal{N}_j$, and $k$ such that $c(i') = c(k) = c(j')$. Thus, $j'\in B_i$. If $j' = j$, then $j\in B_i$; otherwise, $j\in G_i$. For the converse, first note that $j\in B_i$ implies that $j\in B_i \cap B_j \neq \emptyset$ since $j\in B_j$ for any $j$. If $j\in G_i\setminus B_i$, then there exists $j'\in \mathcal{N}_j$ such that $j'\in B_i$. Since $j' \in B_j$, $j'\in B_i \cap B_j \neq\emptyset$.}  Next, we analyze the covariances for the remaining units.

\begin{lem}[Non-zero Covariances] \label{lem:rel} Suppose Assumptions \ref{ass:first_order},  \ref{ass:exposure_restriction}, \ref{ass:clusters} hold. Then
\begin{equation} \label{eqn:bound_covariances} 
\small 
\begin{aligned} 
&\Big|\mathrm{Cov}\Big(\mu_i(D_i, \mathbf{D}_{-i})\Big[2D_i - 1\Big], \mu_j(D_j, \mathbf{D}_{-j})\Big[2D_j - 1\Big]\Big)\Big| = \mathcal{O}(\bar{\phi}_n )\quad \forall j: c(j) \neq c(i).
\end{aligned} 
\end{equation}  
In addition, for $c(i) = c(j)$,
\begin{equation} \label{eqn:bound_covariances2}
\small 
\begin{aligned}
& \mathrm{Cov}\Big(\mu_i(D_i, \mathbf{D}_{-i})[2D_i - 1], \mu_j(D_j, \mathbf{D}_{-j})[2D_j - 1]\Big)
=  \frac{1}{4}\Big(\mu_i(\mathbf{1}) + \mu_i(\mathbf{0}) \Big)\Big(\mu_j(\mathbf{1}) + \mu_j(\mathbf{0}) \Big) + \mathcal{O}(\bar{\phi}_n ).
\end{aligned} 
\end{equation}
Above, the big-O terms are uniform over $i$ and $j$.
\end{lem}

 \begin{proof} See Appendix \ref{proof:psi}. 
 \end{proof}

Lemma \ref{lem:rel} characterizes the covariance between individuals in different clusters (Equation \ref{eqn:bound_covariances}) and individuals in the same cluster (Equation \ref{eqn:bound_covariances2}). For individuals in different clusters, the covariance is of the \textit{same} order of the bias, whereas for individuals in the \textit{same} cluster, the covariance is $\mathcal{O}(1)$. 
The component $\frac{1}{4}\Big(\mu_i(\mathbf{1}) + \mu_i(\mathbf{0}) \Big)\Big(\mu_j(\mathbf{1}) + \mu_j(\mathbf{0}) \Big)$ captures the covariance between individuals in the same clusters up-to a factor of order $\bar{\phi}_n$. The covariance between individuals in different clusters is zero if there are no individuals with neighbors in a different cluster since, in this case, the within-cluster covariance captures all the covariances between individuals.

\subsection{Proof of Lemma \ref{lem:worst_case_bias}} \label{sec:proof_bias}

    We prove the equality by proving inequality on both sides. 
    
    \paragraph{Upper bound ($\le$ case)} We have 
$$
\small 
\begin{aligned} 
\sup_{\mu \in \mathcal{M}} \Big| \mathbb{E}_\mu[\hat{\tau}_n] - \tau_{n,\mu}\Big| & = 
\sup_{\mu \in \mathcal{M}} \Big|
\frac{1}{n}\sum_{i=1}^n \mathbb{E}\Big[\mu_i(D_i, \mathbf{D}_{-i}) - \mu_i(\mathbf{1}) | D_i = 1\Big] - \mathbb{E}\Big[\mu_i(D_i, \mathbf{D}_{-i}) - \mu_i(\mathbf{0}) | D_i = 0\Big]\Big| \\ 
&\le \sup_{\mu(1, \cdot) \in \otimes_{i=1}^n \mathcal{M}_{1, i}} \Big|
\frac{1}{n} \sum_{i=1}^n \mathbb{E}\Big[\mu_i(D_i, \mathbf{D}_{-i}) - \mu_i(\mathbf{1}) | D_i = 1\Big]\Big| \\ & + \sup_{\mu(0, \cdot) \in \otimes_{i=1}^n \mathcal{M}_{0, i}} \Big| \frac{1}{n} \sum_{i=1}^n\mathbb{E}\Big[\mu_i(D_i, \mathbf{D}_{-i}) - \mu_i(\mathbf{0}) | D_i = 0\Big]\Big|. 
\end{aligned} 
$$ 
The last inequality follows from the triangular inequality. By Assumption \ref{ass:exposure_restriction} (ii), we write 
\begin{equation}\label{eq:bias_term1}
\small 
\begin{aligned} 
\sup_{\mu(1, \cdot) \in \otimes_{i=1}^n \mathcal{M}_{1, i}} \Big|
\frac{1}{n} \sum_{i=1}^n \mathbb{E}\Big[\mu_i(D_i, \mathbf{D}_{-i}) - \mu_i(\mathbf{1}) | D_i = 1\Big]\Big| 
& \le \frac{\bar{\phi}_n}{n} \sum_{i=1}^n \frac{1}{|\mathcal{N}_i|}\mathbb{E}\Big[\sum_{k \in \mathcal{N}_i} (1 - D_k )\Big| | D_i = 1\Big]  \\ 
&= \frac{\bar{\phi}_n}{n} \sum_{i=1}^n \frac{1}{|\mathcal{N}_i|}  \mathbb{E}\Big[|\mathcal{N}_i| - \sum_{k \in \mathcal{N}_i} D_k  \Big| D_i = 1\Big] \\ 
&= \frac{\bar{\phi}_n}{2 n} \sum_{i=1}^n \frac{1}{|\mathcal{N}_i|}   \Big|j \in \mathcal{N}_i: c(i) \neq c(j) \Big|, 
\end{aligned} 
\end{equation}
where the last inequality follows from the fact that each unit not assigned to the same cluster of $i$ has treatment probability equal to $1/2$ under Assumption \ref{ass:clusters}. The same reasoning applies to $ \sup_{\mu(0, \cdot) \in \otimes_{i=1}^n \mathcal{M}_{0, i}} \Big| \frac{1}{n} \sum_{i=1}^n\mathbb{E}\Big[\mu_i(D_i, \mathbf{D}_{-i}) - \mu_i(\mathbf{0}) | D_i = 0\Big]\Big|$.

\paragraph{Lower bound ($\ge$ case)} To prove the above upper bound is achievable, we construct potential outcomes as follows:
\begin{equation}\label{eq:worst_case}
\small 
\begin{aligned} 
\mu_i(1, \mathbf{d}_{-i}) = \mu_i(0, \mathbf{d}_{-i}) = \left\{\begin{array}{ll}
\psi_R -  \bar{\phi}_n\frac{1}{|\mathcal{N}_i|} \sum_{k \in \mathcal{N}_i}  (1 - \mathbf{d}_k) & \text{if }|\psi_R| \ge |\psi_L|,\\
\psi_L + \bar{\phi}_n \frac{1}{|\mathcal{N}_i|}\sum_{k \in \mathcal{N}_i} \mathbf{d}_k & \text{otherwise}
\end{array}\right.
\end{aligned}
\end{equation}
Since $0\le \bar{\phi}_n\le \psi_R - \psi_L$, for $d \in \{0, 1\}$,
$\mu_i(d, \cdot)\in [\psi_L, \psi_R].$
Thus, $\mu_i(d, \cdot)\in \mathcal{M}_i$. Note that in both cases, $\mu_i(1, \mathbf{d}_{-i}) - \mu_i(1, \mathbf{1})$ and $\mu_i(0, \mathbf{d}_{-i}) - \mu_i(0, \mathbf{0})$ have different signs. Thus, we have

$$
\small 
\begin{aligned} 
\sup_{\mu \in \mathcal{M}} \Big| \mathbb{E}_\mu[\hat{\tau}_n] - \tau_{n,\mu}\Big| & = 
\sup_{\mu \in \mathcal{M}} \Big|
\frac{1}{n} \sum_{i=1}^n \mathbb{E}\Big[\mu_i(D_i, \mathbf{D}_{-i}) - \mu_i(\mathbf{1}) | D_i = 1\Big] - \mathbb{E}\Big[\mu_i(D_i, \mathbf{D}_{-i}) - \mu_i(\mathbf{0}) | D_i = 0\Big]\Big| \\ 
&\ge  \Big|
\frac{\bar{\phi}_n}{n} \sum_{i=1}^n \frac{1}{|\mathcal{N}_i|} \mathbb{E}\Big[\sum_{k \in \mathcal{N}_i} (D_k - 1) | D_i = 1\Big] -  \frac{\bar{\phi}_n}{n} \sum_{i=1}^n \frac{1}{|\mathcal{N}_i|} \mathbb{E}\Big[\sum_{k \in \mathcal{N}_i} D_k | D_i = 0\Big]\Big| \\
&= \frac{\bar{\phi}_n}{n} \sum_{i=1}^n \frac{1}{|\mathcal{N}_i|} \mathbb{E}\Big[\sum_{k \in \mathcal{N}_i} (1 - D_k) | D_i = 1\Big] +  \frac{\bar{\phi}_n}{n} \sum_{i=1}^n \frac{1}{|\mathcal{N}_i|} \mathbb{E}\Big[\sum_{k \in \mathcal{N}_i} D_k | D_i = 0\Big], 
\end{aligned} 
$$ 
 where we inverted the sign in the last expression using the absolute value in the second expression. 
The proof completes following the same steps for the upper bound.

 \subsection{Proof of Lemma \ref{lem:2}}  \label{proof:lem:2}
Let $C_i = \{c(v): v\in B_i\}$, and $G_i = \Big\{g \in \{1, \cdots, n\}: \mathcal{N}_g \cap B_i \neq \emptyset \Big\}$. Observe that we can write Lemma \ref{lem:2} equivalently as stating that 
$$
\small 
\begin{aligned} 
 & \mu_i(D_i, \mathbf{D}_{-i})\Big[2D_i - 1\Big] \perp  \mu_j(D_j, \mathbf{D}_{-j})\Big[2D_j - 1\Big], \quad \forall j \not \in \{B_i \cup G_i\}. 
 \end{aligned}
$$
We, therefore, prove the above expression. 
By definition of $B_i$, we have (a) $c(v') \in C_i$ iff $v'\in B_i$, and (b) $C_i = \{c(v): v\in \{i\}\cup \mathcal{N}_i\}$. By the local interference assumption (Assumption \ref{ass:first_order}), $\mu_i(D_i, \mathbf{D}_{-i})[2D_i - 1]$ is a function of $\{D_v: v\in \{i\}\cup \mathcal{N}_{i}\} = \{\tilde{D}_{c}: c\in C_i\}$. For any $j$ and $v \in B_j$, property (b) implies there exists $v'\in \{j\}\cup \mathcal{N}_{j}$ such that $c(v') = c(v)$. When $j\not\in \{B_i \cup G_i\}$, $(\{j\}\cup \mathcal{N}_{j}\cap B_i = \emptyset$) and hence $v'\not\in B_i$. Property (a) then implies $c(v) = c(v')\not\in C_i$. Thus, $C_i \cap C_j = \emptyset$. Since $\{\tilde{D}_c: c\in \mathcal{C}_{ n}\}$ are independent (Assumption \ref{ass:clusters}), the desired result follows. 

 \subsection{Proof of Lemma \ref{lem:rel}} \label{proof:psi}

 We consider the case where two units are in the same or different clusters separately. We will refer to $\mu_i(D_i, \mathbf{D}_{-i})$ as $\mu_i(\mathbf{D})$ for notational convenience. 

 \paragraph{Unit $i$ and $j$ are in different clusters}
  First, we write 
\begin{equation} \label{eqn:helper_psi}
 \small 
 \begin{aligned} 
 \mathrm{Cov}\Big(\mu_i(\mathbf{D})\Big[2D_i - 1\Big], \mu_j(\mathbf{D})\Big[2D_j - 1\Big]\Big)  &=  \mathrm{Cov}\Big(\mu_i(1, \mathbf{D}_{-i}) D_i, \mu_j(1, \mathbf{D}_{-j}) D_j\Big) \\ & -  \mathrm{Cov}\Big(\mu_i(1, \mathbf{D}_{-i}) D_i, \mu_j(0, \mathbf{D}_{-j}) (1 - D_j)\Big) \\ &- \mathrm{Cov}\Big(\mu_i(0, \mathbf{D}_{-i}) (1 - D_i), \mu_j(1, \mathbf{D}_{-j}) D_j\Big) \\ & + \mathrm{Cov}\Big(\mu_i(0, \mathbf{D}_{-i})(1 - D_i), \mu_j(0, \mathbf{D}_{-j})(1 - D_j)\Big).
 \end{aligned} 
 \end{equation} 
 We bound the first component in the right-hand side of Equation \eqref{eqn:helper_psi}.
 
 \paragraph{First component in Equation \eqref{eqn:helper_psi}} We can write 
\begin{equation} \label{eqn:helper_psi2} 
 \small 
 \begin{aligned}
 \mathrm{Cov}\Big(\mu_i(1, \mathbf{D}_{-i}) D_i, \mu_j(1, \mathbf{D}_{-j}) D_j\Big) & = \mathrm{Cov}\Big(\mu_i(1, \mathbf{D}_{-i}) D_i - \mu_i(1, \mathbf{1})D_i, \mu_j(1, \mathbf{D}_{-j}) D_j\Big) \\ &+ \mathrm{Cov}\Big(\mu_i(1, \mathbf{1})D_i, \mu_j(1, \mathbf{D}_{-j}) D_j - \mu_j(1, \mathbf{1}) D_j\Big) \\ 
 &+ \mathrm{Cov}\Big(\mu_i(1, \mathbf{1})D_i, \mu_j(1, \mathbf{1}) D_j\Big). 
 \end{aligned} 
 \end{equation} 
 We can now study each component separately. Since $\mu_i(\cdot)$ is uniformly bounded, it is easy to show that 
 $$
 \small 
 \begin{aligned} 
 \mathrm{Cov}\Big(\mu_i(1, \mathbf{D}_{-i}) D_i - \mu_i(1, \mathbf{1})D_i, \mu_j(1, \mathbf{D}_{-j}) D_j\Big) & = \mathcal{O}\left(|\mu_i(1, \mathbf{D}_{-i}) - \mu_i(1, \mathbf{1})| \right) = \mathcal{O}(\bar{\phi}_n),
 \end{aligned} 
 $$ 
 where the last line is implied by Assumption \ref{ass:exposure_restriction}. By symmetry, we also have
 $$
 \small
 \mathrm{Cov}\Big(\mu_i(1, \mathbf{1})D_i, \mu_j(1, \mathbf{D}_{-j}) D_j - \mu_j(1, \mathbf{1}) D_j\Big) = 
\mathcal{O}\left(\bar{\phi}_n\right),
 $$
 where the last line is obtained by swapping $i$ and $j$. 
 Finally, for the third component in Equation \eqref{eqn:helper_psi2} note that since $c(i) \neq c(j)$, we have that 
 $
 \mathrm{Cov}\Big(\mu_i(1, \mathbf{1})D_i, \mu_j(1, \mathbf{1}) D_j\Big) = 0, 
 $ 
 by design. Putting pieces together, we obtain that
 $
 \mathrm{Cov}\Big(\mu_i(1, \mathbf{D}_{-i}) D_i, \mu_j(1, \mathbf{D}_{-j}) D_j\Big) = \mathcal{O}\left(\bar{\phi}_n\right). 
 $

 \paragraph{Remaining components in Equation \eqref{eqn:helper_psi}}
 We now move to the second component in Equation \eqref{eqn:helper_psi}.  We have 
\begin{equation} \label{eqn:helper_psi3} 
 \small 
 \begin{aligned}
 \mathrm{Cov}\Big(\mu_i(1, \mathbf{D}_{-i}) D_i, \mu_j(0, \mathbf{D}_{-j}) (1 - D_j)\Big) &=  \mathrm{Cov}\Big(\mu_i(1, \mathbf{D}_{-i}) D_i - \mu_i(1, \mathbf{1})D_i, \mu_j(0, \mathbf{D}_{-j}) (1 - D_j)\Big) \\ &+ \mathrm{Cov}\Big(\mu_i(1, \mathbf{1})D_i, \mu_j(0, \mathbf{D}_{-j}) (1 - D_j) - \mu_j(0, \mathbf{0}) (1 - D_j)\Big) \\ 
 &+ \mathrm{Cov}\Big(\mu_i(1, \mathbf{1})D_i, \mu_j(0, \mathbf{0}) (1 - D_j)\Big). 
 \end{aligned} 
 \end{equation} 

 The same reasoning used to bound the first component in Equation \eqref{eqn:helper_psi2} directly applies also to Equation \eqref{eqn:helper_psi3}. 
 
 Finally, we can use similar arguments for the third and fourth components in Equation \eqref{eqn:helper_psi}, whose proof follows verbatim as the proof to bound the components in  Equation \eqref{eqn:helper_psi2}.  
 
 \paragraph{Collecting bounds} Collecting the bounds, it is easy to show that for $(i,j)$ in different clusters
 $
  \mathrm{Cov}\Big(\mu_i(\mathbf{D})\Big[2D_i - 1\Big], \mu_j(\mathbf{D})\Big[2D_j - 1\Big]\Big)  =  \mathcal{O}\left(\bar{\phi}_n\right), 
 $
 completing the proof for when $i$, $j$ are in two different clusters.

 \paragraph{Unit $i$ and $j$ are in the same cluster} If $i$ and $j$ are in the same cluster we can invoke Equation \eqref{eqn:helper_psi} similarly as above, with the only difference that for $c(i) = c(j)$,
\begin{equation} \label{eqn:22} 
\small 
 \begin{aligned}
 \mathrm{Cov}\Big(\mu_i(1, \mathbf{1})D_i, \mu_j(1, \mathbf{1}) D_j \Big) &= \frac{1}{4} \mu_i(1, \mathbf{1}) \mu_j(1, \mathbf{1}), \\ \mathrm{Cov}\Big(\mu_i(0, \mathbf{0})(1 - D_i), \mu_j(0, \mathbf{0}) (1 - D_j) \Big) & = \frac{1}{4} \mu_i(0, \mathbf{0}) \mu_j(0, \mathbf{0}), \\
 \mathrm{Cov}\Big(\mu_i(0, \mathbf{0})(1 - D_i), \mu_j(1, \mathbf{1}) D_j \Big) & =  -\frac{1}{4} \mu_j(1, \mathbf{1}) \mu_i(0, \mathbf{0}), \\
  \mathrm{Cov}\Big(\mu_i(1, \mathbf{1})D_i, \mu_j(0, \mathbf{0}) (1 - D_j) \Big) & =  -\frac{1}{4} \mu_j(0, \mathbf{0}) \mu_i(1, \mathbf{1}).
 \end{aligned}
\end{equation} 
Following the same steps as for the case where $i$, $j$ are in different clusters, accounting for Equation \eqref{eqn:22}, the proof completes.

\subsection{Proof of Lemma \ref{prop:main}} \label{app:propmain} 

We organize the proof as follows. First, we decompose the variance into sums of covariances in different sets. We then study the dominant components.

\paragraph{Variance decomposition into multiple components}
As a first step, we characterize the size of the set of units for which the covariances are not zero, but of order at most $\mathcal{O}(\bar{\phi}_n)$. 
Thus, 
$$
\footnotesize  
\begin{aligned}
& \mathbb{E}_\mu\Big[\Big(\hat{\tau}_n(\mathcal{C}_{ n}) - \mathbb{E}_\mu[\hat{\tau}_n(\mathcal{C}_{ n})]\Big)^2\Big] 
= \frac{4}{n^2} \sum_{i,j} \mathrm{Cov}\Big(\mu_i(D_i, \mathbf{D}_{-i})[2 D_i - 1], \mu_j(D_j, \mathbf{D}_{-j})[2 D_j - 1]\Big)\\
& = \frac{4}{n^2} \sum_{i,j: c(i) = c(j)} \mathrm{Cov}\Big(\mu_i(D_i, \mathbf{D}_{-i})[2 D_i - 1], \mu_j(D_j, \mathbf{D}_{-j})[2 D_j - 1]\Big) +  \mathcal{O}\Big(\frac{\bar{\phi}_n}{n^2}\cdot \sum_{i,j: c(i)\neq c(j)}I(j\in B_i\cup G_i)\Big)\qquad \text{(by \ref{lem:2})}\\
& =  \frac{4}{n^2} \sum_{i,j: c(i) = c(j)} \frac{1}{4}(\mu_i(\mathbf{1}) + \mu_i(\mathbf{0}))(\mu_j(\mathbf{1}) + \mu_j(\mathbf{0})) + \mathcal{O}\left(\frac{\bar{\phi}_n}{n^2}\cdot \sum_{i,j: c(i)= c(j)}1\right) + \\
& \qquad \mathcal{O}\left(\frac{\bar{\phi}_n}{n^2}\cdot \sum_{i,j: c(i)\neq c(j)}I(j\in B_i\cup G_i)\right)\qquad \text{(by \ref{lem:rel})}\\
& = \frac{4}{n^2} \sum_{i,j: c(i) = c(j)} \frac{1}{4}(\mu_i(\mathbf{1}) + \mu_i(\mathbf{0}))(\mu_j(\mathbf{1}) + \mu_j(\mathbf{0})) + \\
& \qquad \mathcal{O}\left(\frac{\bar{\phi}_n}{n^2}\cdot \sum_{i,j}I(j\in B_i\cup G_i)\right)\qquad (\text{because }c(i) = c(j) \Longrightarrow j\in B_i\cup G_i)\\
& = \frac{1}{n^2}\sum_{k=1}^{K_n}\left(\sum_{i: c(i)  = k}(\mu_i(\mathbf{1}) + \mu_i(\mathbf{0}))\right)^2 +  \mathcal{O}\left(\frac{\bar{\phi}_n}{n^2}\cdot \sum_{i}|B_i\cup G_i|\right).
\end{aligned}
$$ 
\paragraph{Size of the set $B_i \cup G_i$} Next, we study the size of $B_i\cup G_i$. By definition, $j\in G_i$ only if $j$ is the neighbor of a unit in $B_i$. Thus,
$ 
|B_i\cup G_i| \le \sum_{\ell: \ell\in B_i}(|\mathcal{N}_\ell| + 1) = \sum_{\ell: c(\ell)\in C_i}(|\mathcal{N}_\ell| + 1),
$ 
where $C_i = \{c(v): v\in \{i\}\cup \mathcal{N}_i\}$ is defined in the proof of Lemma \ref{lem:2}. Thus, 
$$
\small 
\begin{aligned} 
\sum_{i}|B_i\cup G_i|\le \sum_{i}\sum_{\ell: c(\ell)\in C_i}(|\mathcal{N}_\ell| + 1) = \sum_{k=1}^{K_n}|\{i: k\in C_i\}|\sum_{\ell: c(\ell) = k}(|\mathcal{N}_\ell| + 1).
\end{aligned}
$$
By definition, $k\in C_i$ only if $c(i) = k$ or $i$ is the neighbor of a unit in cluster $k$. Thus, 
$|\{i: k\in C_i\}|\le \sum_{\ell: c(\ell) = k}(|\mathcal{N}_\ell| + 1).$ 
As a result,
$\sum_{i}|B_i\cup G_i|\le \sum_{k=1}^{K_n}\Big(\sum_{\ell: c(\ell) = k}(|\mathcal{N}_\ell| + 1)\Big)^2.$ 
By Cauchy-Schwarz inequality and Assumption \ref{ass:clusters}, 
$$
\small 
\begin{aligned} 
\sum_{i}|B_i\cup G_i|\le \sum_{k=1}^{K_n}n_k \sum_{\ell: c(\ell) = k}(|\mathcal{N}_\ell| + 1)^2\le \bar{\gamma}\frac{n}{K_n}\sum_{i}(|\mathcal{N}_i| + 1)^2.
\end{aligned} 
$$
On the other hand, $|B_i \cup G_i|\le n$ implies that
$\sum_{i}|B_i\cup G_i|\le n^2.$ 
\paragraph{Collecting all terms together} Putting two pieces together, we obtain that 
\begin{align}
\small 
&\mathbb{E}_\mu\Big[\Big(\hat{\tau}_n(\mathcal{C}_{ n}) - \mathbb{E}_\mu[\hat{\tau}_n(\mathcal{C}_{ n})]\Big)^2\Big] = \frac{1}{n^2}\sum_{k=1}^{K_n}\left(\sum_{i: c(i)  = k}(\mu_i(\mathbf{1}) + \mu_i(\mathbf{0}))\right)^2\nonumber\\
& \quad +  \mathcal{O}(\bar{\phi}_n)\cdot \min\left\{\frac{1}{nK_n}\sum_{i}(|\mathcal{N}_i| + 1)^2, 1\right\}\label{eq:sum_square_degree}. 
\end{align}
Since $\mathcal{N}_{n, \max} = \max\{\max_{i}|\mathcal{N}_i|, 1\}$  , 
\begin{align}
& \mathbb{E}_\mu\Big[\Big(\hat{\tau}_n(\mathcal{C}_{ n}) - \mathbb{E}_\mu[\hat{\tau}_n(\mathcal{C}_{ n})]\Big)^2\Big]\\
& = \frac{1}{n^2}\sum_{k=1}^{K_n}\left(\sum_{i: c(i)  = k}(\mu_i(\mathbf{1}) + \mu_i(\mathbf{0}))\right)^2 +  \mathcal{O}\left(\bar{\phi}_n \cdot \min\left\{\mathcal{N}_{n, \max}^2/K_n, 1\right\}\right)\nonumber\\
& = \frac{1}{n^2}\sum_{k=1}^{K_n}\left(\sum_{i: c(i)  = k}(\mu_i(\mathbf{1}) + \mu_i(\mathbf{0}))\right)^2 +  \mathcal{O}\left(\bar{\phi}_n/K_n \cdot \min\left\{\mathcal{N}_{n, \max}^2, n\right\}\right).
\end{align}
By Assumption \ref{defn:local} and the fact that $K_n \le n$, 
\begin{align}
&\mathbb{E}_\mu\Big[\Big(\hat{\tau}_n(\mathcal{C}_{ n}) - \mathbb{E}_\mu[\hat{\tau}_n(\mathcal{C}_{ n})]\Big)^2\Big]= \frac{1}{n^2}\sum_{k=1}^{K_n}\left(\sum_{i: c(i)  = k}(\mu_i(\mathbf{1}) + \mu_i(\mathbf{0}))\right)^2 + \mathcal{O}(\varepsilon/K_n)\label{eq:worst_case_variance_max_degree}.
\end{align}
By Assumption \ref{ass:exposure_restriction} (i), 
$\mathbb{E}_\mu\Big[\Big(\hat{\tau}_n(\mathcal{C}_{ n}) - \mathbb{E}_\mu[\hat{\tau}_n(\mathcal{C}_{ n})]\Big)^2\Big]\le \frac{1}{n^2}\sum_{k=1}^{K_n}n_k^2 \cdot \bar{\psi} + \mathcal{O}(\varepsilon/K_n).$ 
\paragraph{Conclusions by showing that the upper bound is achievable} Consider the potential outcomes constructed in \eqref{eq:worst_case}. If $|\psi_R|\ge |\psi_L|$, then it must be that $\psi_R > 0$ and
$\mu_i(\mathbf{1}) + \mu_i(\mathbf{0}) = 2\psi_R - \bar{\phi}_n = \bar{\psi}^{1/2} - \bar{\phi}_n.$ 
Similarly, when $|\psi_L| > |\psi_R|$, it must be that $\psi_L < 0$ and 
$\mu_i(\mathbf{1}) + \mu_i(\mathbf{0}) = 2\psi_L + \bar{\phi}_n = -\bar{\psi}^{1/2} + \bar{\phi}_n.$ 
$$  
\small 
\begin{aligned}
& \frac{1}{n^2}\sum_{k=1}^{K_n}\left(\sum_{i: c(i)  = k}(\mu_i(\mathbf{1}) + \mu_i(\mathbf{0}))\right)^2 = \frac{1}{n^2}\sum_{k=1}^{K_n}n_k^2 \cdot (\bar{\psi}^{1/2} - \bar{\phi}_n)^2\ge \frac{1}{n^2}\sum_{k=1}^{K_n}n_k^2 \cdot \bar{\psi} - \frac{2\bar{\psi}^{1/2} \bar{\phi}_n\bar{\gamma}}{K_n}, \label{eq:worst_case_variance_last_step}
\end{aligned}
$$ 
where the last inequality invokes Assumption \ref{ass:clusters} under which 
$\frac{1}{n^2}\sum_{k=1}^{K_n}n_k^2\le \frac{1}{n^2}\sum_{k=1}^{K_n}n_k \cdot \bar{\gamma}\frac{n}{K_n} = \frac{\bar{\gamma}}{K_n}.$ 
Note that Assumption \ref{defn:local} implies $|\bar{\phi}_n| \le \mathcal{O}(\varepsilon)$, we obtain that
$$
\small 
\begin{aligned} \sup_{\mu\in \mathcal{M}}\mathbb{E}_\mu\Big[\Big(\hat{\tau}_n(\mathcal{C}_{ n}) - \mathbb{E}_\mu[\hat{\tau}_n(\mathcal{C}_{ n})]\Big)^2\Big]\ge \frac{1}{n^2}\sum_{k=1}^{K_n}n_k^2 \cdot \bar{\psi} + \mathcal{O}(\varepsilon/K_n).
\end{aligned} 
$$ 
Finally, by Cauchy-Schwarz inequality, 
$
\frac{1}{n^2} \sum_{k=1}^{K_n} n_k^2\ge 1/K_n.
$
The proof completes.

\subsection{Proof of Theorem \ref{thm:Bn_Bnstar}}\label{app:worst_case_mse}

Clearly from Lemma \ref{lem:worst_case_bias} and \ref{prop:main} we have 
$
\sup_{\mu \in \mathcal{M}} \mathbb{E}_\mu\Big[\Big(\hat{\tau}_n(\mathcal{C}_{ n}) - \mathbb{E}_\mu[\hat{\tau}_n(\mathcal{C}_{ n})]\Big)^2\Big] + \lambda  \Big(\tau_{n,\mu} - \mathbb{E}_\mu[\hat{\tau}_n(\mathcal{C}_{ n})]\Big)^2 \le \mathcal{B}_n^*(\mathcal{C}_n, \lambda)(1 + \mathcal{O}(\varepsilon)).
$
We now have to show the other side of the inequality. 
From the proofs of Lemma \ref{lem:worst_case_bias} and Lemma \ref{lem:rel}, we observe that the worst-case bias and variance are both achieved, up to $1 + \mathcal{O}(\varepsilon)$ factors, by the potential outcomes defined in Equation \eqref{eq:worst_case}. The result is then proved.

\subsection{Proof of Theorem \ref{thm:thumb}} \label{app:thumb}

We can write, 
$\mathcal{B}_n^*(\mathcal{C}_{\mathrm{B},n}, \lambda) = \lambda\bar{\phi}_n^2  + \frac{1}{n} \bar{\psi},$ 
and 
$\mathcal{B}_n^*(\mathcal{C}_{ n}, \lambda) = \lambda \bar{\phi}_n^2 b_n(\mathcal{C}_{ n})^2 + \left(\frac{1}{n^2} \sum_{k=1}^{K_n} n_k^2\right)\bar{\psi}.$

Finally, note that $\frac{1}{n^2} \sum_{k=1}^{K_n} n_k^2 = (\frac{1}{K_n} \sum_{k=1}^{K_n} \gamma_k^2)/K_n = \underline{\gamma}/K_n$. The proof completes after rearrangement.

\subsection{Proof of Theorem \ref{thm:opt1}} \label{proof:thm:opt1}
Fix \(K\ge2\) and \(c\in\mathfrak C_K\). Let
$ 
\mathbf X_c(K)=\mathbf M_c(K)\mathbf M_c(K)^\top .
$ 
Then
$ 
X_{c,ij}(K)=1\{c(i)=c(j)\}.
$ 
First,
$ 
\operatorname{tr}\!\left(
\one_n\one_n^\top \mathbf M_c(K)\mathbf M_c(K)^\top
\right)
=
\one_n^\top \mathbf M_c(K)\mathbf M_c(K)^\top\one_n .
$ 
Because the \(k\)-th component of \(\mathbf M_c(K)^\top\one_n\) is the size
\(n_k\) of cluster \(k\), it follows that
$ 
\one_n^\top \mathbf M_c(K)\mathbf M_c(K)^\top\one_n
=
\left\|\mathbf M_c(K)^\top\one_n\right\|_2^2
=
\sum_{k=1}^K n_k^2 .
$ 
Thus the first term in the displayed program equals the variance component in
\(R_n(\mathcal C_n;\xi_n)\).

Second,
$ 
\operatorname{tr}\!\left(
\mathbf L^\top
\left[
\one_n\one_n^\top
-
\mathbf M_c(K)\mathbf M_c(K)^\top
\right]
\right) =
\sum_{i=1}^n\sum_{j=1}^n
L_{ij}\{1-X_{c,ij}(K)\}  =
\sum_{i=1}^n\sum_{j=1}^n
L_{ij}1\{c(i)\ne c(j)\}.
$ 
By the definition \(\mathbf L=\mathbf V^{-1}\mathbf A\), this equals
$ 
\sum_{i=1}^n
\frac{1}{|\mathcal N_i|}
\left|
\mathcal N_i\cap\{j:c(i)\ne c(j)\}
\right|.
$ 
Therefore
$ 
\frac1n
\operatorname{tr}\!\left(
\mathbf L^\top
\left[
\one_n\one_n^\top
-
\mathbf M_c(K)\mathbf M_c(K)^\top
\right]
\right)
=
b_n(\mathcal C_n).
$ 

Finally, since the lower-right entry of the matrix
$ 
\begin{pmatrix}
t & z\\
z & 1
\end{pmatrix}
$ 
is positive, the Schur complement gives
$ 
\begin{pmatrix}
t & z\\
z & 1
\end{pmatrix}
\succeq0
\quad\Longleftrightarrow\quad
t\ge z^2 .
$ 
For any fixed clustering \(c\), the constraint defining \(z\) sets
\(z=b_n(\mathcal C_n)\). Since \(\xi_n\ge0\), the objective is minimized at
\(t=z^2=b_n(\mathcal C_n)^2\) whenever \(\xi_n<\infty\). If \(\xi_n=\infty\), the
choice of \(t\) is irrelevant and the set of minimizing clusterings is unchanged.

Consequently, the displayed program has the same value, for each
\(c\in\mathfrak C_K\), as
$ 
\frac{\xi_n}{n^2}\sum_{k=1}^K n_k^2
+
b_n(\mathcal C_n)^2
=
R_n(\mathcal C_n;\xi_n).
$

\subsection{Proof of Theorem \ref{thm:observable_approximation_error}} \label{proof:thm:observable_approximation_error}
Fix \(K\in\{\underline K,\ldots,\bar K\}\) and consider any clustering
\(\mathcal C_n\) with exactly \(K\) nonempty clusters satisfying $n_k \le \bar{\gamma} n/K$. Its cluster matrix
$ 
\mathbf X_c(K)=\mathbf M_c(K)\mathbf M_c(K)^\top
$ 
is feasible for the SDP relaxation in Equation \eqref{eq:SDP}. Indeed,
$ 
\mathbf X_c(K)\succeq0,X_{c,ii}(K)=1, 0\le X_{c,ij}(K)\le1,
$ 
and
$ 
\operatorname{tr}\!\left(\one_n\one_n^\top\mathbf X_c(K)\right)
=
\sum_{k=1}^K n_k^2 .
$ 
Since the clusters are nonempty and satisfy $n_k \le \bar{\gamma} n/K$,
$ 
\frac{n^2}{K}
\le
\sum_{k=1}^K n_k^2
\le
\bar{\gamma} n^2/K.
$ 
Set
$ 
z=
\frac1n
\operatorname{tr}\!\left(
\mathbf L^\top
\left[
\one_n\one_n^\top-\mathbf X_c(K)
\right]
\right), \quad
t=z^2.
$ 
Then
$ 
\begin{pmatrix}
t & z\\
z & 1
\end{pmatrix}
\succeq0,
$ 
so all constraints in Equation \eqref{eq:SDP} are satisfied. By the proof of
Theorem \ref{thm:opt1}, the corresponding SDP objective equals
$ 
R_n(\mathcal C_n;\xi_n).
$ 
Therefore, because the SDP relaxation minimizes over a feasible set that contains
all exact \(K\)-cluster matrices,
$ 
\min_{K\in\{\underline K,\ldots,\bar K\}}
\underline R_{n,K}^{\mathrm{sdp}}(\xi_n)
\le
\min_{\mathcal C_n \in \bar{\mathfrak C}}
R_n(\mathcal C_n;\xi_n).
$ This inequality is true whether the last constraint in Equation \eqref{eq:SDP} is or is not imposed for optimization in Algorithm \ref{alg:1} (since not imposing such a constraint would make the solution smaller $\min_{K\in\{\underline K,\ldots,\bar K\}}
\underline R_{n,K}^{\mathrm{sdp}}(\xi_n)$). 
Dividing \(R_n(\mathcal C_n^\star;\xi_n)\) by both sides of the last display
yields
\[
\frac{
R_n(\mathcal C_n^\star;\xi_n)
}{
\min_{\mathcal C_n \in \bar{\mathfrak C}}
R_n(\mathcal C_n;\xi_n)
}
\le
\frac{
R_n(\mathcal C_n^\star;\xi_n)
}{
\min_{K\in\{\underline K,\ldots,\bar K\}}
\underline R_{n,K}^{\mathrm{sdp}}(\xi_n)
}. 
\]
Finally, because
$ 
\mathcal B_n^*(\mathcal C_n,\lambda)
=
\lambda \bar{\phi}_n^2 R_n(\mathcal C_n;\xi_n)
$ 
the same relative bound holds for
\(\mathcal B_n^*(\mathcal C_n,\lambda)\). This proves the result.

\subsection{Proof of Theorem \ref{thm:inference1}} \label{proof:lem:inference1}

\begin{lem}\label{thm:CBA_CLT}
    Let $X_1, \ldots, X_n$ be a set of random variables whose dependency graph has degree $m_i$ for variable $X_i$, with the dependency graph as denoted in Theorem 3.5 in \cite{ross2011fundamentals}. 
    Let $m_{\max} = \max_{i}m_i + 1$. Assume that 
    $\max_{i}|X_i|\le B,$ 
    for some constant $B$ that does not depend on $n$, and
    $\left|\frac{1}{n} \sum_{i=1}^{n}\E[X_i]\right|\le b_n, \quad \Var\lb \frac{1}{n} \sum_{i=1}^{n}X_i\rb \le \sigma_n^2,$ 
    for some deterministic sequences $b_n, \sigma_n$ that depend on $n$. Further assume that, as $n\rightarrow \infty$,
    \begin{equation}
      \label{eq:rate_mmax}
      \frac{1}{\sigma_n^3} \frac{m_{\max}^2}{n^2} + \frac{1}{\sigma_n^2} \frac{m_{\max}^{3/2}}{n^{3/2}} = o(1).
    \end{equation}
    Then
    $\liminf_{n\rightarrow \infty}\P\lb \left|\frac{1}{n} \sum_{i=1}^{n}X_i\right|\le \chi(b_n, \sigma_n)\rb\ge 1- \alpha.$ 
  \end{lem}

\begin{proof} 
Let
    $S_n = \frac{1}{n} \sum_{i=1}^{n}X_i, \quad \nu_n^2 = \Var(S_n), \quad \mu_n = \mathbb{E}[S_n].$ 

First, note that if $\nu_n^2 = 0$, the result trivially holds, so that we can consider $\nu_n^2 > 0$. From Theorem 3.5 in \cite{ross2011fundamentals}, 
  $d(\frac{S_n - \mu_n}{\nu_n}, N(0, 1)) = O\lb \frac{m_{\max}^2}{\nu_n^3 n^2}  + \frac{ m_{\max}^{3/2}}{\nu_n^2 n^{3/2}}\rb  ,$ 
  where $d(\cdot)$ is the Wasserstein distance. 
  
  Thus, by standard bounds on the Kolmogorov distance using the Wasserstein distance \citep[e.g.][]{gaunt2023bounding},  we can write 
  $$ 
  \small 
  \begin{aligned}
    \P\lb |S_n| > \chi(|\mu_n|, \nu_n)\rb &\le \P\lb |N(\mu_n, \nu_n^2)| > \chi(|\mu_n|, \nu_n)\rb + O\lb \Big(\frac{m_{\max}^2}{\nu_n^3 n^2}  + \frac{ m_{\max}^{3/2}}{\nu_n^2 n^{3/2}}\Big)^{1/2} \rb \\ &= \alpha + O\lb \Big(\frac{m_{\max}^2}{\nu_n^3 n^2}  + \frac{ m_{\max}^{3/2}}{\nu_n^2 n^{3/2} }\Big)^{1/2}\rb.
  \end{aligned}
  $$ 
  By Lemma \ref{lem:monotonicity},
  \begin{equation}
  \small 
  \begin{aligned} 
    \label{eq:large_nu}
    \P\lb |S_n| > \chi(b_n, \sigma_n)\rb \le  \alpha + O\lb \Big(\frac{m_{\max}^2}{\nu_n^3 n^2 }  + \frac{ m_{\max}^{3/2}}{\nu_n^2 n^{3/2} }\Big)^{1/2} \rb.
    \end{aligned} 
  \end{equation}
  On the other hand, by triangle inequality  
  $$ 
  \small 
  \begin{aligned} 
  \P\lb |S_n| > \chi(b_n, \sigma_n)\rb\le \P\lb |S_n - \mu_n| > \chi(b_n, \sigma_n) - |\mu_n|\rb\le \P\lb |S_n - \mu_n| > \chi(b_n, \sigma_n) - b_n\rb.
  \end{aligned} 
  $$ 
  By Lemma \ref{lem:upper_lower_bound},
  $\chi(b_n, \sigma_n) - b_n\ge q_{1-\alpha}\sigma_n, $ 
  where $q_{1- \alpha}$ is the $1 - \alpha$ quantile of a standard Normal random variable. 
  By Cherbyshev's inequality,
  \begin{equation}\label{eq:small_nu}
  \small 
  \begin{aligned} 
    \P\lb |S_n| > \chi(b_n, \sigma_n)\rb\le \P\lb |S_n - \mu_n| > q_{1-\alpha}\sigma_n\rb\le \frac{\nu_n^2}{q_{1-\alpha}^2 \sigma_n^2} = O\lb \frac{\nu_n^2}{\sigma_n^2}\rb.
    \end{aligned} 
    \end{equation}
  Fix any deterministic sequence $a_n$. If $\nu_n < a_n \sigma_n$, \eqref{eq:small_nu} implies 
  $\P\lb |S_n| > \chi(b_n, \sigma_n)\rb = O(a_n^2).$ 
  If $\nu_n \ge a_n \sigma_n$, \eqref{eq:large_nu} implies
  $\P\lb |S_n| > \chi(b_n, \sigma_n)\rb \le \alpha +  O\lb\frac{1}{a_n^{3/2}} \Big( \frac{m_{\max}^2}{\sigma_n^3 n^2 }  + \frac{m_{\max}^{3/2}}{\sigma_n^2 n^{3/2} }\Big)^{1/2} \rb.$ 
  Define $c_n = \frac{m_{\max}^2}{\sigma_n^3 n^2}  + \frac{ m_{\max}^{3/2}}{\sigma_n^2 n^{3/2}}$. Equation \eqref{eq:rate_mmax} implies $c_n = o(1)$. Setting $a_n = c_n^{1/4}$ yields
  $\P\lb |S_n| > \chi(b_n, \sigma_n)\rb \le \alpha +  O\lb c_n^{1/2} + c_n^{1/8}\rb = \alpha + o(1).$
\end{proof} 
Define 
  \[B_i = \{v: \text{either }c(v) = c(i) \text{ or } c(v) = c(v')\text{ for some }v'\in \mathcal{N}_i\}, \quad G_i = \{g: \mathcal{N}_g \cap B_i \neq \emptyset\}.\]
Lemma \ref{lem:2} implies
  \begin{equation}\label{eq:m-dependence}
  \small 
  \begin{aligned}
   \mu_j(\mathbf{D}) \indep \mu_i(\mathbf{D}), \quad \forall j \not\in B_i \cup G_i.
   \end{aligned} 
  \end{equation}
  It is easy to see that
  $j\in B_i \cup G_i \Longleftrightarrow i \in B_j \cup G_j.$ 
  Thus, \eqref{eq:m-dependence} defines a undirected dependency graph among $(\mu_1(\mathbf{D}), \ldots, \mu_n(\mathbf{D}))$. Define
  $i\sim j \text{ iff }j\in B_i \cup G_i.$
  Under our neighborhood interference assumption, we have shown in the proof of Lemma \ref{lem:2} that each summand in $\hat{\tau}$ forms a local dependency graph with degree $m_i$ as defined in Lemma \ref{thm:CBA_CLT} such that 
  \begin{align}
    m_i &\le \sum_{\ell: \ell \in B_i}(|\mathcal{N}_\ell| + 1) \le |B_i|\cdot (\mathcal{N}_{n, \max} + 1)\nonumber\\
    & \le \max_{k}n_k (|\mathcal{N}_i| + 1)(\mathcal{N}_{n, \max} + 1) \le 4\bar{\gamma} \mathcal{N}_{n, \max}^2 \frac{n}{K_n}.    \label{eq:mi}
  \end{align}
  In addition, by assumption $\bar{\sigma}_n^2 \ge \frac{\eta + o(1)}{K_n}$ for a constant $\eta > 0$. We now invoke Lemma \ref{thm:CBA_CLT}. 
Note that we can write 
$ 
\hat{\tau}_n - \tau_{n,\mu} = \frac{1}{n} \sum_{i=1}^n X_i
$ 
where $X_i:= 2(2D_i - 1)\mu_i(\mathbf{D}) - (\mu_i(\mathbf{1}) - \mu_i(\mathbf{0}))$ are uniformly bounded random variables forming a dependency graph with maximum degree $m_{\max} = \mathcal{O}(\mathcal{N}_{n, \max}^2 \frac{n}{K_n}$). 
To check that the conditions in Lemma \ref{thm:CBA_CLT} we only need to show that 
$
 \frac{1}{\bar{\sigma}_n^3 n^2 } m_{\max}^2 + \frac{1}{ \bar{\sigma}_n^2 n^{3/2}} m_{\max}^{3/2}  = \mathcal{O} \left( \frac{\mathcal{N}_{n,\max}^4}{K_n^{1/2}} + \frac{\mathcal{N}_{n,\max}^3}{K_n^{1/2}} \right)
$
is $o(1)$. This follows by assumption, completing the proof.

\subsection{Proof of Theorem \ref{thm:worst_case_inference}} \label{proof:thm:worst_case_inference} 

Before proving the theorem we need to establish some properties. 

\subsubsection{Useful properties of $\chi(b, \nu)$}
  Let $\Phi(\cdot)$ denote the CDF of $N(0, 1)$. By definition,
  \begin{equation}
  \small 
  \begin{aligned} 
    \label{eq:chi_def}
    \Phi\lb\frac{\chi(b, \nu) - b}{\nu}\rb - \Phi\lb\frac{-\chi(b, \nu) - b}{\nu}\rb = 1 - \alpha.
    \end{aligned}
  \end{equation}
  Throughout the section we assume $b > 0$, which is true in \eqref{eq:chi}. 
  
  \begin{lem}\label{lem:upper_lower_bound}
    Assume $\alpha \le 0.5$. Let $z_{\beta}$ denote the $\beta$-th quantile of $N(0, 1)$. Then
    $\chi(b, \nu) \in [b + q_{1-\alpha}\nu, b + q_{1-\alpha/2}\nu].$ 
  \end{lem}
  \begin{proof}
    Since $\Phi(0) = 0.5 \ge \alpha$, 
    $\chi(b, \nu)\ge b > 0.$ 
    Thus,
    $\Phi\lb\frac{\chi(b, \nu) - b}{\nu}\rb - \Phi\lb\frac{-\chi(b, \nu) - b}{\nu}\rb\le 1 - \alpha\le \Phi\lb\frac{\chi(b, \nu) - b}{\nu}\rb.$ 
    Using the symmetry of $N(0, 1)$, we have
    $1 - \alpha\le \Phi\lb\frac{\chi(b, \nu) - b}{\nu}\rb\le 1 - \alpha / 2.$ 
    The proof is then completed. 
  \end{proof}

  \begin{lem}\label{lem:monotonicity}
    $\chi(b, \nu)$ is increasing in $b$ and $\nu$. 
  \end{lem}
  \begin{proof}
    Let $\chi_b(b, \nu)$ and $\chi_\nu(b, \nu)$ denote the partial derivatives of $\chi$ with respect to $b$ and $\nu$, respectively. We suppress the dependence on $(b, \nu)$ in $\chi, \chi_b, \chi_\nu$ for notational convenience. Let $\phi(x) =  \Phi'(x) = 1/\sqrt{2\pi} \exp\{-x^2 / 2\}$. Differentiating \eqref{eq:chi_def} with respect to $b$, we obtain that
    \[0 = (\chi_b - 1)\phi\lb \frac{\chi - b}{\nu}\rb + (\chi_b + 1)\phi\lb -\frac{\chi + b}{\nu}\rb = (\chi_b - 1)\phi\lb \frac{\chi - b}{\nu}\rb + (\chi_b + 1)\phi\lb \frac{\chi + b}{\nu}\rb,\]
    where the last line uses the symmetry of $\phi$. As a result,
    \begin{equation}
      \label{eq:chi_b}
      \chi_b = \frac{\phi\lb \frac{\chi - b}{\nu}\rb - \phi\lb \frac{\chi + b}{\nu}\rb}{\phi\lb \frac{\chi - b}{\nu}\rb + \phi\lb \frac{\chi + b}{\nu}\rb}.
    \end{equation}
    Since $\chi \ge b > 0$ and $\phi$ is decreasing on $[0, \infty)$, 
    $\phi\lb \frac{\chi - b}{\nu}\rb > \phi\lb \frac{\chi + b}{\nu}\rb.$ 
    Thus, $\chi_b > 0$, implying that $\chi$ is increasing in $b$.

    Differentiating \eqref{eq:chi_def} with respect to $\nu$, we obtain that
    $$  
    \small 
    \begin{aligned}
      0 &= \frac{\nu \chi_\nu - \chi + b}{\nu^2}\phi\lb \frac{\chi - b}{\nu}\rb + \frac{\nu \chi_\nu - \chi - b}{\nu^2}\phi\lb -\frac{\chi + b}{\nu}\rb\\
      & = \frac{\nu \chi_\nu - \chi + b}{\nu^2}\phi\lb \frac{\chi - b}{\nu}\rb + \frac{\nu \chi_\nu - \chi - b}{\nu^2}\phi\lb \frac{\chi + b}{\nu}\rb,
    \end{aligned}
    $$ 
    where, again, the last line uses the symmetry of $\phi$. This implies 
    \begin{equation}
      \label{eq:chi_sigma}
      \chi_\nu = \frac{\frac{\chi - b}{\nu}\phi\lb \frac{\chi - b}{\nu}\rb + \frac{\chi + b}{\nu}\phi\lb \frac{\chi + b}{\nu}\rb}{\phi\lb \frac{\chi - b}{\nu}\rb + \phi\lb \frac{\chi + b}{\nu}\rb} > 0.
    \end{equation}
    Thus, $\chi$ is increasing in $\nu$.
  \end{proof}

  \begin{lem}\label{lem:lipschitz}
    $\chi(b, \nu)$ is Lipschitz in $b$ and $\nu$. The Lipschitz constant is bounded by $L_\alpha$ for some constant $L_\alpha$ that only depends on $\alpha$.
  \end{lem}
  \begin{proof}
    By \eqref{eq:chi_b}, it is clear that
    $\chi_b(b, \nu)\le 1, \quad \forall b, \nu \ge 0.$ 
    Let $a = (\chi(b, \nu) - b)/\nu, z = 2b / \nu$. Then \eqref{eq:chi_sigma} can be rewritten as
    \[\chi_\nu(b, \nu) = \frac{a\phi(a) + (a + z)\phi(a + z)}{\phi(a) + \phi(a + z)}= a + \frac{z}{1 + \phi(a) / \phi(a+ z)} = a + \frac{z}{1 + \exp\{az + z^2/2\}}.\]
    Using the inequality that $\exp\{y\}\ge 1 + y$, we have
    \[\chi_\nu(b, \nu)\le a + \frac{z}{2 + z^2/2 + az}\le a + \frac{1}{a + 2}\le a + 0.5.\]
    By Lemma \ref{lem:upper_lower_bound}, we conclude that
    $\chi_\nu(b, \nu)\le q_{1-\alpha/2} + 0.5, \quad \forall b, \nu \ge 0.$ 
    By Lemma \ref{lem:monotonicity}, $\chi_b(b, \nu), \chi_\nu(b, \nu) \ge 0$. Thus, the Lipschitz constant is bounded by 
    $\sqrt{\chi_b(b, \nu)^2 + \chi_\nu^2}\le \sqrt{1 + (q_{1-\alpha/2} + 0.5)^2}.$ 
  \end{proof}

  By Lemma \ref{lem:monotonicity},
\begin{equation}\label{eq:length1}
\small 
\begin{aligned} 
  |I_{\ea}(\beta(\mu), \nu(\mu))|\le |I_{\ea}^*|.
  \end{aligned} 
  \end{equation}
In the proof of worst-case bias and variance, we have shown there exists $\mu^*$ such that
\[|\beta(\mu^*)| = \bar{\phi}_n b(\mathcal{C}_n), \quad \sigma^2(\mu^*) = \sigma_{n,*}^2(1 + \mathcal{O}(\varepsilon)).\]
By Lemma \ref{lem:lipschitz}, from a first order Taylor approximation  so that $\sigma_{n,*} \sqrt{(1 + \mathcal{O}(\varepsilon))} = \sigma_{n,*}  + \mathcal{O}(\varepsilon)$
\begin{equation}\label{eq:length2}
\small 
\begin{aligned}
  |I(\mu^{*})|\ge |I_{\ea}^*|\cdot (1 + \mathcal{O}(\varepsilon)).
  \end{aligned} 
\end{equation}
Combining \eqref{eq:length1} and \eqref{eq:length2}, we obtain the following result.

  \subsection{Proof of Theorem \ref{thm:adaptation_regret}}
\label{proof:thm:adaptation_regret}

We  follow a similar proof strategy as in \cite{epanomeritakis2025learning}. 
Fix a feasible clustering \(\mathcal C_n\). 
Because \(R_n(\widetilde{\mathcal C}_n;\xi)\) is affine in \(\xi\) for each
\(\widetilde{\mathcal C}_n\), the function \(R_n^*(\xi)\), as the pointwise
infimum of affine functions, is concave. Moreover, since
\(\underline\xi>0\), \(R_n^*(\xi)>0\) on \([\underline\xi,\bar\xi]\). It follows that 
$
\frac{R_n(\mathcal C_n;\xi)}{R_n^*(\xi)}.
$
is quasi-convex in $\xi$ since it is the ratio of an affine function in $\xi$ and a concave function in $\xi$. By quasi-convexity, its largest value is achieved at one of the extreme point so that 
$ 
\sup_{\xi\in[\underline\xi,\bar\xi]}
\frac{R_n(\mathcal C_n;\xi)}{R_n^*(\xi)}
=
\max\left\{
\frac{R_n(\mathcal C_n;\underline\xi)}{R_n^*(\underline\xi)},
\frac{R_n(\mathcal C_n;\bar\xi)}{R_n^*(\bar\xi)}
\right\}.
$

\section{Additional extensions} \label{app:more_extensions}

In this section we present additional extension and formal results. The proof of claims in the online supplement can be found in the additional online supplement \ref{sec:proof_appendix}.

\subsection{Random outcomes} \label{sec:random_outcomes}

While the finite population framework considered in Assumption \ref{ass:first_order} can be viewed as conditioning on potential outcomes, Assumption \ref{ass:exposure_restriction} may not be plausible for realizations of potential outcomes.  In these settings, we can re-interpret our approach by imposing Assumption \ref{ass:exposure_restriction} on the expected potential outcomes, which can take any values in $[0, 1]$. For simplicity, we consider a local asympotic framework with $\varepsilon = o(1)$ in Assumption \ref{defn:local} (i.e., spillovers are local to zero), but clearly one could extend our reasoning to settings with small but non-zero $\varepsilon$. 

To accommodate random outcomes, we can modify Assumption \ref{ass:first_order} by redefining
$\mu_i(\mathbf{d})$.
\begin{ass}[Weaker version of Assumption \ref{ass:first_order}] \label{ass:first_order_random} For $i \in \{1, \cdots,n\}$, 
$
\E[Y_i(\mathbf{d})] = \mu_i(\mathbf{d}_i, \mathbf{d}_{\mathcal{N}_i}), \quad \forall \mathbf{d}\in \{0, 1\}^n.
$
for some functions $\mu_i(1, \cdot) \in \mathcal{M}_{1,i}, \mu_i(0, \cdot) \in 
\mathcal{M}_{0, i}$ for some set of functions $\mathcal{M}_{1, i}, \mathcal{M}_{0, i}$.
\end{ass} 
Similarly, we can redefine the global average treatment effect as 
$\tau_{n, \mu} = \frac{1}{n}\sum_{i=1}^{n}(\mu_i(\mathbf{1}) - \mu_i(\mathbf{0})).$ 
In this setting, the expression of the bias in Lemma \ref{lem:worst_case_bias} does not change since the potential outcomes $Y_i(\mathbf{d})$ are independent of the treatment assignments. However, the variance and MSE will be inflated due to the randomness. To make the problem tractable, we impose the following two assumptions on the variance-covariance structure of potential outcomes.
\begin{ass}[Second moment of potential outcomes] \label{ass:second_moment} For some non-decreasing and Lipschitz 
 function $g: \R\rightarrow \R^{+}$, for any $i\in \{1, \ldots, n\}$ and $\mathbf{d}\in \{0, 1\}^n$,
$\E[Y_i(\mathbf{d})^2] = g(\E[Y_i(\mathbf{d})])$. 
\end{ass}

\begin{ass}[Weak correlation among potential outcomes] \label{ass:weak_correlation} Let $\mathbf{C}(\mathbf{d})$ denote the correlation matrix of $(Y_1(\mathbf{d}), \ldots, Y_n(\mathbf{d}))$. Then, by letting $||\cdot||_{\mathrm{op}}$ the operator norm,
$\sup_{\mathbf{d}\in \{0, 1\}^n}\|\mathbf{C}(\mathbf{d}) - \mathbf{I}\|_{\mathrm{op}} = o(1).$ 
\end{ass}

Assumption \ref{ass:second_moment} holds for a variety of outcome distributions, including the Bernoulli distribution ($g(\mu) = \mu$), Poisson distribution ($g(\mu) = \mu + \mu^2$), Exponential distribution ($g(\mu) = 2\mu^2$), $\chi^2$ distribution ($g(\mu) = 2\mu + \mu^2$), and Gaussian distribution ($g(\mu) = \mu^2 + \sigma^2$). Assumption \ref{ass:weak_correlation} holds when $\{Y_i(\mathbf{d}): \mathbf{d}\in \{0, 1\}^n\}$ are independent or weak dependent across $i$ and we consider a local asymptotic framework in Assumption \ref{defn:local} with $\varepsilon = o(1)$ (i.e., effects are local to zero but possibly non-negligible for inference). In particular, under the endogenous peer effects model \citep{manski1993identification, bramoulle2009identification}, this assumption is satisfied when the peer effect is small; see Appendix \ref{example:endogenous}.

Under these assumptions, we prove that the worst-case variance and MSE are both inflated by a constant that does not depend on the clustering. 

\begin{thm}\label{thm:random}
Let Assumptions \ref{ass:exposure_restriction}, \ref{defn:local}, \ref{ass:clusters}, \ref{ass:first_order_random}, \ref{ass:second_moment}, \ref{ass:weak_correlation} hold, with $\mu_i(\cdot)$ as defined in Assumption \ref{ass:first_order_random} and $\varepsilon = 0$ in Assumption \ref{defn:local}. Further assume $|\psi_R|\ge |\psi_L|$. Then the worst-case bias is 
$$
\small 
\begin{aligned} 
\sup_{\mu \in \mathcal{M} } \Big|\tau_{n, \mu} - \mathbb{E}_\mu[\hat{\tau}_n(\mathcal{C}_{ n})]\Big| = \bar{\phi}_n b_n(\mathcal{C}_n),
\end{aligned} 
$$
the worst-case variance is 
$$
\small 
\begin{aligned} 
&  \sup_{\mu \in \mathcal{M}} \mathbb{E}_\mu\Big[\Big(\hat{\tau}_n(\mathcal{C}_{ n}) - \mathbb{E}[\hat{\tau}_n(\mathcal{C}_{ n})]\Big)^2\Big]  = \left(\frac{\bar{\psi}}{n^2} \sum_{k=1}^{K_n} n_k^2 + \frac{4g(\psi_R) - \bar{\psi}}{n}\right)(1 + o(1)),
\end{aligned} 
$$
and the worst-case MSE is 
$$
\small 
\begin{aligned} 
\sup_{\mu\in \mathcal{M}} \E_{\mu}[(\hat{\tau}_n(\mathcal{C}_{ n}) - \tau_{n, \mu})^2] = \left(\mathcal{B}_n^*(\mathcal{C}_{ n}, 1) + \frac{4g(\psi_R) - \bar{\psi}}{n}\right)(1 + o(1)).
\end{aligned}
$$
\end{thm}
\begin{proof}[Proof of Theorem \ref{thm:random}]
See Appendix \ref{app:proof:thm:random}.
\end{proof}

\subsection{Relaxing the ex-ante symmetry}
\label{sec:heterogeneity}

Assumption \ref{ass:exposure_restriction} imposes ex-ante symmetry in the sense that the range of potential outcomes and the size of spillover effect are identical across individuals. This excludes regression adjustment under which the range of potential outcomes differs across individuals (see Remark \ref{rem:reg_adj}) or heterogeneous spillover effects. To account for the heterogeneity, we consider the following relaxation of Assumption \ref{ass:exposure_restriction}.
\begin{ass} \label{ass:exposure_restriction_heterogeneous} The potential outcomes $\mu = (\mu_1, \cdots, \mu_n) \in \mathcal{M}  = \otimes_{i=1}^n \mathcal{M}_i$ where $\mathcal{M}_i$ is the set such that $\mu_i(1, \cdot) \in \mathcal{M}_{1, i}, \mu_i(0, \cdot) \in \mathcal{M}_{0, i}$. For all $i \in \{1, \cdots, n\}$, $d \in \{0,1\}$, the function classes $\mathcal{M}_{d,i}$ is the class of functions that include all functions $\mu_i(d,\cdot)$ satisfying the following two conditions 
\begin{itemize} 
\item[(i)] $\mu_i(d, \cdot)\in [\psi_{Li}, \psi_{Ri}]$ for some $-\infty < \psi_{Li} < \psi_{Ri} < \infty$;
\item[(ii)] for all $\mathbf{d} \in \{0,1\}^{|\mathcal{N}_i|}$, for some (unknown) $\alpha_i \in [\underline{\alpha},1], \underline{\alpha} > 0$, with $\max_i \alpha_i = 1$, 
\begin{equation} \label{eqn:interference} 
\small 
\begin{aligned} 
\Big|\mu_i(0, \mathbf{d}) - \mu_i(0, \mathbf{0})\Big| & \le \bar{\phi}_n \frac{\alpha_i}{|\mathcal{N}_i|} \sum_{k \in \mathcal{N}_i} \mathbf{d}_k, \qquad 
 \Big|\mu_i(1, \mathbf{d}) - \mu_i(1, \mathbf{1})\Big| & \le  \bar{\phi}_n \frac{\alpha_i}{|\mathcal{N}_i|} \sum_{k \in \mathcal{N}_i} \Big(1 - \mathbf{d}_k\Big), 
 \end{aligned} 
\end{equation} 
for some $\bar{\phi}_n \le \min_{i}\{\psi_{Ri} - \psi_{Li}\}$. 
\end{itemize} 
\end{ass} 

We can still derive the worst-case bias, variance, and MSE assuming the above under mild assumptions on $(\psi_{Li}, \psi_{Ri})$. 
\begin{thm}\label{thm:worst_case_heterogeneous}
Let Assumptions \ref{ass:first_order}, \ref{defn:local}, \ref{ass:clusters}, \ref{ass:exposure_restriction_heterogeneous} hold. Further assume that $|\psi_{Ri}|\ge |\psi_{Li}|$ for all $i\in \{1, \ldots, n\}$, or $|\psi_{Li}|\ge |\psi_{Ri}|$ for all $i\in \{1, \ldots, n\}$. Let $\bar{\psi}_i = 2 \max\{|\psi_{Ri}|, |\psi_{Li}|\} $ and assume that $\underline{\psi}\le (1/n)\sum_{i=1}^{n}\bar{\psi}_i\le \bar{\psi}$.
for some constants $0 < \underline{\psi} < \bar{\psi} < \infty$ that does not vary with $n$. 
Then the worst-case bias is 
$$
\small 
\begin{aligned} 
\sup_{\mu \in \mathcal{M} } \Big|\tau_{n, \mu} - \mathbb{E}_\mu[\hat{\tau}_n(\mathcal{C}_{ n})]\Big| = \bar{\phi}_n b_{n, \alpha}(\mathcal{C}_n),
\qquad b_{n, \alpha}(\mathcal{C}_n) = \frac{1}{n} \sum_{i=1}^n \frac{\alpha_i}{|\mathcal{N}_i|} \Big| \mathcal{N}_i \bigcap \Big\{j: c(j) \neq c(i) \Big\}\Big|.
\end{aligned} 
$$ 
The worst-case variance is 
$$
\small 
\begin{aligned} 
&  \sup_{\mu \in \mathcal{M}} \mathbb{E}_\mu\Big[\Big(\hat{\tau}_n(\mathcal{C}_{ n}) - \mathbb{E}[\hat{\tau}_n(\mathcal{C}_{ n})]\Big)^2\Big]  = \left\{\sum_{k=1}^{K_n}\frac{n_k^2}{n^2} \lb\frac{1}{n_k}\sum_{i\in c_k}\bar{\psi}_i\rb^2\right\}(1 + \mathcal{O}(\varepsilon)),
\end{aligned}
$$
and the worst-case MSE is 
$$
\small 
\begin{aligned} 
\sup_{\mu\in \mathcal{M}} \E_{\mu}[(\hat{\tau}_n(\mathcal{C}_{ n}) - \tau_{n, \mu})^2] = \left\{\bar{\phi}_n^2 b_{n, \alpha}(\mathcal{C}_n)^2 + \sum_{k=1}^{K_n}\frac{n_k^2}{n^2} \lb\frac{1}{n_k}\sum_{i\in c_k}\bar{\psi}_i\rb^2\right\}(1 + \mathcal{O}(\varepsilon)).
\end{aligned} 
$$
\end{thm}

\begin{proof}[Proof of Theorem \ref{thm:worst_case_heterogeneous}] See Appendix \ref{proof:heterogeneity}. 
\end{proof} 

It is easy to show that similarly to what discussed in the main text, the optimization program admits a semi-definite relaxation that depends on the values of $\alpha$ and $\bar{\psi}$. 

\begin{rem}[Non-binary networks] Following a similar argument, it is also possible to consider non-binary networks where condition (ii) in Assumption \ref{ass:exposure_restriction_heterogeneous} by introducing weights $w_{i,k}$. This is omitted for brevity. \qed  
\end{rem} 

\subsection{Trimming estimators} \label{sec:other_estimators}

To illustrate how our framework generalizes to alternative estimators, in this section we study properties of 
\begin{equation} \label{eqn:estimator2}
\small 
\begin{aligned} 
\hat{\tau}_n^{w}(\mathcal{C}_{ n}) = \frac{2}{\sum_{i=1}^n w_i} \sum_{i=1}^n w_i \Big[D_i Y_i - (1 - D_i) Y_i \Big], \,\, w_i\ge 0
\end{aligned} 
\end{equation}  
for arbitrary non-negative weights $w$ as deterministic functions of the matrix $\mathbf{A}$ and on the clustering $\mathcal{C}_n$ (but not of the treatments or outcomes).\footnote{The restriction that the weights do not depend on the treatments is imposed for simplicity since otherwise, correlations may also depend on the correlations between the weights. The focus on difference in means estimators instead of inverse probability weights estimators is motivated by their robustness to poor overlap.}  For example, these weights may upweight units with many neighbors in the same cluster and down-weights units with few neighbors in the same cluster in the spirit of trimming estimators \citep{leung2023design}. We show how our framework extends to estimators different from the difference in means estimators, provide a characterization of their bias and variance as in previous sections, and an algorithmic procedure to optimize over the weights $w$ and the clustering $\mathcal{C}_n$.


\begin{thm}\label{thm:weighted_estimator}
Let Assumptions \ref{ass:first_order}, \ref{ass:exposure_restriction}, \ref{defn:local}, \ref{ass:clusters} hold. Consider an estimator $\hat{\tau}_n^{w}$, 
as in Equation \eqref{eqn:estimator2}.  Then as $n \rightarrow \infty$, 
$$
\small 
\begin{aligned} 
\sup_{\mu \in \mathcal{M} } \Big|\tau_{n, \mu} - \mathbb{E}_\mu[\hat{\tau}_n^{w}(\mathcal{C}_{ n})]\Big| = \left(\bar{\phi}_n b_n^{w}(\mathcal{C}_n) +(\psi_R - \psi_L)\sum_{i=1}^{n}\Big| \frac{w_i}{\|w\|_1} - \frac{1}{n}\Big|\right)\cdot (1 + \mathcal{O}(\varepsilon))
\end{aligned} 
$$ 
where 
$$
\small 
\begin{aligned} 
b_n^{w}(\mathcal{C}_{ n}) = \sum_{i=1}^n \frac{w_i }{||w||_1}\frac{1}{|\mathcal{N}_i|} \Big| \mathcal{N}_i \bigcap \Big\{j: c(i) \neq c(j)\Big\}\Big|.
\end{aligned} 
$$ 
If, further, $w_i\le \bar{w}$ and $\|w\|_1\ge \underline{w}n$ for some constants $0 < \underline{w} < \bar{w} < \infty$ that do not vary with $n$. Then, as $K_n \rightarrow \infty$, 
\[\sup_{\mu \in \mathcal{M}} \mathbb{E}_\mu\Big[\Big(\hat{\tau}_n^{w}(\mathcal{C}_{ n}) - \mathbb{E}_\mu[\hat{\tau}_n^{w}(\mathcal{C}_{ n})]\Big)^2\Big] = \frac{1}{\|w\|_1^2}\sum_{k=1}^{K_n}\left(\sum_{i: c(i)  = k}w_i\right)^2\cdot \bar{\psi} (1 + \mathcal{O}(\varepsilon)).\]
\end{thm}


\begin{proof}[Proof of Theorem \ref{thm:weighted_estimator}]   See Appendix \ref{app:proof:thm:weighted_estimator}.
\end{proof} 

Theorem \ref{thm:weighted_estimator} characterizes the worst-case bias and variance of the trimming estimator. Following the discussion for the difference in means estimator, this result suggests a weighted objective function 
\begin{equation}\label{eq:objective_trimming}
\small 
\begin{aligned} 
\mathcal{B}_{n}^{*}(\mathcal{C}_n, w, \lambda) = \bar{\psi}\frac{1}{\|w\|_1^2}\sum_{k=1}^{K_n}\left(\sum_{i: c(i)  = k}w_i\right)^2 + \lambda \left(\bar{\phi}_n b_n^{w}(\mathcal{C}_n) +(\psi_R - \psi_L)\sum_{i=1}^{n}\Big| \frac{w_i}{\|w\|_1} - \frac{1}{n}\Big|\right)^2.
\end{aligned} 
\end{equation}
For the difference in means estimator, the objective function $\mathcal{B}_n^{*}(\mathcal{C}_n, \lambda)$ defined in Theorem \ref{thm:Bn_Bnstar} can be interpreted as the worst-case MSE, up to $1+\mathcal{O}(\varepsilon)$ terms, as implied by Theorem \ref{thm:Bn_Bnstar}. Unfortunately, this exact equivalence no longer holds for general trimming estimators because the worst-case bias and variance are achieved by different configurations of potential outcomes. In fact, we can show that the worst-case MSE is given by the maximum of a quadractic function in the vectors $(\mu_1(\mathbf{1}), \ldots, \mu_n(\mathbf{1}))$ and $(\mu_1(\mathbf{0}), \ldots, \mu_n(\mathbf{0}))$ over the constraint defined by Assumption \ref{ass:exposure_restriction} (i); see Remark \ref{rem:exact_worst_case_mse_trimming} in Appendix \ref{app:proof:thm:worst_case_mse_trimming} for details. Nevertheless, we can show that the equivalence holds up to a constant. 

\begin{thm}\label{thm:worst_case_mse_trimming}
Under the conditions in Theorem \ref{thm:weighted_estimator}, with $w_i\le \bar{w}$ and $\|w\|_1\ge \underline{w}n$ for some constants $0 < \underline{w} < \bar{w} < \infty$ that do not vary with $n$, 
\[1/4 \le \liminf_{K_n\rightarrow \infty}\frac{\sup_{\mu\in \mathcal{M}}\E_{\mu}[(\hat{\tau}_n^{w}(\mathcal{C}_{ n}) - \tau_{n, \mu})^2]}{\mathcal{B}_n^*(\mathcal{C}_n, w, 1)(1 + \mathcal{O}(\varepsilon))}\le \limsup_{K_n\rightarrow \infty}\frac{\sup_{\mu\in \mathcal{M}}\E_{\mu}[(\hat{\tau}_n^{w}(\mathcal{C}_{ n}) - \tau_{n, \mu})^2]}{\mathcal{B}_n^*(\mathcal{C}_n, w, 1)(1 + \mathcal{O}(\varepsilon))}\le 1.\]
\end{thm}

\begin{proof}[Proof of Theorem \ref{thm:worst_case_mse_trimming}]   See Appendix \ref{app:proof:thm:worst_case_mse_trimming}.
\end{proof} 

Theorem \ref{thm:worst_case_mse_trimming} justifies that the objective function $\mathcal{B}_n^{*}(\mathcal{C}_n, w, 1)$ is a reasonable proxy for the worst-case MSE -- its optimal solution yields a clustering whose worst-case MSE is no more than four times the optimal worst-case MSE. 

As with Section \ref{sec:opt}, to obtain the frontier that trades off the worst-case bias and variance, we can remove the square in the second term and optimize the following proxy objective for a range of $\xi$:
\begin{equation}\label{eq:proxy_trimming}
\small 
\begin{aligned} 
\frac{\xi}{\|w\|_1^2}\sum_{k=1}^{K_n}\left(\sum_{i: c(i)  = k}w_i\right)^2 + \sum_{i=1}^n \frac{w_i }{||w||_1}\frac{1}{|\mathcal{N}_i|} \Big| \mathcal{N}_i \bigcap \Big\{j: c(i) \neq c(j)\Big\}\Big|,
\end{aligned} 
\end{equation}
where the term $(\psi_R - \psi_L)\sum_{i=1}^{n}\Big| \frac{w_i}{\|w\|_1} - \frac{1}{n}\Big|$ is not kept because it does not depend on the clustering. The program admits a semi-definite relaxation of the form
\begin{equation}\label{eq:SDP_weights}
\max_{\mathbf{X}(K)} \tr(\mathbf{L}_{\xi}^w \mathbf{X}(K)), \quad \text{s.t.}\quad \diag(\mathbf{X}(K)) = \one_{n}, \,\, \mathbf{X}(K)\succeq 0, \quad \mathbf{L}_{\xi}^w = ||w||_1 \mathbf{L}^w - \xi w w^\top, 
\end{equation}
where $\mathbf{L}_{i,j}^w = w_i \mathbf{L}_{i,j}$. 


It is possible to optimize sequentially over the weights by first, for given weights $w$ optimizing over the clustering, and for given clustering, optimizing over the weights $w$. Optimization over the weights can be solved as a fractional program \citep{bitran1973linear}. The approximately optimal clustering for given weights $w$ can be obtained via semi-definite programming. 

Algorithm \ref{alg:weights} presents details. In summary, Theorem \ref{thm:weighted_estimator} shows how our insights in the characterization of the worst-case bias and variance more broadly apply to other estimators. 

  \begin{algorithm} [!h]   \caption{Causal Clustering with trimming estimators}\label{alg:weights}
  \footnotesize 
    \begin{algorithmic}[1]
    \Require Adjacency matrix $\mathbf{A}$, \underline{K}, $\bar{K}$ smallest and largest number of clusters, $\xi$, number of iterations $T$
    \State Initialize $w_i^0 = 1$ for all $i \in \{1, \cdots, n\}$
   \For{$t \in \{0, \cdots, T\}$}
    \For{$K \in \{\underline{K}, \cdots, \bar{K}\}$} 
     \begin{algsubstates}
     \State Solve Equation \eqref{eq:SDP_weights} and obtain $\hat{\mathbf{X}}(K)$ as the solution of Equation with weights  \eqref{eq:SDP_weights} under a semi-definite relaxation on $\mathbf{X}(K) \succeq 0$, with weights $w^t$
     \State Retrieve the clusters $c_K$ via 
     $K$-means algorithm on the first $K$ eigenvectors of $\hat{\mathbf{X}}(K)$
     \State Compute the objective function corresponding to the chosen clustering in
     \eqref{eq:SDP_weights}
     \State Define $\mathcal{C}_n^t$ the clustering with lowest objective in Equation \eqref{eq:SDP_weights}
        \end{algsubstates}
    \EndFor
     
    \State Define $w^{\star}$ the weights with lowest objective in Eq. \eqref{eq:objective_trimming} via fractional programming
    \State Define $w^{t+1} \leftarrow w^{\star}$
    \EndFor
    
\Return Clustering $\mathcal{C}_n^T$ and weights $w^{T+1}$ 
         \end{algorithmic}
\end{algorithm}

\subsection{Saturation designs} \label{sec:saturations}

In this section, we study settings where individuals in a given cluster are assigned treatments with different probabilities across clusters, i.e., saturation experiments. For the sake of brevity, we focus on settings where the researcher must choose between two treatment probabilities only, whereas our discussion applies to more than two treatment probabilities.  

\begin{ass}[Treatment assignments] \label{ass:arbitrary_treatments} Suppose that for all $i \in \{1,\cdots, n\}$, 
$$
\small 
\begin{aligned} 
D_i  \sim \mathrm{Bern}(p_{c(i)}), \quad P(p_{c(i)} = q_0) = P(p_{c(i)} = q_1) = 1/2  
\end{aligned} 
$$
where  for some $(q_0, q_1)$, $0 < q_0 < 0.5 < q_1 < 1, q_0 + q_1 = 1$, and $p_1, \cdots, p_{K_n}$ are independent.  In addition $D_i \perp D_j | p_{c(i)}, p_{c(j)}$
\end{ass} 

Assumption \ref{ass:arbitrary_treatments} generalizes the cluster design to designs with arbitrary cluster-level treatment probabilities. In particular, a cluster design has $q_0 = 0$ and $q_1 = 1$. We assume that $q_0 + q_1 = 1$ to ensure the marginal treatment probability for each unit remains $1/2$ and hence the same estimator can be used: 
$ 
\hat{\tau}_n(\mathcal{C}_n) = \frac{2}{n} \sum_{i=1}^n \Big[ Y_i D_i - Y_i(1 - D_i) \Big]
$ 
denoting a difference between treated and control units, reweighted by the individual treatment probability, similar to the difference in means estimator studied in previous sections. Noe that here we focus on $\hat{\tau}_n(\mathcal{C}_n)$ to estimate the GATE for simplicity, but in principle one could consider alternative estimands and estimators, see Remark \ref{rem:other_estimands}.

\begin{thm} \label{thm:saturation} Suppose that Assumptions \ref{ass:first_order}, \ref{ass:exposure_restriction}, \ref{defn:local}, \ref{ass:clusters} hold except that treatments are assigned as in Assumption \ref{ass:arbitrary_treatments}.  Then as $K_n \rightarrow \infty$, the worst-case bias is
\begin{equation} \label{eqn:bound3} 
\small 
\begin{aligned} 
\sup_{\mu \in \mathcal{M}}\Big|\mathbb{E}[\hat{\tau}_n^q(\mathcal{C}_{ n})] - \tau_{n,\mu} \Big| = & \bar{\phi}_n \left\{b_n(\mathcal{C}_n) + 4q_1q_0 (1 - b_n(\mathcal{C}_n))\right\},
\end{aligned}
\end{equation}
the worst-case variance is
\begin{equation}\label{eq:saturated_variance}
\small
\begin{aligned}
\sup_{\mu \in \mathcal{M}}\mathbb{E}\Big[\Big(\mathbb{E}[\hat{\tau}_n^q(\mathcal{C}_{ n})] - \hat{\tau}_n^q\Big)^2\Big] = & \Big((q_1 - q_0)^2\sum_{k=1}^{K_n} \frac{n_k^2}{n^2} + \frac{4 q_0 q_1}{n}\Big)  \bar{\psi} + \mathcal{O}(\varepsilon/K_n),
\end{aligned} 
\end{equation} 
and the worst-case MSE is the sum of the right-hand side of \eqref{eqn:bound3} squared and \eqref{eq:saturated_variance}.
\end{thm} 

\begin{proof}[Proof of Theorem \ref{thm:saturation}] See Appendix \ref{app:thm:saturation}. 
\end{proof} 

Theorem \ref{thm:saturation} characterizes the worst-case bias, variance, and MSE for estimating the global treatment effect under a saturation design. The bias depends on an additional bias component with respect to Lemma \ref{lem:worst_case_bias} that arises from possibly assigning individuals in the same cluster to different treatments. 
The bias under a saturation design increases because individuals in the same clusters are not all assigned to the same treatment status. On the other hand, the variance decreases because treatments are not necessarily perfectly correlated in the same cluster. For given treatment probabilities $(q_1, q_0)$ with $q_1 + q_0 = 1$, and given clustering algorithms, researchers can compare such algorithms by studying their worst-case bias, variance, or MSE. In practice, if researchers are mostly concerned about the bias of the estimated overall treatment effects, cluster experiments are recommended. 

In conclusion, Theorem \ref{thm:saturation} shows how our characterization of the bias and variance can be extended to more complex designs. 

 \begin{rem}[Estimating spillover effects] \label{rem:other_estimands} Whereas we focus here on the GATE for their relevance in many applications, such as online experiments \citep{karrer2021network}, saturation experiments can be useful to study other estimands of interest, such as direct and spillover effects \citep{baird2018optimal}. In particular, suppose researchers use an estimator
\begin{equation} \label{eqn:spillovers}
 \small 
 \begin{aligned} 
 \hat{\tau}_n^s(\mathcal{C}_n) = \frac{2}{n} \sum_{i=1}^n \Big[ \frac{Y_i D_i 1\{p_c(i) = q_1\}}{q_1} - \frac{Y_i D_i 1\{p_c(i) = q_0\}}{q_0} \Big] 
 \end{aligned}
 \end{equation}
 denoting the contrast between potential outcomes under treatment when switching the treatment status of their friends from one to zero. The estimator compares average outcomes between treated units in clusters with larger and lower treatment probability. It is easy to show that the tools developed in this paper allow us to characterize the bias and variance of the estimator in Equation \eqref{eqn:spillovers} relative to the estimand measuring exposure contrast for different saturation probabilities. In particular, results similar to those in Theorem \ref{thm:saturation} for the overall treatment effect can also apply to other estimands (spillover effects in Equation \ref{eqn:spillovers}, direct effects, and other forms of spillover effects) in the presence of saturation experiments discussed in the current subsection. \qed 
 \end{rem}

\subsection{Endogenous peer effects} \label{example:endogenous} 

Our results can extend to higher order spillover effects as long as we can write  
$
Y_i(\mathbf{d}) = \mu_i(\mathbf{d}_i, \mathbf{d}_{\mathcal{N}_i}) + \mathcal{O}(h_n), 
$
for some $h_n \rightarrow 0$. Our results hold if $h_n = o(1/n)$, capturing the idea that first-order effects $\bar{\phi}_n$ are larger than second-order effects. 

We provide an example where this holds. 
Following models in \cite{manski1993identification}, let  
$$
\small 
\begin{aligned} 
Y_i(\mathbf{d}) = \alpha + \beta_n \mathbf{d}_i + \kappa_n \frac{\sum_{j \neq i} \mathbf{A}_{i,j} \mathbf{d}_j}{\sum_{j \neq i} \mathbf{A}_{i,j} } + \gamma_n \frac{\sum_{j \neq i} \mathbf{A}_{i,j} Y_i(\mathbf{d})}{\sum_{j \neq i} \mathbf{A}_{i,j}} + \nu_i, \quad \beta_n, \kappa_n, \gamma_n \in [-\bar{\phi}_n, \bar{\phi}_n], \bar{\phi}_n \in (-1, 1),
\end{aligned} 
$$
for some (deterministic) $\nu_i$. By \eqref{eq:laplacian}, we can rewrite it as (recall $\mathbf{L}$ in Equation \eqref{eq:laplacian})
$$
\small 
\begin{aligned} 
Y_i(\mathbf{d}) = \alpha + \beta_n \mathbf{d}_i + \kappa_n \sum_{j \neq i} \mathbf{L}_{i,j} \mathbf{d}_j+ \gamma_n \sum_{j \neq i} \mathbf{
L}_{i,j} Y_i(\mathbf{d}) + \nu_i.
\end{aligned} 
$$
Assuming $\mathbf{A}$ is connected  \citep[Eq (6) in][]{bramoulle2009identification}, for arbitrary $G$, and letting $\bar{\mathbf{d}}_i^A = \sum_{j \neq i} \mathbf{L}_{i,j} \mathbf{d}_j$, 
$$
\small 
\begin{aligned} 
Y_i(\mathbf{d}) &= \alpha_{n, i} + \beta_n \mathbf{d}_i + \kappa_n \bar{\mathbf{d}}_i^A \\
& + \gamma_n \underbrace{\Big\{\underbrace{ \sum_{g=0}^{G - 1}  \gamma_n^{g} \mathbf{L}_i^{g+1} (\beta_n \mathbf{d} + \kappa_n \bar{\mathbf{d}}^A)}_{\text{Effect of } G \text{ closest friends}}  + \underbrace{\gamma_n^{G } \sum_{g=0}^{\infty}  \gamma_n^{g} \mathbf{L}_i^{g+G + 1} (\beta_n \mathbf{d} + \kappa_n \bar{\mathbf{d}}^A)}_{\text{Approximation error from $G + 1$ neighbors' effects}} \Big\}}_{\text{Interference from higher order friends}} + \tilde{\nu}_i,\end{aligned} 
$$
where $\alpha_n = (\alpha_{n, 1}, \ldots, \alpha_{n, n})' = \alpha(\mathbf{I} - \gamma_n\mathbf{L})^{-1}\mathbf{1}$ and $\tilde{\nu} = (\tilde{\nu}_1, \ldots, \tilde{\nu}_n)' = (\mathbf{I} - \gamma_n\mathbf{L})^{-1}\nu$.

Observe that the additional term that depends on higher order interference, multiplies by a factor $\gamma_n \beta_n + \gamma_n \kappa_n$. Taking $\varepsilon = 0$ in Assumption \ref{defn:local}, so that spillovers are local to zero (but possibly  non-negligible for inference), 
by letting each of these coefficients $\beta_n, \kappa_n, \gamma_n \propto \bar{\phi}_n$, Assumption \ref{ass:first_order} holds up to an order $o(1/n)$ (see Remark \ref{rem:higher_order}) for an adjacency matrix $\mathbf{A}_G$ where two individuals are connected if they are friends up to order $G$, provided that $\bar{\phi}_n^{G + 1} \sum_{g=0}^{\infty}  \bar{\phi}_n^{g} \mathbf{L}^{g+G + 1} \mathbf{1} = o(1/n)$. This example illustrates that the approximation error due to local interference decreases exponentially fast in $G$ as $\bar{\phi}_n = o(1)$ is local to zero. However, in settings where researcher use an adjacency matrix $\mathbf{A}_G$ ($G > 1$), we require stronger sparsity restrictions for our local asymptotic framework to hold (Assumption \ref{defn:local}), with $\bar{\phi}_n G \mathcal{N}_{n, \max}^{2G} = o(1)$, where $\mathcal{N}_{n, \max}$ is the maximum degree under the original adjacency matrix $\mathbf{A}$. 

Finally, we show that Assumption \ref{ass:weak_correlation} holds under this model if $\gamma_n$ is small. 
\begin{lem}\label{lem:endogenous_peer_effects_weak_correlation}
Assume that $\mathbf{A}$ is connected and $\nu_1, \ldots, \nu_n$ are independent with $\Var(\nu_i) = \sigma^2 < \infty$. Then Assumption \ref{ass:weak_correlation} holds if 
$\gamma_{n}\sqrt{\mathcal{N}_{n, \max}} = o(1).$ 
\end{lem}

\begin{proof}[Proof of Lemma \ref{lem:endogenous_peer_effects_weak_correlation}] See Appendix \ref{proof:lem:endogenous}.  
\end{proof}

\subsection{Estimated variance for inference on treatment effects} \label{sec:estimated_variance2}

Returning to Section \ref{sec:bias_aware_inferece}, in this section we discuss the estimation of the variance for ex-post inference.  We define the \emph{ex post bias-aware confidence interval} as
$ 
    \mathcal{I}_{\ep} = [\hat{\tau} \pm \chi(\bar{b}_{*, n}, \hat{\sigma})],
 $ where $\hat{\sigma}(\mathcal{C}_n)$ is an estimate of $\Var(\hat{\tau})$ using the experimental data.

  Let
  $ 
    X_i = 2Y_i (2D_i - 1),
  $ and 
  $$ 
  \small 
  \begin{aligned} 
  B_i = \{v: \text{either }c(v) = c(i) \text{ or } c(v) = c(v')\text{ for some }v'\in \mathcal{N}_i\}, \quad G_i = \{g: \mathcal{N}_g \cap B_i \neq \emptyset\}.
  \end{aligned} 
  $$
Lemma \ref{lem:2} below implies
  $ 
    X_j \indep X_i, \quad \forall j \not\in B_i \cup G_i.$
  It is easy to see that
  $j\in B_i \cup G_i \Longleftrightarrow i \in B_j \cup G_j.$ 
  Thus, Lemma \ref{lem:2} defines a undirected dependency graph among $(X_1, \ldots, X_n)$. Define
  $i\sim j \text{ iff }j\in B_i \cup G_i.$ 
  A particular choice is the network HAC estimator studied in \cite{gao2023causal}, which is guaranteed to be positive. 
 $ 
    \hat{\sigma}^2(\mathcal{C}_n) = \frac{1}{n^2}\sum_{i}\sum_{j\in B_i\cup G_i} K_{ij}(X_i - \hat{\tau})(X_j - \hat{\tau}),
 $ 
  where $K = (K_{ij})_{i,j}$ is a matrix satisfying
 $
K \succeq (I(i \sim j))_{i,j}$. 
  For example, \cite{gao2023causal} suggests
  $K = Q \max\{\Lambda, 0\}Q^T$
  where $Q\Lambda Q^T$ is the eigendecomposition of $(I(i \sim j))_{i,j}$.\footnote{As with other work in the literature, \cite{gao2023causal} prove that
$n\hat{\sigma}^2(\mathcal{C}_n)\ge n\Var(\hat{\tau}) + o_\P(1).$}

\newpage 

\section{Proofs of claims in online Appendix} \label{sec:proof_appendix}

\subsection{Proof of Theorem \ref{thm:random}}\label{app:proof:thm:random}

We study the worst-case bias and variance separately. We conclude with a characterization of the MSE. Recall that  for simplicity here we take a local asymptotic framework with $\varepsilon = 0$, although our derivations can be directly extended to positive $\varepsilon$ as for the other proofs. 

\paragraph{Worst-case bias}
Note that 
\[\E_\mu[\hat{\tau}_n(\mathcal{C}_n)] = \E_\mu[\E_\mu[\hat{\tau}_n(\mathcal{C}_n)\mid \mathbf{D}]] = \E_\mu\left[\frac{2}{n}\sum_{i=1}^{n}(2D_i - 1)\mu_i(\mathbf{D})\right].\]
The right-hand side is the expectation of the difference in means estimator with known $\mu_i$ which is studied in Lemma \ref{lem:worst_case_bias}. Thus, the worst-case bias remains the same and the worst case is achieved when $\mu_i(\mathbf{d})$ is given by \eqref{eq:worst_case} in the proof of Lemma \ref{lem:worst_case_bias}. 

\paragraph{Worst-case variance: Part I} Moving to the variance, by variance decomposition formula, 
\begin{equation} \label{eqn:variance_Part1}\mathbb{E}_\mu\Big[\Big(\hat{\tau}_n(\mathcal{C}_{ n}) - \mathbb{E}[\hat{\tau}_n(\mathcal{C}_{ n})]\Big)^2\Big] = \E_\mu[(\E_\mu[\hat{\tau}_n(\mathcal{C}_n)\mid \mathbf{D}] - \mathbb{E}[\hat{\tau}_n(\mathcal{C}_{ n})])^2] + \E_\mu[(\hat{\tau}_n(\mathcal{C}_n) - \E_\mu[\hat{\tau}_n(\mathcal{C}_n)\mid \mathbf{D}])^2].
\end{equation}
Again, the second term is the variance of the difference in means estimator with known $\mu_i(\mathbf{d})$'s. By \eqref{eq:worst_case_variance_max_degree} in the proof of Lemma \ref{lem:rel}, 
\begin{equation}\label{eq:binary_variance_term1}
\E_\mu[(\E_\mu[\hat{\tau}_n(\mathcal{C}_n)\mid \mathbf{D}] - \mathbb{E}[\hat{\tau}_n(\mathcal{C}_{ n})])^2] = \frac{1}{n^2}\sum_{k=1}^{K_n}\left(\sum_{i: c(i)  = k}(\mu_i(\mathbf{1}) + \mu_i(\mathbf{0}))\right)^2 + o(1/K_n).
\end{equation}
\paragraph{Worst-case variance: Part II} We decompose the second term in the right-hand side of Equation \eqref{eqn:variance_Part1} as 
\begin{equation}\label{eq:binary_variance_term2}
\small
\begin{aligned}
&\E_\mu[(\hat{\tau}_n(\mathcal{C}_n) - \E_\mu[\hat{\tau}_n(\mathcal{C}_n)\mid \mathbf{D}])^2]\\
& = \frac{4}{n^2}\sum_{i\neq j} \E_{\mu}\left[\mathbf{C}_{ij}(\mathbf{D})(2\mathbf{D}_i - 1)(2\mathbf{D}_j - 1)\sqrt{\Var(Y_i(\mathbf{D})\mid \mathbf{D})}\sqrt{\Var(Y_j(\mathbf{D})\mid \mathbf{D})}\right]\\
& \quad + \frac{4}{n^2}\sum_{i=1}^{n}\E_{\mu}[\Var(Y_i(\mathbf{D})\mid \mathbf{D})].
\end{aligned}
\end{equation}
Since $|\psi_R|\ge |\psi_L|$, it must be that $\psi_R > 0$. Let $v(\mathbf{D}) = (v_1(\mathbf{D}), \ldots, v_n(\mathbf{D}))$ where $v_i(\mathbf{D}) = (2\mathbf{D}_i - 1)\sqrt{\Var(Y_i(\mathbf{D}\mid \mathbf{D}))}$. Then 
$$
\small
\begin{aligned}
&\frac{4}{n^2}\sum_{i\neq j} \E_{\mu}\left[\mathbf{C}_{ij}(\mathbf{D})(2\mathbf{D}_i - 1)(2\mathbf{D}_j - 1)\sqrt{\Var(Y_i(\mathbf{D})\mid \mathbf{D})}\sqrt{\Var(Y_j(\mathbf{D})\mid \mathbf{D})}\right]\\
& = \frac{4}{n^2}\E_{\mu}[v(\mathbf{D})'(\mathbf{C}(\mathbf{D}) - \mathbf{I})v(\mathbf{D})]\\
& \le \frac{4}{n^2}\E_\mu[\|v(\mathbf{D})\|_2^2]\cdot \sup_{\mathbf{d}\in \{0, 1\}^n}\|\mathbf{C}(\mathbf{d}) - \mathbf{I}\|_{\mathrm{op}}.
\end{aligned}
$$
By Assumption \ref{ass:second_moment},
\[\Var(Y_i(\mathbf{D})\mid \mathbf{D}) = g(\mu_i(\mathbf{D})) - \mu_i(\mathbf{D})^2\le g(\psi_R).\]
Thus, $\|v(\mathbf{D})\|_2^2\le ng(\psi_R) = \mathcal{O}(n)$. By Assumption \ref{ass:weak_correlation}, 
\begin{equation}\label{eq:binary_variance_term2_term1}
\small
\begin{aligned}
&\frac{4}{n^2}\sum_{i\neq j} \E_{\mu}\left[\mathbf{C}_{ij}(\mathbf{D})(2\mathbf{D}_i - 1)(2\mathbf{D}_j - 1)\sqrt{\Var(Y_i(\mathbf{D})\mid \mathbf{D})}\sqrt{\Var(Y_j(\mathbf{D})\mid \mathbf{D})}\right] = o(1/n) = o(1/K_n).
\end{aligned}
\end{equation}
Returning to the second term of \eqref{eq:binary_variance_term2}. By Assumption \ref{ass:second_moment},
$$
\small
\begin{aligned}
 &\frac{4}{n^2}\sum_{i=1}^{n}\E_{\mu}[\Var(Y_i(\mathbf{D})\mid \mathbf{D})] = \frac{4}{n^2}\sum_{i=1}^{n}\E_{\mu}[g(\mu_i(\mathbf{D})) - \mu_i(\mathbf{D})^2]\\
 & = \frac{2}{n^2}\sum_{i=1}^{n}\E_{\mu}[g(\mu_i(1, \mathbf{D}_{-i})) - \mu_i(1, \mathbf{D}_{-i})^2\mid D_i = 1] + \E_{\mu}[g(\mu_i(0, \mathbf{D}_{-i})) - \mu_i(0, \mathbf{D}_{-i})^2\mid D_i = 0]
 \end{aligned}
$$
By Assumption \ref{ass:exposure_restriction}, and $|\psi_R| \ge |\psi_L|$,  
\[|\mu_i(1, \mathbf{D}_{-i}) - \mu_i(\mathbf{1})|\le \bar{\phi}_n, \quad -\psi_R \le \mu_i(1, \mathbf{D}_{-i})\le \psi_R.\]
Let $L_g$ denote the Lipschitz constant of $g$. Then by Assumption \ref{ass:second_moment}
$$
\begin{aligned}
&|g(\mu_i(1, \mathbf{D}_{-i})) - \mu_i(1, \mathbf{D}_{-i})^2 - \left\{g(\mu_i( \mathbf{1})) - \mu_i(\mathbf{1})^2\right\}|\\
& \le |\mu_i(1, \mathbf{D}_{-i}) - \mu_i(\mathbf{1})|\cdot |L_g + \mu_i(1, \mathbf{D}_{-i}) + \mu_i(\mathbf{1})|\le (L_g + 2 \psi_R)\bar{\phi}_n.
\end{aligned}
$$
Similarly, 
$$
\begin{aligned}
&|g(\mu_i(0, \mathbf{D}_{-i})) - \mu_i(0, \mathbf{D}_{-i})^2 - \left\{g(\mu_i( \mathbf{0})) - \mu_i(\mathbf{0})^2\right\}|\le (L_g +2 \psi_R)\bar{\phi}_n.
\end{aligned}
$$
Thus,
\begin{equation} 
\small 
\begin{aligned}
\frac{4}{n^2}\sum_{i=1}^{n}\E_{\mu}[\Var(Y_i(\mathbf{D})\mid \mathbf{D})]  &= \frac{2}{n^2} \sum_{i=1}^{n}\left(\E_\mu[g(\mu_i( \mathbf{1})) - \mu_i(\mathbf{1})^2] + \E_\mu[g(\mu_i( \mathbf{0})) - \mu_i(\mathbf{0})^2]\right) + \mathcal{O}\left(\bar{\phi}_n / n\right)\\
& = \frac{2}{n^2} \sum_{i=1}^{n}\left\{g(\mu_i( \mathbf{1})) - \mu_i(\mathbf{1})^2 + g(\mu_i( \mathbf{0})) - \mu_i(\mathbf{0})^2\right\} + o\left(1 / K_n\right),\label{eq:binary_variance_term2_term2}
\end{aligned}
\end{equation}
where the last line uses the fact that $\mu_i(\cdot)$ is non-random and Assumption \ref{defn:local} that implies $\bar{\phi}_n = o(1)$. 

\paragraph{Worst case variance: Part III, putting together bounds} Putting together \eqref{eq:binary_variance_term1}, \eqref{eq:binary_variance_term2}, \eqref{eq:binary_variance_term2_term1}, and \eqref{eq:binary_variance_term2_term2}, we have 
\begin{align}
&\E_\mu[(\hat{\tau}_n(\mathcal{C}_n) - \mathbb{E}[\hat{\tau}_n(\mathcal{C}_{ n})])^2] =  \frac{1}{n^2}\sum_{k=1}^{K_n}\left(\sum_{i: c(i)  = k}(\mu_i(\mathbf{1}) + \mu_i(\mathbf{0}))\right)^2\nonumber\\
& \quad + \frac{2}{n^2} \sum_{i=1}^{n}\left\{g(\mu_i( \mathbf{1})) - \mu_i(\mathbf{1})^2 + g(\mu_i( \mathbf{0})) - \mu_i(\mathbf{0})^2\right\} + o\left(1 / K_n\right)\label{eq:binary_variance}\\
& = \frac{2}{n^2}\sum_{k=1}^{K_n}\sum_{i < j: c(i)  = c(j) = k}(\mu_i(\mathbf{1}) + \mu_i(\mathbf{0}))(\mu_j(\mathbf{1}) + \mu_j(\mathbf{0}))\nonumber\\
& \quad + \frac{1}{n^2}\sum_{i=1}^{n}(\mu_i(\mathbf{1}) + \mu_i(\mathbf{0}))^2 + 2\left\{g(\mu_i( \mathbf{1})) - \mu_i(\mathbf{1})^2 + g(\mu_i( \mathbf{0})) - \mu_i(\mathbf{0})^2\right\} + o\left(1 / K_n\right)\nonumber\\
& = \frac{2}{n^2}\sum_{k=1}^{K_n}\sum_{i < j: c(i)  = c(j) = k}(\mu_i(\mathbf{1}) + \mu_i(\mathbf{0}))(\mu_j(\mathbf{1}) + \mu_j(\mathbf{0}))\nonumber\\
& \quad + \frac{1}{n^2}\sum_{i=1}^{n} 2g(\mu_i( \mathbf{1})) + 2g(\mu_i( \mathbf{0})) - (\mu_i(\mathbf{1}) - \mu_i(\mathbf{0}))^2+ o\left(1 / K_n\right).\nonumber
\end{align}
Since $\psi_R > 0$ and $|\psi_R| \ge |\psi_L|$, $\bar{\psi}^{1/2} = 2\psi_R $, 
\[(\mu_i(\mathbf{1}) + \mu_i(\mathbf{0}))(\mu_j(\mathbf{1}) + \mu_j(\mathbf{0}))\le 4\psi_R^2 = \bar{\psi}, \]
and
\[2g(\mu_i( \mathbf{1})) + 2g(\mu_i( \mathbf{0})) - (\mu_i(\mathbf{1}) - \mu_i(\mathbf{0}))^2\le 4g(\psi_R),\]
where both terms achieve their maximums when $\mu_i(\mathbf{1}) = \mu_i(\mathbf{0}) = \psi_R$ for all $i$. Thus, 
$$
\begin{aligned}
\E_\mu[(\hat{\tau}_n(\mathcal{C}_n) - \mathbb{E}[\hat{\tau}_n(\mathcal{C}_{ n})])^2]& \le \frac{\bar{\psi}}{n^2}\sum_{k=1}^{K_n}n_k^2 + \frac{4(g(\psi_R) - \psi_R^2)}{n} + o(1/K_n)\\
& = \frac{\bar{\psi}}{n^2}\sum_{k=1}^{K_n}n_k^2 + \frac{4g(\psi_R) - \bar{\psi}}{n} + o(1/K_n).
\end{aligned}
$$
Since $4g(\psi_R) - \bar{\psi} = 4(g(\psi_R) - \psi_R^2)\ge 0$ and the first term is lower bounded by $\bar{\psi} / K_n$, we have that
\[\E_\mu[(\hat{\tau}_n(\mathcal{C}_n) - \mathbb{E}[\hat{\tau}_n(\mathcal{C}_{ n})])^2]\le \left(\frac{\bar{\psi}}{n^2}\sum_{k=1}^{K_n}n_k^2 + \frac{4g(\psi_R) - \bar{\psi}}{n}\right)(1 + o(1)).\]
\paragraph{Upper bound on MSE} We can upper bound the worst-case MSE as 
\begin{align*}
\E_\mu[(\hat{\tau}_n(\mathcal{C}_n) - \tau_{n, \mu})^2]
& \le \left(\bar{\phi}_n^2 b_n(\mathcal{C}_n)^2 + \frac{\bar{\psi}}{n^2}\sum_{k=1}^{K_n}n_k^2 + \frac{4g(\psi_R) - \bar{\psi}}{n}\right)(1 + o(1))\\
& = \left(\mathcal{B}_n^*(\mathcal{C}_n, 1) + \frac{4g(\psi_R) - \bar{\psi}}{n}\right)(1 + o(1)).
\end{align*}
\paragraph{Lower Bound on MSE (Upper Bound is achievable)} To prove the upper bounds are achievable, it is left to show that the potential outcomes defined in \eqref{eq:worst_case} achieves the bound for the variance, because we have already proved that it achieves the worst-case bias in the proof of Lemma \ref{lem:worst_case_bias}. For this construction, $\mu_i(\mathbf{1}) = \psi_R$ and $\mu_i(\mathbf{0}) = \psi_R - \bar{\phi}_n$. Plugging them into \eqref{eq:binary_variance}, it is easy to see that the variance is within $(1 + O(\bar{\phi}_n))$ times the upper bound. By Assumption \ref{defn:local}, it implies the worst-case variance is achieved by \eqref{eq:worst_case}.

\subsection{Proof of Theorem \ref{thm:worst_case_heterogeneous}} \label{proof:heterogeneity}
Assume without loss of generality that $|\psi_{Ri}|\ge |\psi_{Li}|$ for all $i$. Then $\bar{\psi}_i = 2 \psi_{Ri} > 0$. 

\paragraph{Worst-case bias} Following the proofs of Lemma \ref{lem:worst_case_bias}, we can prove that 
$\Big|\tau_{n, \mu} - \mathbb{E}_\mu[\hat{\tau}_n(\mathcal{C}_{ n})]\Big| \le \bar{\phi}_n b_{n, \alpha}(\mathcal{C}_n),$ 
and it can be achieved by the following construction of potential outcomes:
\begin{equation}\label{eq:worst_case_heterogeneous}
\mu_i(1, \mathbf{d}_{-i}) = \mu_i(0, \mathbf{d}_{-i}) = 
\bar{\psi}_{i}/2 -  \bar{\phi}_n\frac{\alpha_i}{|\mathcal{N}_i|} \sum_{k \in \mathcal{N}_i}  (1 - \mathbf{d}_k).
\end{equation}
Note that the above construction is legitimate because $\bar{\phi}_n \le \psi_{Ri} - \psi_{Li}$ and $\alpha_i \le 1$, implying that $\mu_i(\mathbf{d})\in [\psi_{Li}, \psi_{Ri}]$. 

\paragraph{Worst-case variance} Following the proof of Lemma \ref{lem:rel}, we can recover \eqref{eq:worst_case_variance_max_degree}:
\begin{align*}
& \sup_{\mu\in \mathcal{M}}\mathbb{E}_\mu\Big[\Big(\hat{\tau}_n(\mathcal{C}_{ n}) - \mathbb{E}_\mu[\hat{\tau}_n(\mathcal{C}_{ n})]\Big)^2\Big]= \frac{1}{n^2}\sum_{k=1}^{K_n}\left(\sum_{i: c(i)  = k}(\mu_i(\mathbf{1}) + \mu_i(\mathbf{0}))\right)^2 + \mathcal{O}(\varepsilon/K_n).
\end{align*}
Since $|\mu_i(\mathbf{1})|, |\mu_i(\mathbf{0})|\le \bar{\psi}_i/2$, 
$\sup_{\mu \in \mathcal{M}} \mathbb{E}_\mu\Big[\Big(\hat{\tau}_n(\mathcal{C}_{ n}) - \mathbb{E}[\hat{\tau}_n(\mathcal{C}_{ n})]\Big)^2\Big]  \le \sum_{k=1}^{K_n}\frac{n_k^2}{n^2} \lb\frac{1}{n_k}\sum_{i\in c_k}\bar{\psi}_i\rb^2 + \mathcal{O}(\varepsilon/K_n).$ 
On the other hand, for the potential outcomes defined in \eqref{eq:worst_case_heterogeneous}, $\mu_i(\mathbf{1}) = \bar{\psi}_i / 2, \mu_i(\mathbf{0}) = \bar{\psi}_i / 2 - \bar{\phi}_n \alpha_i$. In this case, let $\alpha_i \le \bar{\alpha} < \infty$
$$  
\small 
\begin{aligned}
& \frac{1}{n^2}\sum_{k=1}^{K_n}\left(\sum_{i: c(i)  = k}(\mu_i(\mathbf{1}) + \mu_i(\mathbf{0}))\right)^2 = \sum_{k=1}^{K_n}\frac{n_k^2}{n^2} \lb\frac{1}{n_k}\sum_{i\in c_k}\bar{\psi}_i - \bar{\phi}_n \frac{1}{n_k}\sum_{i\in c_k} \alpha_i\rb^2 \\
& \ge \sum_{k=1}^{K_n}\frac{n_k^2}{n^2} \lb\frac{1}{n_k}\sum_{i\in c_k}\bar{\psi}_i\rb^2 - 2 \bar{\alpha} \bar{\phi}_n\sum_{k=1}^{K_n}\frac{n_k}{n^2}\left(\sum_{i\in c_k}\bar{\psi}_i\right) \ge \sum_{k=1}^{K_n}\frac{n_k^2}{n^2} \lb\frac{1}{n_k}\sum_{i\in c_k}\bar{\psi}_i\rb^2 - \frac{2\bar{\alpha} \bar{\gamma} \bar{\phi}_n}{K_n}\lb\frac{1}{n}\sum_{i=1}^{n}\bar{\psi}_i\rb,
\end{aligned}
$$ 
where the last line applies Assumption \ref{ass:clusters} for the second term. Since $(1/n)\sum_{i=1}^{n}\bar{\psi}_i \le \bar{\psi} = O(1)$ and $\bar{\phi}_n = \mathcal{O}(\varepsilon)$ by Assumption \ref{defn:local},
$\frac{1}{n^2}\sum_{k=1}^{K_n}\left(\sum_{i: c(i)  = k}(\mu_i(\mathbf{1}) + \mu_i(\mathbf{0}))\right)^2\ge \sum_{k=1}^{K_n}\frac{n_k^2}{n^2} \lb\frac{1}{n_k}\sum_{i\in c_k}\bar{\psi}_i\rb^2 + \mathcal{O}(\varepsilon/K_n)$. 
Thus, 
$\sup_{\mu \in \mathcal{M}} \mathbb{E}_\mu\Big[\Big(\hat{\tau}_n(\mathcal{C}_{ n}) - \mathbb{E}[\hat{\tau}_n(\mathcal{C}_{ n})]\Big)^2\Big]= \sum_{k=1}^{K_n}\frac{n_k^2}{n^2} \lb\frac{1}{n_k}\sum_{i\in c_k}\bar{\psi}_i\rb^2 + \mathcal{O}(\varepsilon/K_n).$ 
By Cauchy-Schwarz inequality,
$\sum_{k=1}^{K_n}\frac{n_k^2}{n^2} \lb\frac{1}{n_k}\sum_{i\in c_k}\bar{\psi}_i\rb^2 = \frac{1}{n^2} \sum_{k=1}^{K_n}\lb\sum_{i\in c_k}\bar{\psi}_i\rb^2\ge \frac{1}{K_n}\left(\frac{1}{n}\sum_{i=1}^{n}\bar{\psi}_i\right)^2\ge \frac{\underline{\psi}^2}{K_n}.$ 
Thus, 
$\sup_{\mu \in \mathcal{M}} \mathbb{E}_\mu\Big[\Big(\hat{\tau}_n(\mathcal{C}_{ n}) - \mathbb{E}[\hat{\tau}_n(\mathcal{C}_{ n})]\Big)^2\Big]= \sum_{k=1}^{K_n}\frac{n_k^2}{n^2} \lb\frac{1}{n_k}\sum_{i\in c_k}\bar{\psi}_i\rb^2 (1 + \mathcal{O}(\varepsilon)).$

\paragraph{Worst-case MSE} Finally, since the worst-case bias and variance are both achieved by \eqref{eq:worst_case_heterogeneous}, the worst-case MSE is the worst-case bias square plus the worst-case variance.

 \subsection{Proof of Theorem \ref{thm:weighted_estimator}} \label{app:proof:thm:weighted_estimator} 

 We separately provide bounds on the worst-case bias and variance. 
 
\paragraph{Upper bound on the worst-case bias.} Let 
\[\tau_{n, \mu}^{w} = \frac{1}{\|w\|_1}\sum_{i=1}^{n}w_i (\mu_i(\mathbf{1}) - \mu_i(\mathbf{0})),\]
and 
\[U_n^{*}(\mathcal{C}_n, w) = \bar{\phi}_n b_n^w(\mathcal{C}_n) 
 + (\psi_R - \psi_L)\sum_{i=1}^{n}\Big| \frac{w_i}{\|w\|_1} - \frac{1}{n}\Big|.\]

Following verbatim the proof of Lemma \ref{lem:worst_case_bias}, we have 
$$
\small 
\begin{aligned} 
\sup_{\mu \in \mathcal{M}} \Big| \mathbb{E}_\mu[\hat{\tau}_n^{w}] - \tau_{n,\mu}^w\Big| 
& = 
\sup_{\mu \in \mathcal{M}} \Big| 
\sum_{i=1}^n \frac{w_i}{\|w\|_1}\left(\mathbb{E}\Big[\mu_i(D_i, \mathbf{D}_{-i}) - \mu_i(\mathbf{1}) | D_i = 1\Big] - \mathbb{E}\Big[\mu_i(D_i, \mathbf{D}_{-i}) - \mu_i(\mathbf{0}) | D_i = 0\Big]\right)\Big| \\ 
&\le \sup_{\mu(1, \cdot) \in \otimes_{i=1}^n \mathcal{M}_{1, i}} \Big|
 \sum_{i=1}^n \frac{w_i}{\|w\|_1}\mathbb{E}\Big[\mu_i(D_i, \mathbf{D}_{-i}) - \mu_i(\mathbf{1}) | D_i = 1\Big]\Big| \\
& \qquad + \sup_{\mu(0, \cdot) \in \otimes_{i=1}^n \mathcal{M}_{0, i}} \Big|  \sum_{i=1}^n \frac{w_i}{\|w\|_1}\mathbb{E}\Big[\mu_i(D_i, \mathbf{D}_{-i}) - \mu_i(\mathbf{0}) | D_i = 0\Big]\Big|\\
& \le \bar{\phi}_n \sum_{i=1}^n \frac{w_i }{||w||_1}\frac{1}{|\mathcal{N}_i|} \Big| \mathcal{N}_i \bigcap \Big\{j: c(i) \neq c(j)\Big\}\Big|,
\end{aligned} 
$$ 
where the last line follows \eqref{eq:bias_term1} in the proof of Lemma \ref{lem:worst_case_bias}. Now we bound $|\tau_{n, \mu}^{w} - \tau_{n, \mu}|$. By the triangle inequality and Assumption \ref{ass:exposure_restriction} (i), 
$$
\small 
\begin{aligned} 
|\tau_{n, \mu}^{w} - \tau_{n, \mu}| & = \Big|\sum_{i=1}^{n}(\mu_i(\mathbf{1}) - \mu_i(\mathbf{0}))\lb \frac{w_i}{\|w\|_1} - \frac{1}{n}\rb\Big| \le (\psi_R - \psi_L)\sum_{i=1}^{n}\Big| \frac{w_i}{\|w\|_1} - \frac{1}{n}\Big|.
\end{aligned}
$$
Thus, we have proved that 
\begin{equation}\label{eq:bias_upper_bound_trimming}
\sup_{\mu \in \mathcal{M}} \Big| \mathbb{E}_\mu[\hat{\tau}_n^{w}] - \tau_{n,\mu}\Big|\le U_n^{*}(\mathcal{C}_n, w).
\end{equation}
\paragraph{Lower bound (achievability) of worst-case bias} To prove the achievability, we consider the following constructions of the potential outcomes:
\begin{equation}\label{eq:worst_case_trimming_bias}
\mu_i(1, \mathbf{d}_{-i}) = \mu_i(\mathbf{1}) - \bar{\phi}_n  \frac{1}{|\mathcal{N}_i|} \sum_{k \in \mathcal{N}_i}(1 - \mathbf{d}_k), \quad \mu_i(0, \mathbf{d}_{-i}) = \mu_i(\mathbf{0}) + \bar{\phi}_n  \frac{1}{|\mathcal{N}_i|} \sum_{k \in \mathcal{N}_i}  \mathbf{d}_k,
\end{equation}
and 
\[(\mu_i(\mathbf{1}), \mu_i(\mathbf{0})) = \left\{\begin{array}{ll}
(\psi_L + \bar{\phi}_n, \psi_R - \bar{\phi}_n) & (\text{if }w_i / \|w\|_1 \ge 1/n)\\
(\psi_R, \psi_L) & (\text{otherwise})
\end{array}\right..\]
Under Assumption \ref{ass:exposure_restriction},$\bar{\phi}_n\le \psi_R - \psi_L$, $\mu \in \mathcal{M}$. For this choice of $\mu$, repeating the above steps, 
\begin{equation}\label{eq:bias_trimming_term1}
\small 
\begin{aligned} 
\mathbb{E}_\mu[\hat{\tau}_n^{w}] - \tau_{n,\mu}^w
& = 
\sum_{i=1}^n \frac{w_i}{\|w\|_1}\left(\mathbb{E}\Big[\mu_i(D_i, \mathbf{D}_{-i}) - \mu_i(\mathbf{1}) | D_i = 1\Big] - \mathbb{E}\Big[\mu_i(D_i, \mathbf{D}_{-i}) - \mu_i(\mathbf{0}) | D_i = 0\Big]\right)\\ 
& = -\bar{\phi}_n\sum_{i=1}^n \frac{w_i}{\|w\|_1}\left(\frac{1}{|\mathcal{N}_i|}\sum_{k\in \mathcal{N}_i}\mathbb{E}[(1 - D_k)\mid D_i = 1] + \frac{1}{|\mathcal{N}_i|}\sum_{k\in \mathcal{N}_i}\mathbb{E}[D_k\mid D_i = 0]\right)\\ 
& = -\bar{\phi}_n \sum_{i=1}^n \frac{w_i }{||w||_1}\frac{1}{|\mathcal{N}_i|} \Big| \mathcal{N}_i \bigcap \Big\{j: c(i) \neq c(j)\Big\}\Big|,
\end{aligned} 
\end{equation}
and (letting $(\cdot)_{+} = \max\{\cdot, 0\}$)
$$
\small 
\begin{aligned} 
\tau_{n, \mu}^{w} - \tau_{n, \mu} &= \sum_{i=1}^{n}(\mu_i(\mathbf{1}) - \mu_i(\mathbf{0}))\lb \frac{w_i}{\|w\|_1} - \frac{1}{n}\rb\\
& = -(\psi_R - \psi_L)\sum_{i=1}^{n}\Big|\frac{w_i}{\|w\|_1} - \frac{1} {n}\Big| + 2\bar{\phi}_n \sum_{i=1}^{n}\lb\frac{w_i}{\|w\|_1} - \frac{1} {n}\rb_{+}\\
& = -(\psi_R - \psi_L - \bar{\phi}_n)\sum_{i=1}^{n}\Big|\frac{w_i}{\|w\|_1} - \frac{1} {n}\Big|,
\end{aligned}
$$
where the last line is due to the fact that because $w_i \ge 0$, 
\[\sum_{i=1}^{n}\lb\frac{w_i}{\|w\|_1} - \frac{1} {n}\rb = 0 \Longrightarrow \sum_{i=1}^{n}\lb\frac{w_i}{\|w\|_1} - \frac{1} {n}\rb_{+} = \sum_{i=1}^{n}\lb\frac{w_i}{\|w\|_1} - \frac{1} {n}\rb_{-} = \frac{1}{2}\sum_{i=1}^{n}\Big|\frac{w_i}{\|w\|_1} - \frac{1} {n}\Big|.\]
Putting two pieces together, for this choice of $\mu$, 
\begin{equation}\label{eq:bias_lower_bound_trimming}
|\E_{\mu}[\hat{\tau}_{n}^{w}] - \tau_{n, \mu}| = U_n^{*}(\mathcal{C}_n, w) - \bar{\phi}_n\sum_{i=1}^{n}\Big|\frac{w_i}{\|w\|_1} - \frac{1} {n}\Big|.
\end{equation}
By definition, 
\[U_n^{*}(\mathcal{C}_n, w)\ge (\psi_R - \psi_L)\sum_{i=1}^{n}\Big|\frac{w_i}{\|w\|_1} - \frac{1} {n}\Big|.\]
Thus, because $\bar{\phi}_n = \mathcal{O}(\varepsilon)$ under Assumption \ref{defn:local},  $\bar{\phi}_n\sum_{i=1}^{n}\Big|\frac{w_i}{\|w\|_1} - \frac{1} {n}\Big| = \mathcal{O}(\varepsilon U_n^{*}(\mathcal{C}_n, w))$ and 
\[|\E_{\mu}[\hat{\tau}_{n}^{w}] - \tau_{n, \mu}| = U_n^{*}(\mathcal{C}_n, w)(1 + \mathcal{O}(\varepsilon)). \]
Combining \eqref{eq:bias_upper_bound_trimming} and \eqref{eq:bias_lower_bound_trimming}, we prove the result for the worst-case bias.

\paragraph{Upper bound on the worst-case variance.} 

Following the proof of Lemma \ref{lem:rel}, we obtain that 
$$
\small
\begin{aligned}
&\mathbb{E}_\mu\Big[\Big(\hat{\tau}_n^{w}(\mathcal{C}_{ n}) - \mathbb{E}_\mu[\hat{\tau}_n^{w}(\mathcal{C}_{ n})]\Big)^2\Big]\\
&= \frac{4}{\|w\|_1^2}\sum_{i,j} w_i w_j\mathrm{Cov}\Big(\mu_i(D_i, \mathbf{D}_{-i})[2 D_i - 1], \mu_j(D_j, \mathbf{D}_{-j})[2 D_j - 1]\Big)\\
& = \frac{4}{\|w\|_1^2} \sum_{i,j: c(i) = c(j)} w_i w_j\mathrm{Cov}\Big(\mu_i(D_i, \mathbf{D}_{-i})[2 D_i - 1], \mu_j(D_j, \mathbf{D}_{-j})[2 D_j - 1]\Big) + \\
& \qquad \mathcal{O}\left(\frac{\bar{\phi}_n}{\|w\|_1^2}\cdot \sum_{i,j: c(i)\neq c(j)}w_i w_jI(j\in B_i\cup G_i)\right)\qquad \text{(by Lemma \ref{lem:2} and that $b_n(\mathcal{C}_n)\le 1$)}\\
& =  \frac{4}{\|w\|_1^2}  \sum_{i,j: c(i) = c(j)} \frac{w_i w_j}{4}(\mu_i(\mathbf{1}) + \mu_i(\mathbf{0}))(\mu_j(\mathbf{1}) + \mu_j(\mathbf{0})) + \mathcal{O}\left(\frac{\bar{\phi}_n}{\|w\|_1^2}\cdot \sum_{i,j: c(i)= c(j)}w_i w_j\right) + \\
& \qquad \mathcal{O}\left(\frac{\bar{\phi}_n }{\|w\|_1^2}\cdot \sum_{i,j: c(i)\neq c(j)}w_i w_jI(j\in B_i\cup G_i)\right)\qquad \text{(by Lemma \ref{lem:rel})}\\
& = \frac{4}{\|w\|_1^2} \sum_{i,j: c(i) = c(j)} \frac{w_i w_j}{4}(\mu_i(\mathbf{1}) + \mu_i(\mathbf{0}))(\mu_j(\mathbf{1}) + \mu_j(\mathbf{0}))\\
& \qquad \mathcal{O}\left(\frac{\bar{\phi}_n }{\|w\|_1^2}\cdot \sum_{i,j}w_i w_j I(j\in B_i\cup G_i)\right)\qquad (\text{because }c(i) = c(j) \Longrightarrow j\in B_i\cup G_i)\\
& = \frac{1}{\|w\|_1^2}\sum_{k=1}^{K_n}\left(\sum_{i: c(i)  = k}w_i(\mu_i(\mathbf{1}) + \mu_i(\mathbf{0}))\right)^2 +  \mathcal{O}\left(\frac{\bar{\phi}_n }{\|w\|_1^2}\cdot \sum_{i,j}w_i w_j I(j\in B_i\cup G_i)\right)\\
& = \frac{1}{\|w\|_1^2}\sum_{k=1}^{K_n}\left(\sum_{i: c(i)  = k}w_i(\mu_i(\mathbf{1}) + \mu_i(\mathbf{0}))\right)^2 +  \mathcal{O}\left(\frac{\bar{\phi}_n }{n^2}\cdot \sum_{i,j}I(j\in B_i\cup G_i)\right)\quad (\text{because }w_i \le \bar{w} \text{ and }\|w\|_1 \ge \underline{w}n)\\
& = \frac{1}{\|w\|_1^2}\sum_{k=1}^{K_n}\left(\sum_{i: c(i)  = k}w_i(\mu_i(\mathbf{1}) + \mu_i(\mathbf{0}))\right)^2 +  \mathcal{O}\left(\frac{\bar{\phi}_n }{n^2}\cdot \sum_{i}|B_i\cup G_i|\right).
\end{aligned}
$$
It has been shown in the proof of Lemma \ref{lem:worst_case_bias}, that 
\[\frac{\bar{\phi}_n }{n^2}\cdot \sum_{i}|B_i\cup G_i| = \mathcal{O}(\varepsilon/K_n).\]
Thus, 
\begin{equation}\label{eq:var_trimming_term1}
\mathbb{E}_\mu\Big[\Big(\hat{\tau}_n^{w}(\mathcal{C}_{ n}) - \mathbb{E}_\mu[\hat{\tau}_n^{w}(\mathcal{C}_{ n})]\Big)^2\Big] = \frac{1}{\|w\|_1^2}\sum_{k=1}^{K_n}\left(\sum_{i: c(i)  = k}w_i(\mu_i(\mathbf{1}) + \mu_i(\mathbf{0}))\right)^2 + \mathcal{O}(\varepsilon/K_n).
\end{equation}
By Assumption \ref{ass:exposure_restriction} (i) and the definition \eqref{eq:psibar} of $\bar{\psi}$, we have
\[\sup_{\mu\in \mathcal{M}}\mathbb{E}_\mu\Big[\Big(\hat{\tau}_n^{w}(\mathcal{C}_{ n}) - \mathbb{E}_\mu[\hat{\tau}_n^{w}(\mathcal{C}_{ n})]\Big)^2\Big]\le \frac{1}{\|w\|_1^2}\sum_{k=1}^{K_n}\left(\sum_{i: c(i)  = k}w_i\right)^2\cdot \bar{\psi} + \mathcal{O}(\varepsilon/K_n).\]
\paragraph{Lower bound (achievability) of worst-case variance} On the other hand, for potential outcomes defined in \eqref{eq:worst_case}, similar to \eqref{eq:worst_case_variance_last_step}, 
$$
\begin{aligned}
\mathbb{E}_\mu\Big[\Big(\hat{\tau}_n^{w}(\mathcal{C}_{ n}) - \mathbb{E}_\mu[\hat{\tau}_n^{w}(\mathcal{C}_{ n})]\Big)^2\Big]
& = \frac{1}{\|w\|_1^2}\sum_{k=1}^{K_n}\left(\sum_{i: c(i)  = k}w_i\right)^2\cdot (\bar{\psi}^{1/2} - \bar{\phi}_n)^2 + \mathcal{O}(\varepsilon/K_n)\\
& = \frac{1}{\|w\|_1^2}\sum_{k=1}^{K_n}\left(\sum_{i: c(i)  = k}w_i\right)^2\cdot \bar{\psi} (1 + o(1)) + \mathcal{O}(\varepsilon/K_n).
\end{aligned}
$$
By Cauchy-Schwarz inequality, 
\begin{equation}\label{eq:Cauchy_Schwarz_trimming}
\sum_{k=1}^{K_n}\left(\sum_{i: c(i)  = k}w_i\right)^2\ge \frac{1}{K_n}\left(\sum_{k=1}^{K_n}\sum_{i: c(i)  = k}w_i\right)^2 = \frac{\|w\|_1^2}{K_n}.
\end{equation}
Therefore, the proof of worst-case variance is completed.

\subsection{Proof of Theorem \ref{thm:worst_case_mse_trimming}}\label{app:proof:thm:worst_case_mse_trimming}
We prove the following more general result:
$$
\small 
\begin{aligned}
\sup_{\mu \in \mathcal{M}} \mathbb{E}_\mu\Big[\Big(\hat{\tau}_n^{w}(\mathcal{C}_{ n}) - \mathbb{E}_\mu[\hat{\tau}_n^{w}(\mathcal{C}_{ n})]\Big)^2\Big] + \lambda  \Big(\tau_{n,\mu} - \mathbb{E}_\mu[\hat{\tau}_n^{w}(\mathcal{C}_{ n})]\Big)^2
 \quad \in [1/4, 1]\cdot \mathcal{B}_n^*(\mathcal{C}_n, w, \lambda)(1 + \mathcal{O}(\varepsilon)).
\end{aligned}
$$
\paragraph{Upper bound} The upper bound is a direct consequence of Theorem \ref{thm:weighted_estimator}. 

\paragraph{Lower bound} To prove the lower bound, assume without loss of generality that
\(|\psi_R|\ge |\psi_L|\); otherwise replace all potential outcomes by their
negatives. Then \(\psi_R>0\) and, by the definition of \(\bar\psi\),
$
\bar\psi^{1/2}=2\psi_R.
$ 
Consider \(\mu\) satisfying \eqref{eq:worst_case_trimming_bias} and taking $n$ large enough so that $\bar{\phi}_n \le \frac{\psi_R - \psi_L}{4}$ and 
\[
(\mu_i(\mathbf 1),\mu_i(\mathbf 0))
=
\begin{cases}
\left((\psi_R+\psi_L)/2-\bar\phi_n,\ \psi_R-\bar\phi_n\right),
& \text{if } w_i/\|w\|_1\ge 1/n,\\
\left(\psi_R,\ (\psi_R+\psi_L)/2\right),
& \text{otherwise.}
\end{cases}
\]
For this construction, the bias term is at least one half of the bias component
appearing in \(\mathcal B_n^*(\mathcal C_n,w,\lambda)\). Moreover,
$ 
\mu_i(\mathbf 1)+\mu_i(\mathbf 0)
\ge
\psi_R-2\bar\phi_n
=
\frac{\bar\psi^{1/2}}{2}-2\bar\phi_n.
$ 
Hence, by \eqref{eq:var_trimming_term1},
$ 
\mathbb E_\mu\left[
\left(
\hat\tau_n^w(\mathcal C_n)
-
\mathbb E_\mu[\hat\tau_n^w(\mathcal C_n)]
\right)^2
\right]
\ge
\frac{1}{4}\bar\psi
\frac{1}{\|w\|_1^2}
\sum_{k=1}^{K_n}
\left(
\sum_{i:c(i)=k}w_i
\right)^2
(1+\mathcal O(\varepsilon)).
$ 
Putting together  the final expression, we have
$$
\begin{aligned}
&\sup_{\mu \in \mathcal{M}} \mathbb{E}_\mu\Big[\Big(\hat{\tau}_n^{w}(\mathcal{C}_{ n}) - \mathbb{E}_\mu[\hat{\tau}_n^{w}(\mathcal{C}_{ n})]\Big)^2\Big] + \lambda  \Big(\tau_{n,\mu} - \mathbb{E}_\mu[\hat{\tau}_n^{w}(\mathcal{C}_{ n})]\Big)^2\\
& \ge \frac{1}{4}\mathcal{B}_n^{*}(\mathcal{C}_n, w, \lambda)(1 + \mathcal{O}(\varepsilon)).
\end{aligned}
$$

\begin{rem}\label{rem:exact_worst_case_mse_trimming}
From the above proof, it is not hard to show that the worst-case MSE, up to $1+\mathcal{O}(\varepsilon)$ terms, can be formulated as 
$$
\begin{aligned}
\max_{\mu_i(\mathbf{1}), \mu_i(\mathbf{0})\in [\psi_L, \psi_R]} & \frac{1}{\|w\|_1^2}\sum_{k=1}^{K_n}\left(\sum_{i: c(i)  = k}w_i(\mu_i(\mathbf{1}) + \mu_i(\mathbf{0}))\right)^2 \\
& + \lambda\left(\bar{\phi}_n b_n^w(\mathcal{C}_n) + \sum_{i=1}^{n}(\mu_i(\mathbf{1}) - \mu_i(\mathbf{0}))\lb \frac{w_i}{\|w\|_1} - \frac{1}{n}\rb\right)^2
\end{aligned}
$$ \qed 
\end{rem}

\subsection{Proof of Theorem \ref{thm:saturation}} 
\label{app:thm:saturation} 

\paragraph{Upper bound on the worst-case bias.} Consider first the worst-case bias. Following the argument in the proof of Lemma \ref{lem:worst_case_bias}, 
$$
\begin{aligned} 
\sup_{\mu \in \mathcal{M}} \Big| \mathbb{E}_\mu[\hat{\tau}_n] - \tau_{n,\mu} \Big| & \le \sup_{\mu(1, \cdot) \in \otimes_{i=1}^n \mathcal{M}_{1, i} } \Big| \frac{1}{n} \sum_{i=1}^n \mathbb{E}\Big[ \mu_i(D_i, \mathbf{D}_{-i}) - \mu_i(\mathbf{1}) | D_i = 1  \Big] \Big| \\ 
& + \sup_{\mu(0, \cdot) \in \otimes_{i=1}^n \mathcal{M}_{0, i} } \Big| \frac{1}{n} \sum_{i=1}^n \mathbb{E}\Big[ \mu_i(D_i, \mathbf{D}_{-i}) - \mu_i(\mathbf{0}) | D_i = 0  \Big] \Big|. 
\end{aligned} 
$$
Following verbatim the proof of Lemma \ref{lem:worst_case_bias}, we can write 
$$
\small 
\begin{aligned} 
& \sup_{\mu(1, \cdot) \in \otimes_{i=1}^n \mathcal{M}_{1, i} } \Big| \frac{1}{n} \sum_{i=1}^n \mathbb{E}\Big[ \mu_i(D_i, \mathbf{D}_{-i}) - \mu_i(\mathbf{1}) | D_i = 1  \Big] \Big|  = \frac{\bar{\phi}_n}{n} \sum_{i=1}^n \frac{1}{|\mathcal{N}_i|} \mathbb{E}\Big[\sum_{k \in \mathcal{N}_i} (1 - D_k) \mid D_i = 1\Big] \\ 
&\le \frac{\bar{\phi}_n}{2 n} \sum_{i=1}^n \frac{1}{|\mathcal{N}_i|} \Big|j \in \mathcal{N}_i: c(i) \neq c(j)\Big| + \frac{\bar{\phi}_n}{n} \sum_{i=1}^n \frac{1}{|\mathcal{N}_i|} \Big(|\mathcal{N}_i| - \sum_{k \in \mathcal{N}_i} \mathbb{E}[D_k | D_i = 1] \Big)1\{c(i) = c(k)\}, 
\end{aligned} 
$$
where the first term of the second line uses the fact that $\mathbb{E}[(1 - D_k)\mid D_i = 1] = 1 - 
(q_1 + q_0) / 2 = 1/2$ if $c(k) \neq c(i)$. 
For $c(i) = c(k)$, we can write 
$$
\small 
\begin{aligned} 
\mathbb{E}[D_k | D_i = 1] & = \frac{P(D_k = 1, D_i = 1)}{P(D_i = 1)} = 2 P(D_k = 1, D_i = 1) = (q_1^2 + q_0^2). 
\end{aligned} 
$$
Similarly, we can write 
$$
\small 
\begin{aligned} 
& \sup_{\mu(0, \cdot) \in \otimes_{i=1}^n \mathcal{M}_{0, i} } \Big| \frac{1}{n} \sum_{i=1}^n \mathbb{E}\Big[ \mu_i(D_i, \mathbf{D}_{-i}) - \mu_i(\mathbf{0}) | D_i = 0  \Big] \Big| \\ 
&\le \frac{\bar{\phi}_n}{2 n} \sum_{i=1}^n \frac{1}{|\mathcal{N}_i|} \Big|j \in \mathcal{N}_i: c(i) \neq c(j) \Big| + \frac{\bar{\phi}_n}{n} \sum_{i=1}^n \frac{1}{|\mathcal{N}_i|} \Big(\sum_{k \in \mathcal{N}_i} \mathbb{E}[D_k | D_i = 0]1\{c(i) = c(k) \}\Big). 
\end{aligned} 
$$
For $c(i) = c(k)$, we can write 
$$
\mathbb{E}[D_k | D_i = 0] = \frac{P(D_k = 1, D_i = 0)}{P(D_i = 0)} = q_0(1 - q_0) + q_1(1 - q_1). 
$$
By collecting terms, we can write (since $q_1 + q_0 = 1$)  
$$
\small 
\begin{aligned} 
\sup_{\mu \in \mathcal{M}} \Big| \mathbb{E}_\mu[\hat{\tau}_n] - \tau_{n,\mu} \Big| &\le \frac{\bar{\phi}_n}{n} \sum_{i=1}^n \frac{1}{|\mathcal{N}_i|} \Big| j \in \mathcal{N}_i: c(i) \neq c(j) \Big|\\
& \quad + \frac{\bar{\phi}_n}{n} \sum_{i=1}^n \frac{1}{|\mathcal{N}_i|} 2(1 - q_1^2 - q_0^2) \Big|\mathcal{N}_i \bigcap \{j\in \mathcal{N}_i: c(i) = c(j)\}\Big|. 
\end{aligned} 
$$
Since $q_1 + q_0 = 1$, $1 - q_1^2 - q_0^2 = (q_1 + q_0)^2 - q_1^2 - q_0^2 = 2q_1q_0$. Recall the definition \eqref{eq:bn}, we can write the upper bound as 
\[\bar{\phi}_{n} \left\{b_n(\mathcal{C}_n) + 4q_1q_0(1 
- b_n(\mathcal{C}_n))\right\}.\]

\paragraph{Lower bound (achievability) of worst-case bias}
Following the proof of lower bound in Lemma \ref{lem:worst_case_bias}, we can show that the potential outcomes defined in \eqref{eq:worst_case} achieves the above upper bound. 

\paragraph{Upper bound on the worst-case variance} Consider now the variance component. Following verbatim the reasoning of Lemma \ref{prop:main}, we can write
\[\mathbb{E}\Big[ \Big(\hat{\tau}_n - \mathbb{E}[\hat{\tau}_n]\Big)^2 \Big] = \frac{4}{n^2}\sum_{i,j: c(i) = c(j)}\mathrm{Cov}\Big( D_i, D_j \Big)(\mu_i(\mathbf{1}) + \mu_i(\mathbf{0}))(\mu_j(\mathbf{1}) + \mu_j(\mathbf{0})) + \mathcal{O}(\varepsilon/K_n).\]
Thus,
$$
\sup_{\mu \in \mathcal{M}} \mathbb{E}\Big[ \Big(\hat{\tau}_n - \mathbb{E}[\hat{\tau}_n]\Big)^2 \Big] \le\frac{4}{n^2} \sum_{i,j: c(i) = c(j)} \mathrm{Cov}\Big( D_i, D_j \Big) \bar{\psi} + \mathcal{O}(\varepsilon/K_n).  
$$
Note that we can write for $c(i) = c(j), i \neq j$,
$$
\begin{aligned} 
\mathbb{E}[D_i D_j]- \mathbb{E}[D_i]\mathbb{E}[D_j] = \Big[\frac{q_1^2 + q_0^2}{2}  - \frac{1}{4}\Big] = \frac{1}{4} (2 q_1^2 + 2 q_0^2 - 1). 
\end{aligned} 
$$
Since $q_1 + q_0 = 1$, $2q_1^2 + 2q_0^2 - 1 = 2q_1^2 + 2q_0^2 - (q_1 + q_0)^2 = (q_1 - q_0)^2$.
 For $i = j$, summing over $n$ terms the variance contribution instead is $1/4 = (q_1 - q_0)^2/4 + q_0 q_1$ since $q_1 + q_0 =1$. Collecting the terms we obtain the desired bound 
\paragraph{Lower bound (achievability) of worst-case variance}
Following the proof of lower bound in Lemma \ref{prop:main}, we can show that the potential outcomes defined in \eqref{eq:worst_case} also achieves the above upper bound. 

\paragraph{Worst-case MSE} Since \eqref{eq:worst_case} achieves both the worst-case bias and variance, the worst-case MSE is simply their sum.

\subsection{Proof of Lemma \ref{lem:endogenous_peer_effects_weak_correlation}}
\label{proof:lem:endogenous}

Let
$
\mathbf B_n=(\mathbf I-\gamma_n\mathbf L)^{-1}.
$ 
Because \(\mathbf A\) is symmetric and the graph is connected,
$ 
\left\|
\mathbf V^{-1/2}\mathbf A\mathbf V^{-1/2}
\right\|_{\mathrm{op}}\le 1.
$ 
Therefore,
$ 
\|\mathbf L\|_{\mathrm{op}}
=
\left\|
\mathbf V^{-1/2}
(\mathbf V^{-1/2}\mathbf A\mathbf V^{-1/2})
\mathbf V^{1/2}
\right\|_{\mathrm{op}}
\le
\|\mathbf V^{-1/2}\|_{\mathrm{op}}\|\mathbf V^{1/2}\|_{\mathrm{op}}
\le
\sqrt{\mathcal N_{n,\max}}.
$ 
The assumption \(\gamma_n\sqrt{\mathcal N_{n,\max}}=o(1)\) therefore implies
$ 
\|\gamma_n\mathbf L\|_{\mathrm{op}}=o(1).
$ 
Hence, for large \(n\), \(\|\gamma_n\mathbf L\|_{\mathrm{op}}<1\), and the
Neumann series gives
$ 
\mathbf B_n-\mathbf I
=
(\mathbf I-\gamma_n\mathbf L)^{-1}-\mathbf I
=
\sum_{m=1}^{\infty}(\gamma_n\mathbf L)^m.
$ 
Thus,
$ 
\|\mathbf B_n-\mathbf I\|_{\mathrm{op}}
\le
\frac{\|\gamma_n\mathbf L\|_{\mathrm{op}}}
{1-\|\gamma_n\mathbf L\|_{\mathrm{op}}}
=
o(1).
$ 
Consequently,
$ 
\|\mathbf B_n\mathbf B_n^\top-\mathbf I\|_{\mathrm{op}}=o(1).
$ 
Since
$ 
\mathbf D_n=\operatorname{diag}(\mathbf B_n\mathbf B_n^\top),
$ 
we also have
$ 
\|\mathbf D_n-\mathbf I\|_{\mathrm{op}}=o(1),
$ 
and hence
$ 
\left\|
\mathbf D_n^{-1/2}
(\mathbf B_n\mathbf B_n^\top)
\mathbf D_n^{-1/2}
-
\mathbf I
\right\|_{\mathrm{op}}
=o(1).
$ 
Therefore \(\|\mathbf C-\mathbf I\|_{\mathrm{op}}=o(1)\).

\end{document}